%% file: main.tex
\documentclass[3p]{elsarticle}
\makeatletter
\def\ps@pprintTitle{%
 \let\@oddhead\@empty
 \let\@evenhead\@empty
 \def\@oddfoot{\centerline{\thepage}}%
 \let\@evenfoot\@oddfoot}
\makeatother
\usepackage{tabularx} 
\usepackage{amsmath}  
\usepackage{amssymb}
\usepackage{amsthm}
\usepackage{bm}
\usepackage{graphicx} 
\usepackage[final]{hyperref}
\usepackage[ruled,vlined]{algorithm2e}
\usepackage{enumitem}
\usepackage{array}
\usepackage{multirow}
\usepackage[title]{appendix}
\usepackage{empheq}
\usepackage{cleveref}
\usepackage{cases}
\usepackage{textcomp}
\usepackage{gensymb}
\usepackage{subcaption}
\hypersetup{
	colorlinks=true,       
	linkcolor=blue,        
	citecolor=blue,        
	filecolor=magenta,     
	urlcolor=blue         
}
\newcommand{\vect}[1]{\boldsymbol{\mathbf{#1}}}

\newcommand{\R}{\mathbb{R}}
\newcommand{\N}{\mathbb{N}}

\newcommand{\norm}[1]{\left\lVert #1 \right\rVert}

\newcommand{\lap}{\Delta}
\newcommand{\grad}{\nabla}

\newtheorem{theorem}{Theorem}[section]
\newtheorem{remark}{Remark}[section]
\newtheorem{lemma}{Lemma}[section]

\newtheorem{proposition}{Proposition}[section]

\DeclareMathOperator*{\argmin}{arg\,min}

\newcolumntype{C}{ >{\centering\arraybackslash} m{4cm} }

\makeatletter
\def\BState{\State\hskip-\ALG@thistlm}
\makeatother
\begin{document}
\begin{frontmatter}
\title{A Globally Convergent Modified Newton Method for the Direct Minimization of the Ohta-Kawasaki Energy with Application to the Directed Self-Assembly of Diblock Copolymers}

\author{Lianghao Cao*}
\author{Omar Ghattas}
\author{J. Tinsley Oden}

\address{Oden Institute for Computational Engineering \& Sciences, The University of Texas at Austin, Austin, TX, 78712, USA}
\cortext[cor1]{Email address: \texttt{lianghao@ices.utexas.edu}}

\begin{abstract}
We propose a fast and robust scheme for the direct minimization of the Ohta–Kawasaki energy that characterizes the microphase separation of diblock copolymer melts. The scheme employs a globally convergent modified Newton method with line search which is shown to be mass-conservative, energy-descending, asymptotically quadratically convergent, and typically three orders of magnitude more efficient than the commonly-used gradient flow approach. The regularity and the first-order condition of minimizers are carefully analyzed. A numerical study of the chemical substrate guided directed self-assembly of diblock copolymer melts, based on a novel polymer-substrate interaction model and the proposed scheme, is provided. 
\end{abstract}

\end{frontmatter}

\input{sec_01_introduction.tex}
\input{sec_02_problem_statement.tex}
\input{sec_03_the_first-order_condition_in_h-1.tex}

\input{sec_04_the_newton_raphson_iterations.tex}
\input{sec_05_numerical_results.tex}
\input{sec_06_chemical_subsrate_guided_dsa_of_bcp}

\input{sec_07_conclusion.tex}
\input{acknowledgement.tex}
\input{appendix.tex}

\addcontentsline{toc}{section}{References}                                                                     
\normalsize                                                                                                    

\bibliographystyle{model5-names}
\biboptions{sort,numbers,comma,compress}
\bibliography{main}

\end{document}

%% file: sec_01_introduction.tex
\section{Introduction}
A \textit{block copolymer} (BCP) is a polymer consisting of sub-chains or blocks of chemically distinct monomers joined by covalent bonds, each block being a linear series of identical monomers. A large collection of one type of block copolymer is called a melt. At high temperatures, the blocks in an incompressible melt are mixed homogeneously. As the temperature is reduced, the dislike blocks tend to segregate and lead to a process termed \textit{microphase separation}. The microphase separation of BCP melts results in the self-assembly of meso-scale multi-phase ordered structures, such as lamellae, spheres, cylinders, and gyroids \cite{Bates1999, Hamley1999,Abetz2005}. The microphase separation could be further guided by chemically and/or topologically patterned templates formed on the underlying surface, enabling the design of complex nano-structures. This process is referred to as the \textit{directed self-assembly} (DSA) of BCPs. The design of the DSA of BCPs to reproduce nano-structures with desired features is highly attractive in nano-manufacturing applications \cite{Bates2014, Stoykovich2007, Ji2016, Park2003}.\\
\indent Computational studies of the DSA of BCPs have proven to be valuable in determining the effects of material properties, film thickness, polymer-substrate interactions, and geometric confinement on the self-assembly process \cite{Ginzburg2015, Liu2013, Yi2013, Yang2015}. Continuum models of microphase separation of BCP melts \cite{Muller2018}, such as the self-consistent field theory (SCFT) model, the Ohta-Kawasaki (OK) model, and the Swift–Hohenberg model, make possible the exploration of the space of nano-structures formed by the DSA process with a relatively low computational cost. They are often used in design and inverse problems associated with the DSA of BCPs \cite{Hannon2013,Qin2013,Latypov2014, Hannon2014, Gadelrab2017, Hannon2018, Mukherjee2016}. To further reduce the computational cost, it is essential to develop fast and robust algorithms for obtaining model solutions, particularly since the model must be repetitively solved in the course of solving design and inverse problems.\\
\indent In this work, we focus on the OK model for microphase separation of diblock copolymer (BCP with two blocks) melts. The OK model, first introduced in \cite{Ohta1986}, is a density functional theory derived from the SCFT model. Its connection to particle-based models of polymers is discussed in \cite{Ohta1986,Choksi2003,Fredrickson2005,Muller2018}. The OK model leads to a free energy functional in terms of the relative local densities of the monomers that characterize the phase separation. The minimizers of the free energy subject to incompressibility constraints represent the possible equilibrium morphologies of the microphase separation. Various studies have demonstrated success in using the OK model to depict periodic structures of the diblock copolymer microphase separation that agree with those observed in experiments and fine-scale models \cite{Choksi2009, Choksi2011, Choksi2012, van_den_Berg_2017}.\\
\indent The OK model also belongs to a class of nonlocal Cahn-Hilliard models and minimizing the OK free energy through the nonlocal Cahn-Hilliard equation, i.e., the mass-conservative gradient flow of the OK energy, is commonly-adopted in the literature \cite{Zhang2006, Parsons2012,Qin2013,Choksi2009, Choksi2011,Choksi2012}. The nonlocal Cahn-Hilliard equation is a time-dependent fourth-order nonlinear PDE, known for its stiffness. To solve for the minimizers of the OK free energy through the nonlocal Cahn-Hilliard equation, one needs to solve for the steady-state solutions while retaining energy stability and accuracy. This typically requires a large number of small time steps and results in slow linear convergence. Despite this downside, the nonlocal Cahn-Hilliard equation is used in many computationally challenging studies on the DSA of BCPs in the literature, such as the numerical investigation of the equilibrium morphology phase diagram \cite{Choksi2009, Choksi2011, Choksi2012} and the design of optimal DSA guiding patterns \cite{Qin2013}. The gradient flow approach is also referred to as the \textit{continuous steepest descent method} in the context of the SCFT model. As pointed out in \cite[p.234]{Fredrickson2005}, since the SCFT model is an equilibrium model and only the steady states are of physical interest, one ultimately does not care much about the accuracy of the PDE solve, as the goal is to obtain the steady-state solutions as quickly as possible. The same argument applies to the OK model, which is posed as an energy minimization problem similar to the SCFT model. Therefore, it is natural that we seek to develop numerical schemes that preserve the advantages of the gradient flow approach, such as mass conservation, global convergence, and energy stability, while discarding its slow and fictitious dynamical trajectory.\\
\indent In this paper, we propose a fast and robust algorithm for the direct minimization of the OK energy functional that greatly reduces the computational cost of using the OK model. A globally convergent modified Newton method with line search, defined in an appropriate function space setting, is used for the mass-conservative minimization of the OK energy. The minimization iterations are shown to monotonically decrease the OK energy, attain quadratic convergence, and find local minimizers three order of magnitude faster than the traditional gradient flow approach. \\
\indent The rest of the paper is organized as followings. In Section 2, the OK energy functional is introduced and relevant Hilbert spaces are defined. The minimization problem of the OK energy functional is posed in these function spaces. In Section 3, the first order optimality condition is reformulated by transforming the test space. The regularity of the minimizers is analyzed for establishing the two-way equivalency of the transformation. In Section 4, a mass-conservative Newton method for minimizing the OK energy is proposed and analyzed. A modified Hessian operator is introduced and we show that it can be employed to generate energy-descending Newton steps, thus leading to global convergence. The choice of appropriate initial guesses is discussed. In Section 5, numerical examples are presented to demonstrate the properties of the proposed scheme. In Section 6, we provide a numerical study of the chemical substrate guided DSA of BCPs with the proposed scheme. A novel polymer-substrate interaction model based on the OK model is introduced. The conclusion is given in Section 7.

%% file: sec_02_problem_statement.tex
\section{Problem Statement}
\input{sec_02/subsec_01_OK_energy_functional.tex}

\input{sec_02/subsec_02_laplacian_inverse_operator.tex}

\input{sec_02/subsec_03_the_minimization_problem_in_h1.tex}

%% file: sec_02/subsec_01_ok_energy_functional.tex
\subsection{The OK energy functional}
Let $\Omega\subset \R^d$, $d=1,2,3$, be a bounded domain. Let $u_A$,$u_B:\Omega\to[0,1]$ be the normalized segment densities of monomer A and B inside the domain. They satisfy the incompressibility constraint, $u_A+u_B = 1$. Define an order parameter $u\coloneqq u_A-u_B\in[-1,1]$ that represents the normalized local density difference between the two phases. We consider the following form of the OK energy functional:
\begin{equation}\label{def:ok_energy_strong}
    F_{OK}(u) \coloneqq \frac{1}{2}\int_\Omega \Big\{ 2\kappa W(u) + \epsilon^2 |\grad u|^2+\sigma (u-m)(-\lap^{-1})(u-m)\Big\}\;d\vect{x}\;.
\end{equation}
where $\kappa,\epsilon,\sigma\in\R^+$ are scalar parameters. The constant $m$ is the spatial or mass average of the order parameter, i.e., $m = |\Omega|^{-1}\int_\Omega u\;d\vect{x}\in[-1,1]$. This equality is typically imposed as the incompressibility constraint on the order parameter for a given $m$.\\
\indent $W$ is a double well potential function that reaches minimum at $\pm 1$. It originally had the following form:
\begin{equation}
    W(u) = 
    \begin{cases}
    \frac{1}{4}(1-u^2)\;,&u\in[-1,1]\;;\\
    \infty\;,&\textrm{otherwise}\;.
    \end{cases}
\end{equation}
The infinite walls at $u=\pm 1$ allow the range of $u$ to be extended to the real axis. The approximate fourth-order polynomial form of the double well potential with global minimum at $u=\pm 1$ is often considered as an alternative:
\begin{equation}\label{eq:dw}
    W(u) = 
    \frac{1}{4}(1-u^2)^2\;.
\end{equation}
The action of the Laplacian inverse in \eqref{def:ok_energy_strong} is defined by the solution of the Poisson equation with the homogeneous Neumann boundary condition:
\begin{equation}\label{eq:lap_inv}
    w = (-\lap^{-1})(u-m)\iff \begin{cases}
        -\lap w = u-m &\text{in }\Omega\;;\\
        \grad w \cdot \nu = 0 &\text{on }\partial\Omega\;,
    \end{cases}
\end{equation}
where $\nu$ is the unit vector normal to $\partial\Omega$.\\
\indent The $\kappa$ and $\epsilon$ terms in \eqref{def:ok_energy_strong} define the Ginzburg-Landau energy and are often associated with the Cahn-Hilliard and Allen-Cahn models of phase change. The $\sigma$ term is responsible for the long-range interaction and is usually expressed as follows:
\begin{equation}\label{eq:nonlocal_energy}
    F_\sigma = \frac{\sigma}{2}\int_\Omega\int_\Omega \big(u(\vect{x})-m\big)G(\vect{x},\vect{y})\big(u(\vect{y})-m\big)\;d\vect{x}\;d\vect{y}\;,
\end{equation}
where $G$ is the Green's function of the Laplacian operator. $F_\sigma$ is therefore often referred to as the nonlocal energy.\\
\indent The scalar parameters in the OK energy are linked to the material parameters of the diblock copolymer. In particular, an appeal to self-consistent field theory in \cite{Choksi2003} leads to the relations
\begin{equation}\label{eq:model_param}
    \kappa = 1\;,\quad \epsilon^2 = \frac{l^2}{3f(1-f)\chi|D|^{2/d}}\;,\quad\sigma= \frac{36|D|^{2/d}}{f^2(1-f)^2l^2\chi N^2}\;,
\end{equation}
where $N$ is the degree of polymerization of the polymer, $\chi$ is the Flory-Huggins parameter that indicates the strength of repulsion between dissimilar monomers, $l$ is the statistical segment length of the polymer, and $f\in(0,1)$ is the segment ratio of one of the sub-chains in the polymer. The domain $D$ refers to the physical domain, whereas $\Omega$ in the energy is originally defined as a normalized domain of unit volume. To maintain generality of the model, we treat $\kappa$ as an arbitrary positive number and $\Omega$ as an arbitrary bounded domain with sufficiently regular boundary. The regularity of the boundary is discussed in Section \ref{subsec:regularity}.

%% file: sec_02/subsec_02_laplacian_inverse_operator.tex
\subsection{The Laplacian inverse operator and the Hilbert space \texorpdfstring{$\mathring{H}^{-1}$}{}}
To arrive at a mathematically rigorous problem formulation, we define the Laplacian inverse operator \eqref{eq:lap_inv} in a weak sense. First, we need to define some function spaces.\\
\indent Consider decomposing $H^1\coloneqq H^1(\Omega)$ into $\mathring{H^1}\oplus \mathcal{C}$,
where $\mathring{H}^1$ is a subspace of $H^1$ consisting of functions with zero spatial average and $\mathcal{C}$ consists of constant-almost-everywhere functions:
\begin{equation}
    \mathring{H}^1 \coloneqq \Big\{v\in H^1 : \int_\Omega v\;d\vect{x} = 0\Big\}\;,\quad\mathcal{C} \coloneqq \Big\{v\in H^1:v = \text{const }a.e.\Big\}\;.
\end{equation}
In particular, $\mathring{H}^1$ is a Hilbert space equipped with the inner product $(u,v)_{\mathring{H}^1} = (\grad u,\grad v)$, where $(\cdot,\cdot)$ denotes the $L^2(\Omega)$ inner product. \\
\indent The dual space $(\mathring{H}^1)^*$ can be naturally extended to a subspace of $(H^{1})^*$, denoted $\mathring{H}^{-1}$, by taking $\mathcal{C}$ as the kernel,
\begin{equation}
    \mathring{H}^{-1} \coloneqq \big\{f\in (H^1)^* \:\big|\: \langle f, c\rangle=0\; \forall c\in \mathcal{C}\big\}\;,
\end{equation}
where $\langle\cdot,\cdot\rangle$ is the duality pairing of $(H^1)^*$ and $H^1$.\\
\indent Now we define a Laplacian inverse operator $(-\lap_N)^{-1}:\mathring{H}^{-1}\to \mathring{H}^1$ that weakly solves the Poisson equation with the homogeneous Neumann boundary condition:
\begin{equation}\label{def:lap_inv_weak}
    \Big(\grad \big((-\lap_N)^{-1}f\big),\grad \tilde{v}\Big) = \langle f, \tilde{v}\rangle\quad\forall \tilde{v}\in H^1,f\in \mathring{H}^{-1}\;.
\end{equation}
The Laplacian inverse operator is linear, bounded and bijective. As a result, we can define an inner product on $\mathring{H}^{-1}$ via the Laplacian inverse operator:
\begin{equation}\label{eq:h-1_inner_prod}
    (f, g)_{\mathring{H}^{-1}} \coloneqq \big((-\lap_N)^{-1}f, (-\lap_N)^{-1}g\big)_{\mathring{H}^1}=
    \begin{cases}
        \big\langle f, (-\lap_N)^{-1}g \big\rangle\;;\\
        \big\langle g, (-\lap_N)^{-1}f \big\rangle\;.
    \end{cases}
\end{equation}
$\mathring{H}^{-1}$ is complete under the inner product \cite[p.~32]{Cowan2005}, which makes $\big(\mathring{H}^{-1}, (\cdot,\cdot)_{\mathring{H}^{-1}}\big)$ a Hilbert space.\\
\indent Let $R$ be the $L^2$ Riesz map, $R:L^2\ni v \mapsto (v,\cdot)\in \big(L^2\big)^*$. The nonlocal energy \eqref{eq:nonlocal_energy} can be expressed weakly as:
\begin{equation}
    F_\sigma = \frac{\sigma}{2}\norm{R(u-m)}^2_{\mathring{H}^{-1}}\;.
\end{equation}


%% file: sec_02/subsec_03_the_minimization_problem_in_h1.tex
\subsection{The minimization problem in \texorpdfstring{$\mathring{H}^1$}{}}
Given $m\in(-1,1)$ and $\kappa,\epsilon,\sigma\in \R^+$, consider the weak form of the OK energy $F_{OK}:\mathring{H}^1_m\to\R^+$:
\begin{equation}\label{eq:ok_energy_weak}
    F_{OK}(u) =\kappa \norm{W(u)}_{L^1} + \frac{\epsilon^2}{2}\norm{u-m}^2_{\mathring{H}^1} + \frac{\sigma}{2} \norm{R(u-m)}^2_{\mathring{H}^{-1}} \;,
\end{equation}
where $\mathring{H}^1_m\coloneqq\big\{v\in H^1: v-m\in\mathring{H}^1\big\}$ and $W:H^1\to L^1$ with $W(u) = (1-u^2)^2/4$.\\
\indent The minimization problem is posed as follows:
\begin{equation}\label{def:min_prob_h1}
    \begin{aligned}
    \text{Find } u_{\text{min}}\in \mathring{H}^1_m \text{ such that :} && u_{\text{min}}=\argmin_{u\in\mathring{H}^1_m} F_{\text{OK}}(u)\;.    \end{aligned}
\end{equation}
The first and second order conditions for the minimizers are
\begin{equation}\label{eq:min_prob_var_h1}
u_{\text{min}}\in \mathring{H}^1_m \text{ such that :}\quad
    \begin{cases}
        \big\langle DF(u_{\text{min}}), \tilde{u}_0\big\rangle=0 &\forall \tilde{u}_0\in \mathring{H}^1\;;\\
        \exists \alpha > 0:\big\langle D^2F(u_{\text{min}})\tilde{u}_0, \tilde{u}_0\big\rangle \geq \alpha\norm{\grad \tilde{u}_0} &\forall \tilde{u}_0\in \mathring{H}^1\;,
    \end{cases}
\end{equation}
where $DF:\mathring{H}^1_m\to(\mathring{H}^1)^*$ is the G\^ateaux derivative of $F_{\text{OK}}$ at $u\in \mathring{H}^1_m$:
\begin{equation}\label{eq:1st_var}
    \big\langle DF(u),\tilde{u}_0\big\rangle \coloneqq \big(\kappa W'(u), \tilde{u}_0\big) + \big(\epsilon^2\grad u, \grad\tilde{u}_0\big) + \big(\sigma (-\lap_N)^{-1}R(u-m), \tilde{u}_0\big)\quad\forall \tilde{u}_0\in\mathring{H}^1\;,
\end{equation} 
and $D^2F(u):\mathring{H}^1\to (\mathring{H}^1)^*$ is a linear Hessian operator that corresponds to the second-order G\^ateaux derivative of $F_{\text{OK}}$ at $u\in\mathring{H}^1_m$:
\begin{equation}\label{eq:2nd_var}
    \big\langle D^2F(u)\hat{u}_0, \tilde{u}_0\big\rangle \coloneqq \big(\kappa W''(u)\hat{u}_0, \tilde{u}_0\big) + (\epsilon^2\grad \hat{u}_0, \grad\tilde{u}_0) + \big(\sigma (-\lap_N)^{-1}R(\hat{u}_0), \tilde{u}_0\big)\quad\forall \hat{u}_0,\tilde{u}_0\in\mathring{H}^1\;.
\end{equation}
\indent As the G\^ateaux derivatives $DF(u)$ and $D^2F(u)$ are defined on the subspace $\mathring{H}^1$ of $H^1$, this minimization problem is difficult to implement numerically without adaptation. One could adopt a constrained minimization approach and introduce a Lagrange multiplier into the energy functional in order to extend the test space to $H^1$. A first-order method that utilizes this approach is the mass-constrained Allen-Cahn equation. In the following section, we adopt a different approach. The problem is reformulated via the $\mathring{H}^{-1}$ inner-product such that the solution space and the test space of the transformed first-order condition are all in $H^1$. Such an approach is known to give rise to the local/nonlocal Cahn-Hilliard equation \cite{Fife2000}. We instead focus on the equivalency between the original formulation in $\mathring{H}^1$ and the transformed formulation in $\mathring{H}^{-1}$.

%% file: sec_03_the_first-order_condition_in_h-1.tex
\section{The First-Order Condition in \texorpdfstring{$\mathring{H}^{-1}$}{}}
 In this section, we first show regularity results for the critical points of the OK energy functional. Then we utilize the $\mathring{H}^{-1}$ inner-product to transform the first-order condition \eqref{eq:1st_var}. The two-way equivalency of the transformation is established along the way. 
 
 \input{sec_03/subsec_01_on_the_properties_of_the_minimizers.tex}

 \input{sec_03/subsec_02_the_h-1_first-order_condition.tex}

 \input{sec_03/subsec_03_the_mixed_h-1_first-order_condition.tex}

%% file: sec_03/subsec_01_on_the_properties_of_the_minimizers.tex
\subsection{On the regularity of the critical points of \texorpdfstring{$F_{OK}$}{}}\label{subsec:regularity}
Despite residing in $\mathring{H}^1_m$, the critical points of $F_{OK}$ that satisfy the first-order condition \eqref{eq:1st_var} usually possess bounded derivatives of order higher than one. Let $u_c\in\mathring{H}^1_m$ such that $DF(u_c)\equiv 0$; we have
\begin{equation}
    (\epsilon^2\grad u_c, \grad \tilde{u}_0) = \big(\kappa W'(u_c),\tilde{u}_0\big) + \big(\sigma(-\lap_N)^{-1}R(u_c-m),\tilde{u}_0\big)\quad\forall \tilde{u}_0\in\mathring{H}^1\;.
\end{equation}
Consequently, the first-order condition implies $u_c$ can be expressed as weak solutions to the pure Neumann Poisson problem, which admits the possibility of applying the Neumann elliptic regularity theorem on $u_c$. A general observation is that the highest order of the bounded derivatives of $u_c$ is dictated by the highest order permissible by the Neumann elliptic regularity theorem for the given domain $\Omega$ \cite{Cowan2005}.\\
\indent The following proposition is a reflection of this general observation. It shows that high regularity of the critical points can be achieved with high regularity of the domain boundary. The proof of the proposition is provided in Appendix \ref{appendix:a}. Here we define $H^k\coloneqq H^k(\Omega)$, and $\mathring{H}^k_m \coloneqq \mathring{H}^1_m\cap H^k$ with $k\geq 1$. 
\begin{proposition}[Regularity of the critical points]\label{prop:regularity}
    Let $u_c\in\mathring{H}^1_m$ and $DF(u_c)\equiv 0$. If $\Omega$ has a boundary of class $C^{r,1}$ \cite{Gilbarg1983} with $r\geq 1$ and satisfies the cone condition \cite{Adams2003}, then $u_{c}\in \mathring{H}^{r+1}_m$.
\end{proposition}
From now on, we proceed with the assumption that \textit{the domain boundary $\partial\Omega$ is sufficiently regular to imply at least $u_c\in\mathring{H}^3_m$ if $DF(u_c)\equiv 0$.} For the numerical examples in Section \ref{sec:numerical_results} and \ref{sec:chemical_substrate}, we consider rectangular domains in $2$d and box domains in $3$d, with which we have at least $u_c\in\mathring{H}^4_m$ if $DF(u_c)\equiv 0$ \cite{Grisvard2011,Hell2014}.\\
\indent The following proposition further establishes that the critical points and the Laplacian of the critical points satisfy the homogeneous Neumann boundary condition if the critical points have sufficiently high regularity. The proof is also provided in Appendix \ref{appendix:a}.
\begin{proposition}[Boundary conditions of the critical points]\label{prop:bc}
Let $u_c\in \mathring{H}^1_m$ and $DF(u_c)\equiv 0$.
\begin{enumerate}[label = (\roman*)]
    \item If $u_c\in\mathring{H}^2_m$, then $\grad u_c\cdot \nu=0$.
    \item If $u_c\in\mathring{H}^4_m$, then $\grad (\lap u_c) \cdot \nu = 0$ on $\partial\Omega$.
\end{enumerate}
\end{proposition}

%% file: sec_03/subsec_02_the_h-1_first-order_condition.tex
\subsection{The \texorpdfstring{$\mathring{H}^{-1}$}{} first-order condition}
Assume $u\in \mathring{H}^3_m$ and $\grad u\cdot\nu=0$ on $\partial\Omega$. We consider the equivalent representation of $DF(u)$ in $\mathring{H}^1$, in the sense that $\big\langle DF(u), \tilde{u}_0\big\rangle = \big(G(u), \tilde{u}_0\big)\;\forall \tilde{u}_0\in \mathring{H}^{1}$:
\begin{equation}\label{def:h1_df}
    G(u) = \kappa W'(u) - \epsilon^2\lap u + \sigma (-\lap_N)^{-1}R(u-m) - s_1\;,
\end{equation}
where $s_1 = (kW'(u) - \epsilon^2\lap u,1)$. \\
\indent With $R(\tilde{u}_0)\in \mathring{H}^{-1}$, define $\tilde{v}_0\in\mathring{H}^1$ such that $\tilde{v}_0 = (-\lap_N)^{-1}R(\tilde{u}_0)$. We invoke the definition of the $\mathring{H}^{-1}$ inner product \eqref{eq:h-1_inner_prod}:
\begin{subequations}\label{eq:1st_h-1_inner_prod}
\begin{gather}
    \big\langle DF(u), \tilde{u}_0\big\rangle = \big\langle R(\tilde{u}_0), G(u)\big\rangle = \big(G(u), \tilde{v}_0\big)_{\mathring{H}^{1}}\;,\\
    \big(G(u), \tilde{v}_0\big)_{\mathring{H}^{1}} = \big(\kappa \grad W'(u), \grad \tilde{v}_0\big) - \big(\epsilon^2\grad \lap u, \grad \tilde{v}_0\big) + \big(\sigma (u-m), \tilde{v}_0)\;.
\end{gather}
\end{subequations}
Let us define $A: H^3\to (H^1)^*$ such that 
\begin{equation}\label{def:h-1_1st_A}
\big\langle A(u), \tilde{u}\big\rangle \coloneqq \big(\kappa \grad W'(u), \grad \tilde{u}\big) - \big(\epsilon^2\grad \lap u, \grad \tilde{u}\big) + \big(\sigma (u-m), \tilde{u})\;.
\end{equation}
Notice that the operator $A$ identifies the subset of its domain $\mathring{H}^3_m$ with $\mathring{H}^{-1}$, as $\mathcal{C}$ is in the kernel of the gradient operator and $R(\mathring{H}^1)$. This leads to the following lemma.
\begin{lemma}\label{lemma:H-1_1st_cond}
    We have $A(u)\in\mathring{H}^{-1}$ if and only if $u\in \mathring{H}^3_m$.
\end{lemma}
The following theorem establishes a connection between the kernel of $DF$ and $A$ that gives rise to the $\mathring{H}^{-1}$ first-order condition, under which the solution space and the test space are both extended to $H^1$:
\begin{theorem}[The $\mathring{H}^{-1}$ first-order condition]\label{thm:h-1_1st}
We have $u_c\in \mathring{H}^1_m$ and $DF(u_c)\equiv 0$ if and only if $u_c\in H^3$ and $A(u_c)\equiv 0$.
\end{theorem}
\begin{proof}
    ($\Rightarrow$) Equation \eqref{def:h1_df} and \eqref{eq:1st_h-1_inner_prod} lead to $\big\langle A(u_c), \tilde{v}_0\big\rangle=0\;\forall \tilde{v}_0\in\big\{ v\in \mathring{H}^3:\grad v\cdot \nu = 0 \text{ on }\partial\Omega\big\}$. The test space can be trivially extended to $H^1$, as no high-order derivative or boundary term appears in $A$ and $\mathcal{C}$ is in the kernel of $\mathring{H}^{-1}$.\\ 
    \indent ($\Leftarrow$) If $u_c\in H^3$ and $A(u_c)\equiv 0$, then $u_c\in\mathring{H}^3_m$ by Lemma \ref{lemma:H-1_1st_cond}. If $A(u_c)$ is restricted to $(-\lap_N)^{-1}R(\mathring{H}^1)$, it is clear that $u_c$ satisfies $DF(u_c)\equiv 0$ by reversing the arguments that lead to \eqref{eq:1st_h-1_inner_prod}. 
\end{proof}
\begin{remark}
    The strong form of the $\mathring{H}^{-1}$ first order condition is 
    \begin{subequations}
    \begin{align}
        \lap(\kappa W'(u_c) - \epsilon^2\lap u_c) -\sigma(u_c-m) &= 0 && \text{in }\Omega\;;\\
        \grad u_c \cdot\nu = \grad(\lap u_c)\cdot\nu &= 0 &&\text{on }\partial\Omega\;.
    \end{align}
    \end{subequations}
    These equations characterize the steady state of the nonlocal Cahn-Hilliard equation. The strong form above is equivalent to the $\mathring{H}^1$ first-order condition only if there is enough regularity to claim $u_c\in\mathring{H}^4_m$ if $DF(u_c)\equiv 0$.
\end{remark}

%% file: sec_03/subsec_03_the_mixed_h-1_first-order_condition.tex
\subsection{The mixed \texorpdfstring{$\mathring{H}^{-1}$}{} first-order condition}
\indent We define the following operator $B:H^1 \times H^1\to (H^1)^*\times (H^1)^*$:
\begin{equation}\label{eq:h-1_1st_mixed}
    \big\langle B(u,\mu), (\tilde{u},\tilde{\mu})\big\rangle\coloneqq \Big((\grad \mu, \grad\tilde{u}) + (\sigma (u-m),\tilde{u}), (\mu, \tilde{\mu})-\big(\kappa W'(u), \tilde{\mu}\big) - (\epsilon^2\grad u, \grad\tilde{\mu})\Big)\;,\quad\forall\tilde{u},\tilde{\mu}\in H^1\;.
\end{equation}
The following theorem identifies the kernel of $DF$ with the the kernel of $B$.
\begin{theorem}[The mixed $\mathring{H}^{-1}$ first-order condition]\label{thm:mixed_h-1_1st}
We have $B(u_c, \mu_c) \equiv 0$ if and only if $u_c\in \mathring{H}^1_m$ and $DF(u_c)\equiv 0$.
\end{theorem}
\begin{proof}

     \indent ($\Rightarrow$) $u_c\in\mathring{H}^1_m$ is implied by taking $\tilde{u}\in\mathcal{C}$. Consider subtracting $s_2=(\kappa W'(u_c),1)$ from $\mu_c$ such that $\mu_c-s_2\in \mathring{H}^1$ and restricting $\tilde{u}$ to be in $(-\lap_N)^{-1}R(\mathring{H}^1)$. Then we invoke the definition of the $\mathring{H}^{-1}$ inner-product \eqref{eq:h-1_inner_prod} to arrive at
    \begin{equation*}
        \big(\mu_c - \sigma (-\lap_N)^{-1}(u_c-m), \tilde{u}_0\big) = 0\quad\forall \tilde{u}_0\in\mathring{H}^1\;.
    \end{equation*}
    Additionally, restricting $\tilde{\mu}$ to $\mathring{H}^1$ gives us
    \begin{equation*}
        (\mu_c,\tilde{u}_0)=(\kappa W'(u_c),\tilde{u}_0) + (\epsilon^2\grad u_c, \grad \tilde{u}_0)\quad\forall \tilde{u}_0 \in\mathring{H}^1\;.
    \end{equation*}
    Combining the two equations above, we have $DF(u_c)\equiv 0$.\\
    \indent ($\Leftarrow$) Theorem \ref{thm:h-1_1st} established the equivalence between $DF(u_c)\equiv 0$ and $A(u_c)\equiv 0$. $A(u_c)\equiv 0$ implies $B(u_c,\mu_c)\equiv 0$ by defining $\mu_c = \kappa W'(u_c) - \epsilon^2\lap u_c$.
\end{proof}

%% file: sec_04_the_newton_raphson_iterations.tex
\section{The Newton Iteration for Energy Minimization}
Given $u_0\in \mathring{H}_m^1$, consider the sequence of functions $\{u_n\}_{n=1}^\infty$ generated by a Newton method with line search globalization for minimizng the OK energy:
\begin{equation}\label{eq:newton_seq}
    u_{n+1} = u_n + t_n\delta u_n\;\quad \forall n\in\N\cup\{0\}\;,
\end{equation}
where the step length $t_n\in(0,1]$ and the Newton step $\delta u_n$ is given by the following variational problem in $\mathring{H}^1$:
\begin{align}\label{eq:newton_step_original}
    \text{Find } \delta u_n\in \mathring{H}^1 \text{ such that}: &&\big\langle D^2F(u_n)\delta u_n, \tilde{u}_0\big\rangle = - \big\langle DF(u_n), \tilde{u}_0 \big\rangle\quad\forall \tilde{u}_0\in\mathring{H}^1\;.
\end{align}
Now consider the $\mathring{H}^{-1}$ variational Newton step problem defined with $H^1$ test spaces:
\begin{subequations}\label{eq:h-1_newton_step_var}
\begin{align}
&\text{Find }(\delta u_n,\hat{\mu})\in H^1\times H^1 \text{ such that:}\nonumber\\
    &\qquad\qquad(\grad\hat{\mu}_n, \grad\tilde{u}_n) + (\sigma\delta u_n, \tilde{u}) = -(\grad \mu_n, \grad\tilde{u}) - \big(\sigma(u_n-m),\tilde{u})\big)&&\forall \tilde{u}\in H^1\;,\label{eq:h-1_newton_step_delta_u}\\
    &\qquad\qquad(\hat{\mu}_n, \tilde{\mu}) - \big(\kappa W''(u_n)\delta u_n, \tilde{\mu}\big) - \big(\epsilon^2 \grad (\delta u_n), \grad\tilde{\mu}\big) = 0  &&\forall\tilde{\mu}\in H^1\;.\label{eq:h-1_newton_step_mu_hat}
\end{align}
\end{subequations}
with $\mu_n\in H^1$ subject to
\begin{equation}\label{eq:mu_n_var}
    (\mu_n, \tilde{\mu}) = \big(\kappa W'(u_n), \tilde{\mu}\big) + (\epsilon^2\grad u_n, \grad\tilde{\mu})) \quad\forall\tilde{\mu}\in H^1\;.
\end{equation}
\indent The following proposition shows that, if $\delta u_n$ solves the $\mathring{H}^{-1}$ variational Newton step problem at $u_n\in\mathring{H}^1_m$, then it also solves the corresponding $\mathring{H}^1$ problem at $u_n$.
\begin{proposition}\label{prop:equiv_newton_step}
    Suppose $\delta u$ is a solution to the $\mathring{H}^{-1}$ Newton step problem (\eqref{eq:h-1_newton_step_var} and \eqref{eq:mu_n_var}) at $u_n\in H^1$. Then
    \begin{enumerate}[label = (\roman*)]
        \item $\delta u_n\in\mathring{H}^1$ if and only if $u_n\in\mathring{H}^1_m$.
        \item if $u_n\in \mathring{H}^1_m$, then $\delta u_n$ also solves the $\mathring{H}^1$ variational Newton step problem \eqref{eq:newton_step_original}.
    \end{enumerate}
\end{proposition}
\begin{proof}
    \begin{enumerate}[label = (\roman*)]
        \item Assume $\delta u\in\mathring{H}^{1}$, taking $\tilde{u} = c \in\mathcal{C}$ in \eqref{eq:h-1_newton_step_delta_u} gives $(\sigma (u_n-m),c) = 0$ and $u_n\in\mathring{H}^1_m$. The converse is true for the same argument.
        \item $u_n\in\mathring{H}^1_m$ implies $\delta u_n\in\mathring{H}^1$. The rest of the proof follows closely to that of Theorem \ref{thm:mixed_h-1_1st}. Define $s_3 = \big(kW'(u_n),1\big)$ and $s_4 = \big(kW''(u_n),1\big)$. Consider replacing $\mu$ and $\hat{\mu}$ with $\mu-s_3\in\mathring{H}^1$ and $\hat{\mu}-s_4\in\mathring{H}^1$ in \eqref{eq:h-1_newton_step_delta_u}. Restricting $\tilde{u}$ to $(-\lap_N)^{-1}R(\mathring{H}^1)$ in \eqref{eq:h-1_newton_step_delta_u} and invoking the definition of the $\mathring{H}^{-1}$ inner-product \eqref{eq:h-1_inner_prod}, we have
        \begin{equation*}
            \big(\hat{\mu}_n + \sigma(-\lap_N)^{-1}\delta u_n + \mu_n + \sigma (-\lap_N)^{-1}(u_n-m)), \tilde{u}_0\big) = 0\quad\forall \tilde{u}_0\in\mathring{H}^1\;. 
        \end{equation*}
        Restricting $\tilde{\mu}$ in \eqref{eq:h-1_newton_step_mu_hat} to $\mathring{H}^1$, we have $D^2F(u_n)\delta u_n - DF(u_n) \equiv 0$.
    \end{enumerate}
\end{proof}
The following lemma further establishes that the Newton iterates \eqref{eq:newton_seq} generated by the $\mathring{H}^{-1}$ Newton step problem are indeed in $\mathring{H}^{1}_m$ for $n\geq1$ if $t_0=1$, even when the initial guess $u_0\in H^1$ does not have the spatial average of $m$.
\begin{lemma}\label{lemma:projection}
Given $t_0 = 1$ and arbitrary $u_0,\mu_n, \hat{\mu}_n,\in H^1$, the $\mathring{H}^{-1}$ Newton iterates $\{u_n\}_{n=1}^\infty$ generated by \eqref{eq:newton_seq} and the $\mathring{H}^{-1}$ variational Newton steps \eqref{eq:h-1_newton_step_var} are in $\mathring{H}^1_m$.
\begin{proof}
    Let $\tilde{u} = c\in \mathcal{C}$ in \eqref{eq:h-1_newton_step_delta_u}, then it implies $(\delta u_n ,c) = -  (u_n-m,c)\;\forall n\in\N\cup\{0\}$. With $t_0 = 1$, \eqref{eq:newton_seq} implies
    \begin{align*}
        (u_1, c) = (u_0, c) + (\delta u_0, c)= (m,c) &\Rightarrow u_1\in \mathring{H}^{1}_m\\
        (t_1\delta u_1, c) = -\big(t_1(u_1-m), c\big) = 0 &\Rightarrow \delta u_1 \in \mathring{H}^{1} \text{ and } u_2 \in \mathring{H}^{1}_m\;.
    \end{align*}
    By induction, $u_n\in\mathring{H}^1_m\;\forall n\in \N$.
\end{proof}
\end{lemma}
The above proposition and lemma establish that the $\mathring{H}^{-1}$ Newton iteration is equivalent to the $\mathring{H}^1$ Newton iteration, if an appropriate initial guess is provided. Furthermore, the fact that $\hat{\mu}$ can be arbitrary allows us to modify the Hessian operator to achieve monotonic energy descent. In the following subsections, we propose a modified Hessian operator that generates energy-descending Newton iterations when implemented along with line search methods. This leads to global convergence of the Newton iteration. Then we discuss the choice of appropriate initial guess. A summary of the scheme is provided at the end of the section.
\begin{remark}
    When $\sigma = 0$, the solutions to the $\mathring{H}^{-1}$ Newton step problem \eqref{eq:h-1_newton_step_var} do not satisfy the $\mathring{H}^1$ Newton step problem \eqref{eq:newton_step_original}. In this case, one could take $\sigma$ to be very small and essentially interpret it as a stabilization term. Moreover, having $\sigma \sim 0$ implies that the physical domain is much smaller in scale than the statistical segment length of the polymer according to \eqref{eq:model_param}, and thus not of interest in the context of the phase separation of BCPs.
\end{remark}

\input{sec_04/subsec_01_generating_energy-stable_newton-raphson_iterations.tex}
\input{sec_04/subsec_02_bactracking_line_search.tex}
\input{sec_04/subsec_03_on_the_initial_guess_and_the_homogeneous_state.tex}

\input{sec_04/subsec_04_summary}

%% file: sec_04/subsec_01_generating_energy-stable_newton-raphson_iterations.tex
\subsection{Generating energy-descending Newton iteration}\label{subsubsec:bt_gamma}
To achieve global convergence of the Newton iteration with line search, we seek to produce search directions that are descent directions on the energy \cite{Nocedal2006}.  If the search direction $\delta u_n$ that solves the $\mathring{H}^1$ Newton step problem \eqref{eq:newton_step_original} at $u_n\in\mathring{H}^1_m$  also satisfies the following condition:
\begin{equation}\label{eq:descent_cond}
    \big\langle DF(u_n), \delta u_n \big\rangle < 0 \iff \big\langle D^2F(u_n)\delta u_n, \delta u_n\big\rangle > 0\;,
\end{equation}
then, with a sufficiently small $t_n\in(0,1]$, we have $F_{OK}(u_n + t_n\delta u_n)< F_{OK}(u_n)$. The appropriate $t_n$ can be found by line search methods, such as backtracking line search based on the Armijo condition. If the energy is reduced at each step, the iterations will eventually enter the ball of convergence of a critical point, leading to quadratic convergence with $t_n = 1$ and $D^2F(u_n)$ coercive. The energy-descending Newton iteration is global convergent if, additionally, the Hessian operator is Lipschitz continuous and bounded from above \cite[p.40]{Nocedal2006}. We note that the Hessian operator satisfies both of the conditions if the state is bounded point-wise, which is generally true due to the double well potential which has global minimum at $\pm 1$.\\
\indent Of course, the condition \eqref{eq:descent_cond} is not necessarily satisfied at any $u_n\in \mathring{H}^1_m$, as the Hessian operator $D^2F(u_n)$ may not be coercive due to the non-convex double well potential term. We thus seek to modify the Hessian operator by convexifying the double well potential.

\subsubsection{The modified Hessian operator: a Gauss-Newton type approximation}
\indent Consider the following decomposition of $W''(u)$:
\begin{equation}
    W''(u) = 2u^2 + (u^2-1)\;.
\end{equation}
The first term is always positive and the second term could be negative, especially since $u$ typically takes values between $\pm 1$. It is also typical that the minimizers contain large regions with $u$ close to $\pm 1$, which makes the second term relatively small. As a result, we choose an approximation to $W''(u)$ by scaling down this term with a weight $\gamma\in[1,0]$:
\begin{equation}\label{eq:weighted_dw}
    W_\gamma''(u) = 2u^2 + \gamma (u^2-1)\;,
\end{equation}
thus $\gamma$ allows interpolation between a Newton ($\gamma = 1$) and a Gaussian-Newton ($\gamma = 0$) method.
\begin{figure}
    \centering
    \includegraphics[width = 0.6\linewidth]{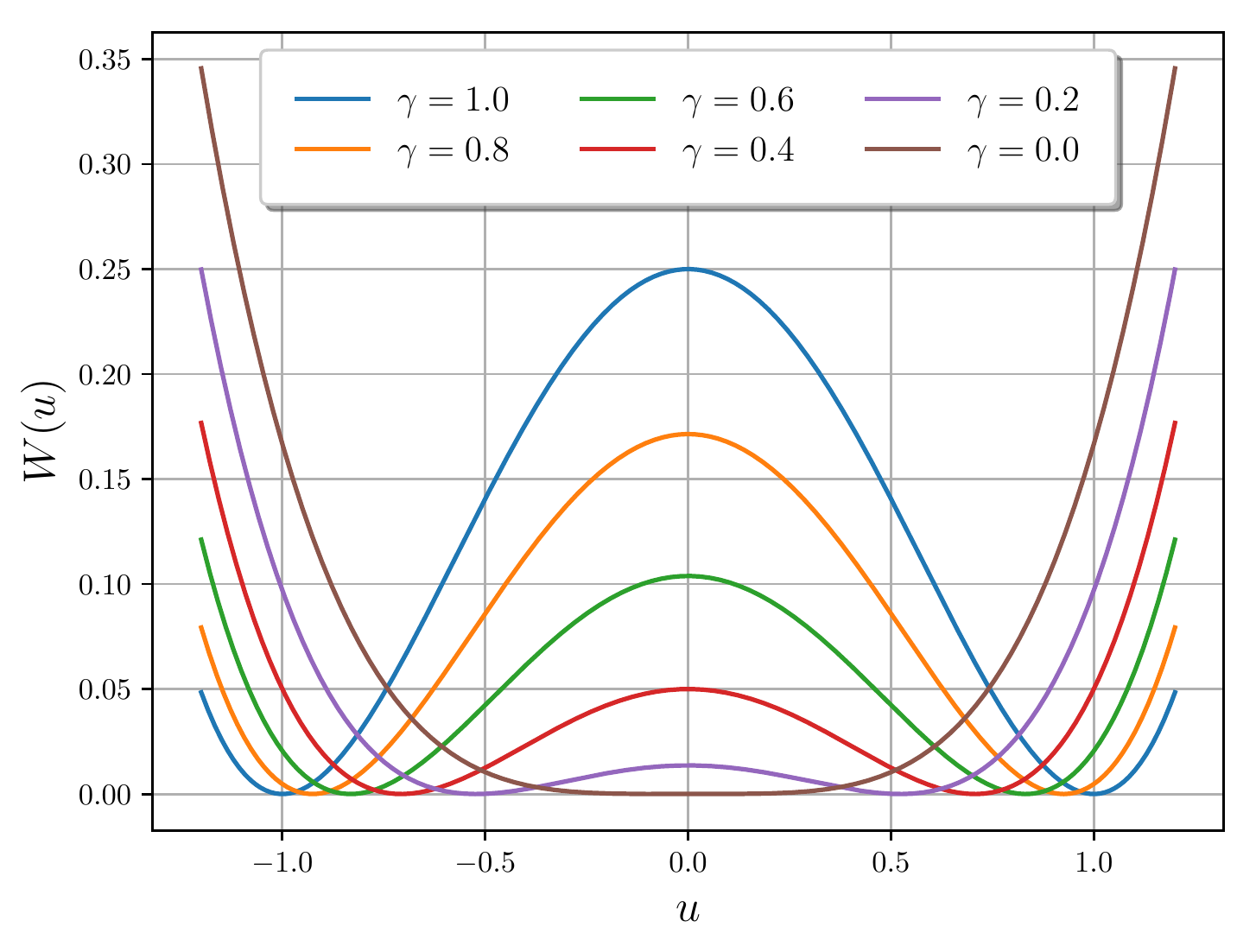}
    \caption{The modified double well potential with $\gamma = 0,0.2,0.4,0.6,0.8,1$.}
    \label{fig:mod_dw}
\end{figure}
Consider a weighted linear operator $H_{\gamma}(u): \mathring{H}^1\to(\mathring{H}^1)^*$ defined with $W_\gamma''$:
\begin{equation}
    \big\langle H_\gamma (u)\hat{u}_0, \tilde{u}_0\big\rangle = \big(\kappa W_\gamma''(u)\hat{u}_0, \tilde{u}_0\big) + (\epsilon^2\grad \hat{u}_0, \grad\tilde{u}_0) + \big(\sigma (-\lap_N)^{-1}R(\hat{u}_0), \tilde{u}_0\big)\;.
\end{equation}
In particular, $H_{\gamma}(u) \to D^2F(u)$ as $\gamma \to 1$ and $H_0(u)$ is coercive at any $u\in \mathring{H}^{1}_m$:
\begin{equation}\label{eq:gn_hessian}
    \big\langle H_0(u)\tilde{u}_0, \tilde{u}_0\big\rangle = (2\kappa u^2,\tilde{u}_0^2) + \epsilon^2\norm{\grad \tilde{u}_0}^2 + \sigma \norm{R(\tilde{u}_0)}_{\mathring{H}^{-1}_0}^2 \geq \epsilon^2 \norm{\grad \tilde{u}_0}^2 \quad\forall \tilde{u}_0\in \mathring{H}^1\;.
\end{equation}
The following proposition indicates that the modified Hessian operator $H_\gamma(u)$ is coercive with a sufficiently small $\gamma$:
\begin{proposition}\label{prop:coercivity}
There exists a constant $\gamma_c(\kappa,\epsilon, \Omega)>0$ such that $H_{\gamma}(u)$ is coercive at any $u\in \mathring{H}^1_m$ for $\gamma \in[0, \gamma_c]$.
\end{proposition}
\begin{proof}
    Inequality \eqref{eq:gn_hessian} implies $\big\langle H_{\gamma}(u)\tilde{u}_0, \tilde{u}_0 \big\rangle = \big\langle H_0(u)\tilde{u}_0, \tilde{u}_0  \big\rangle + \big(\kappa \gamma(u^2-1)\tilde{u}_0,\tilde{u}_0\big)
    \geq \epsilon^2\norm{\grad \tilde{u}_0}^2
     - \kappa\gamma\norm{\tilde{u}_0}^2\geq(\epsilon^2- c_p^2\kappa\gamma)\norm{\grad\tilde{u}_0}^2\;\forall\tilde{u}_0\in\mathring{H}^1$, where $c_p$ is the Poincar\'e constant. Therefore, $\gamma_c =\epsilon^2/(c_p^2\kappa)$.
\end{proof}
\subsubsection{The modified Newton step problem: a backtracking approach}
Now consider the $\mathring{H}^1$ Hessian-modified Newton step problem:
\begin{align}\label{eq:newton_step_mod}
    \text{Find } \delta u_n\in \mathring{H}^1 \text{ such that}: &&H_{\gamma_n}(u_n)\delta u_n = - DF(u_n)\;,
\end{align}
where $\gamma_n\in(0,1]$ is small enough so that $\big\langle DF(u_n),\delta u_n) \big\rangle < 0$, i.e., $\delta u_n$ is a descent direction. The following algorithm, a corollary to Proposition \ref{prop:coercivity}, finds the appropriate $\gamma_n$ by backtracking along a decreasing sequence $\vect{\gamma}$ from $1$ to $0$.\\

\begin{algorithm}[H]\label{alg:newton_step}
\SetAlgoLined
\KwResult{Obtain $\delta u_n$ such that $\big\langle DF(u_n),\delta u_n) \big\rangle < 0$}
Given $u_n\in\mathring{H}^1_m$ such that $DF(u_n)\not\equiv 0$ and $\vect{\gamma} = \{\gamma_1 = 1,\gamma_2<1, \gamma_3<\gamma_2, \dots, 0\}$\;
 \For{$\gamma_k$ in $\vect{\gamma}$, $k=1,2,\dots$}{
  Solve \eqref{eq:newton_step_mod} for $\delta u_n$ with $\gamma_n = \gamma_k$\;
  \If{$\big\langle DF(u_n),\delta u_n) \big\rangle < 0$}{
    break
   }
 }
 \caption{Generating an energy-descending Newton step $\delta u_n\in\mathring{H}^1$ at $u_n\in\mathring{H}^1_m$}
\end{algorithm}
The above algorithm also applies to the corresponding $\mathring{H}^{-1}$ Hessian-modified Newton step problem and leads to monotonic energy descent after the initial projection step according to Proposition \ref{prop:equiv_newton_step} and Lemma \ref{lemma:projection}:
\begin{subequations}\label{eq:scaled_h-1_newton_step}
\begin{align}
&\text{Find }(\delta u_n,\hat{\mu})\in H^1\times H^1 \text{ such that:}\nonumber\\
    &\qquad\qquad(\grad\hat{\mu}_n, \grad\tilde{u}) + (\sigma\delta u_n, \tilde{u}) = -(\grad \mu_n, \grad\tilde{u}) - \big(\sigma(u_n-m),\tilde{u})\big)&&\forall \tilde{u}\in H^1\;,\label{eq:scaled_h-1_newton_step_delta_u}\\
    &\qquad\qquad(\hat{\mu}_n, \tilde{\mu}) - \big(\kappa W_{\gamma_n}''(u_n)\delta u_n, \tilde{\mu}\big) - \big(\epsilon^2 \grad (\delta u_n), \grad\tilde{\mu}\big) = 0  &&\forall\tilde{\mu}\in H^1\;.\label{eq:scaled_h-1_newton_step_mu_hat}
\end{align}
\end{subequations}

\indent Backtracking on the weight $\gamma$ from $1$ to $0$ until a descent direction is found guarantees that in the vicinity of a local minimum, the resulting coercivity of the Hessian will lead to termination of Algorithm \ref{alg:newton_step} with $\gamma = 1$, and thus the full Hessian will be used in producing the Newton step. This in turn guarantees that the proposed method attain asymptotic quadratic convergence.

%% file: sec_04/subsec_02_bactracking_line_search.tex
\subsubsection{Backtracking line search on \texorpdfstring{$F_{OK}$}{}}\label{subsec:btls}
Once $\delta u_n\in\mathring{H}^1$ is obtained such that $\big\langle DF(u_n), \delta u_n\big\rangle < 0$, we seek to determine a step size $t_n$ that provides sufficient descent in the energy. Here we consider using backtracking line search via the Armijo condition to obtain $t_n$. Given an Armijo constant $c\in(0,1)$, typically chosen to be $10^{-4}$, and a sequence of step sizes $\alpha_k = 2^{1-k}$, $k\in\N$, we seek to find the smallest $K\in\N$ such that
\begin{equation}
    F_{OK}(u_n + \alpha_K\delta u_n)\leq F_{OK}(u_n) + c\alpha_K\big\langle DF(u_n), \delta u_n\big\rangle\;.
\end{equation}
We take $t_n = \alpha_K$.\\
\indent This procedure involves evaluating the energy and the action of $DF(u_n)$ and, thus, requires computing $(-\lap_N)^{-1}R(u_n + \alpha_k \delta u_n-m)$ at various $\alpha_k$. Let us define $w_n\coloneqq (-\lap_N)^{-1}R(u_n-m)$ and $\delta w_n \coloneqq (-\lap_N)^{-1}R(\delta u_n)$. Then the linearity of the Laplacian inverse operator leads to:
\begin{equation}
    (-\lap_N)^{-1} R(u_n + \alpha_k \delta u_n -m) = w_n + \alpha_k\delta w_n\;.
\end{equation}
One only needs to find $\delta w_n$ to evaluate energy and conduct the backtracking line search, if $w_n$ is known.

%% file: sec_04/subsec_03_on_the_initial_guess_and_the_homogeneous_state.tex
\subsection{On the initial guess and the homogeneous state}
A natural choice of the initial guess for the Newton iteration is the homogeneous state, i.e., $u_0 = m$, as it is an unbiased high-energy state for spinodal decomposition problems. However, this particular initial guess is not suitable for initiating the iterations, as it is a trivial solution to the first order condition \eqref{eq:1st_var}. Moreover, whether it satisfies the second-order condition \eqref{eq:2nd_var} is parameter-dependent:
\begin{subequations}
\begin{align}
    \big\langle D^2F(m)\tilde{u}_0, \tilde{u}_0\big\rangle &\leq (\epsilon^2 + c_p^2\sigma)\norm{\grad \tilde{u}_0}^2\qquad\qquad\qquad|m| <1/\sqrt{3}\;,\\
    \big\langle D^2F(m)\tilde{u}_0, \tilde{u}_0\big\rangle &\geq
    \begin{cases}
         \big(\epsilon^2+c_p^2\kappa(3m^2-1)\big)\norm{\grad\tilde{u}_0}^2 & |m|< 1/\sqrt{3}\;;\\
         \epsilon^2\norm{\grad\tilde{u}_0}^2 & |m|\geq 1/\sqrt{3}\;, 
    \end{cases}
\end{align}
\end{subequations}
for all $u_0\in\mathring{H}^1$. The homogeneous state is a local minimizer when $|m|\geq 1/\sqrt{3}$. It is also a local minimizer when $\epsilon^2\geq c_p^2\kappa(1-3m^2)$ and $|m|< 1/\sqrt{3}$. In other parameter regions, it is not straightforward to determine whether the homogeneous state is a local minimizer or a local maximizer on a general domain.\\
\indent For certain domains and boundary conditions, one can exactly determine whether the homogeneous state is a global minimizer. For example, Choksi et al. \cite{Choksi2009} showed that, for cubic domains with periodic boundary conditions and $\kappa = 1$, the homogeneous state is the global minimizer when $1-m^2\leq 2\epsilon\sqrt{\sigma}$. \\
\indent We accept the homogeneous state as a solution when it is a dominant minimizer, or even a global minimizer in some cases. However, for the purpose of acquiring the minimizers that exhibit phase separation, the initial guess should be sufficiently bounded away from the homogeneous state. The typical initial guess or initial condition used in the literature is the homogeneous state perturbed by pointwise-uniform noise \cite{Qin2013,Parsons2012, Zhang2006}:
\begin{equation}\label{eq:iid_ig}
    u_0(\vect{x}) = m + sr(\vect{x})\;,
\end{equation}
where $r(\vect{x})$ is a random variable with uniform distribution in $[-1, 1]$ and the scalar $s$ controls the distance to the homogeneous state, typically chosen to be $0.05$ for the gradient flow approach.\\
\indent However, $u_0$ does not belong to $H^1$ and depends on the spatial disctretization in the numerical implementation. We thus consider a similar form given by an isotropic Gaussian random field (GRF) transformed by the error function:
\begin{equation}\label{eq:grf_ig}
    u_0(\vect{x}) = m + s\,\textrm{erf}\big(u_G(\vect{x})\big)\;,\quad u_G\sim\mathcal{N}(0, \mathcal{C}_0)\;,
\end{equation}
where $\mathcal{C}_0$ is the covariance operator. We consider a covariance operator defined as the inverse of an elliptic differential operator \cite{Bui-thanh2013}:
\begin{equation}\label{eq:covariance}
    \mathcal{C}_0(\delta_G, \gamma_G) = (\delta_G - \gamma_G\lap)^{-2}\;.
\end{equation}
The Laplacian operator is equipped with the Robin boundary condition
\begin{equation}
    \grad u_G\cdot \nu+\beta u_G=0\quad \text{ on } \partial \Omega\;,
\end{equation}
where $\beta_G = \sqrt{\gamma_G\delta_G}/1.42$ is chosen to maintain a spatially-uniform pointwise-variance proportional to $\delta_G^{-4}(\gamma_G/\delta_G)^{-d/2}$.
The correlation length is proportional to $\sqrt{\gamma_G/\delta_G}$. The resulting GRF initial guesses can be tuned to mimic the initial guesses produced by \eqref{eq:iid_ig} for a given mesh while staying invariant to the spatial discretization.

%% file: sec_04/subsec_04_summary.tex
\subsection{Summary of the proposed scheme}
\begin{algorithm}[H]
\SetAlgoLined
\KwResult{Given an initial guess $u_0\in H^1$, generate an minimization sequence $\{(u_n,\mu_n)\}_{n=0}^N$ of $F_{\textrm{OK}}$ such that $\norm{B(u_N, \mu_N)}_{((H^1)^*)^2} < \textit{tol}$.}
$n = 0$\;
\If{$u_0\not\in \mathring{H}^1_m$}{
    Solve \eqref{eq:mu_n_var} for $\mu_0$\;
    Solve \eqref{eq:scaled_h-1_newton_step} for $\delta u_0$ with $\gamma_0 = 1$\;
    $u_1 = u_0 + \delta u_0$\;
    $n = n+1$\;
}
Solve \eqref{eq:mu_n_var} for $\mu_n$\;
Solve the Poisson problem for $w_n$\;
\While{$\norm{B(u_n, \mu_n)}_{((H^1)^*)^2} \geq \textit{tol}$}{
    Solve for $\delta u_n$ by Algorithm \ref{alg:newton_step} with \eqref{eq:scaled_h-1_newton_step}\;
    Solve the Poisson problem for $\delta w_n$\;
    Conduct Armijo backtracking line search on $F_{\textrm{OK}}$ to obtain the step size $t_n$\;
    $u_{n+1} = u_n + t_n\delta u_n$\;
    $w_{n+1} = w_n + t_n\delta w_n$\;
    Solve \eqref{eq:mu_n_var} for $\mu_{n+1}$\;
    $n = n+1$\;
 }
 \caption{The energy-descending Newton iteration for the minimization of the OK energy.}
\end{algorithm}

%% file: sec_05_numerical_results.tex
\section{Numerical Examples}\label{sec:numerical_results}
In this section, we present numerical examples of the minimization of the OK energy functional by the proposed algorithm described in the previous section. We use the finite element method to discretize the $H^1$ space. The discretizations and solves are implemented through version 2019.1.0 of the FEniCS library \cite{AlnaesBlechta2015a, Logg2007a}. The numerical examples include: (1) initial guesses generated by GRFs; (2) the demonstration of mass-conserving, energy-descending, and asymptotically quadratically convergent properties of the proposed scheme; (3) a mesh refinement study to show the independence of the number of minimization iterations from mesh discretization; (4) a comparison between the solutions obtained by the conventional gradient flow approach and the proposed scheme.
\input{sec_05/subsec_01_initial_guess}

\input{sec_05/subsec_02_mass-conservation_and_energy_stability}
\input{sec_05/subsec_03}
\input{sec_05/subsec_04_comparison_to_the_gradient_flow_approach}

%% file: sec_05/subsec_01_initial_guess.tex
\subsection{Initial guesses}
We consider a unit square domain discretized by uniform linear triangular elements. Initial guesses formed by both the pointwise-uniform noise \eqref{eq:iid_ig} and the GRFs \eqref{eq:grf_ig} are generated on the domain, with $m = 0$ and $s = 1$. The GRFs are obtained through the covariance operator \eqref{eq:covariance} with parameters $(\delta_G,\gamma_G) = (100, 0.0016)$. The numerical implementation of the GRF generation can be found in the hIPPYlib library \cite{VillaPetraGhattas18}. The initial guesses for both approaches are formed on meshes with different cell sizes and compared in Figure \ref{tab:ig_plots}. It is clear that the structure of the initial guesses generated by the GRFs is invariant to changes in the discretization while the structure of the initial guesses generated by the pointwise-uniform noise is not. Additionally, the former follow the pointwise-uniform distribution after the transformation through the error function, as indicated by Figure \ref{fig:pwv_grf}.\\
\indent A heuristic approach to generate initial guesses by the GRFs that mimic the conventional initial guesses on the mesh with a distance of scale $l_e$ between its neighboring degree-of-freedoms (DoFs) is to set $\gamma_G$ and $\delta_G$ such that $\sqrt{\gamma_G/\delta_G} = 0.4l_e$ and $\delta_G^{-2}(\gamma_G/\delta_G)^{-d/4} = 2.5l_e$.
\begin{table}[!htbp]
    \centering
    \begin{tabular}{|c|c|c|}\hline
        $100\times100$ cells & $200\times200$ cells & $300\times300$ cells  \\\hline
        \includegraphics[width = 0.3\linewidth]{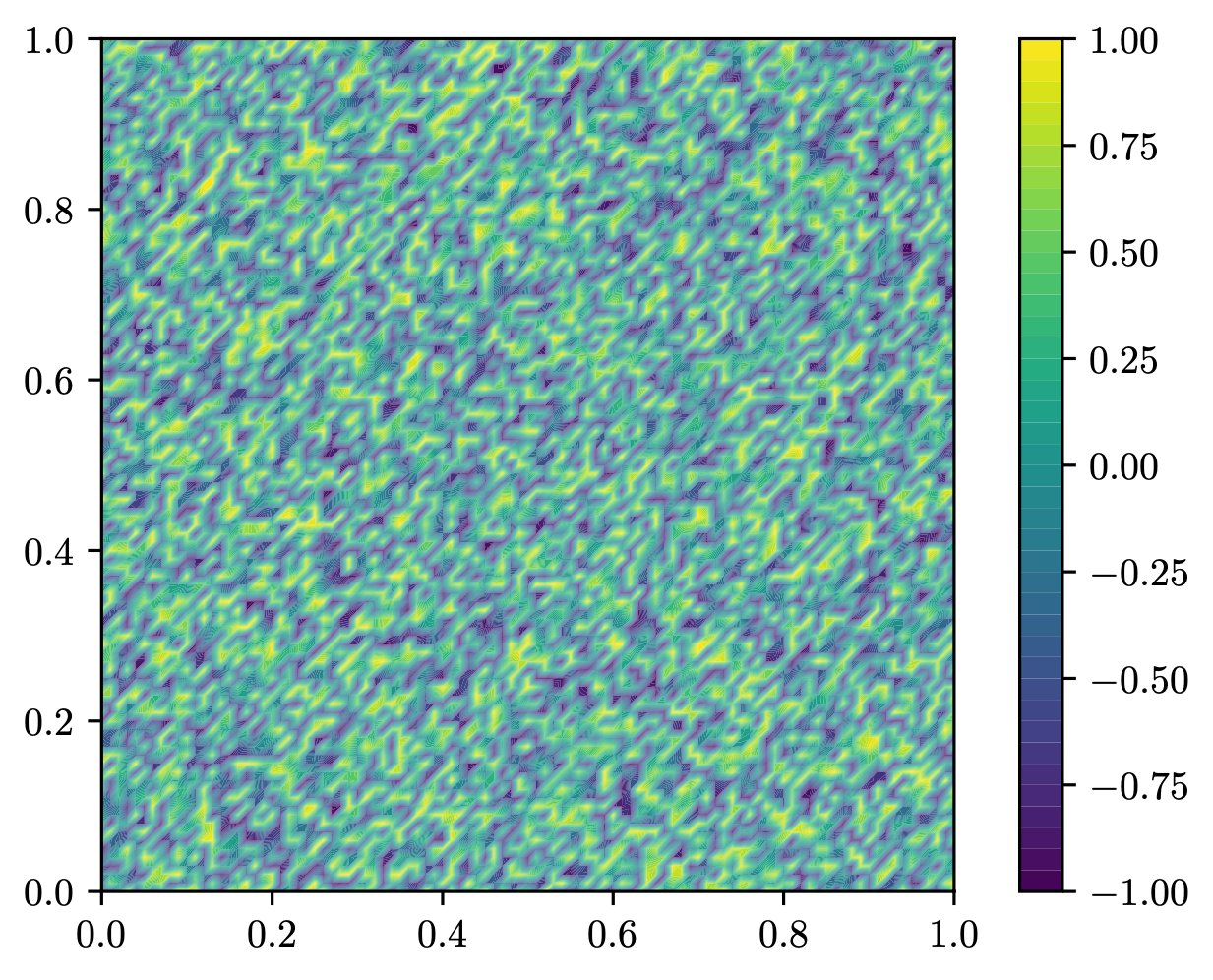} &        \includegraphics[width = 0.3\linewidth]{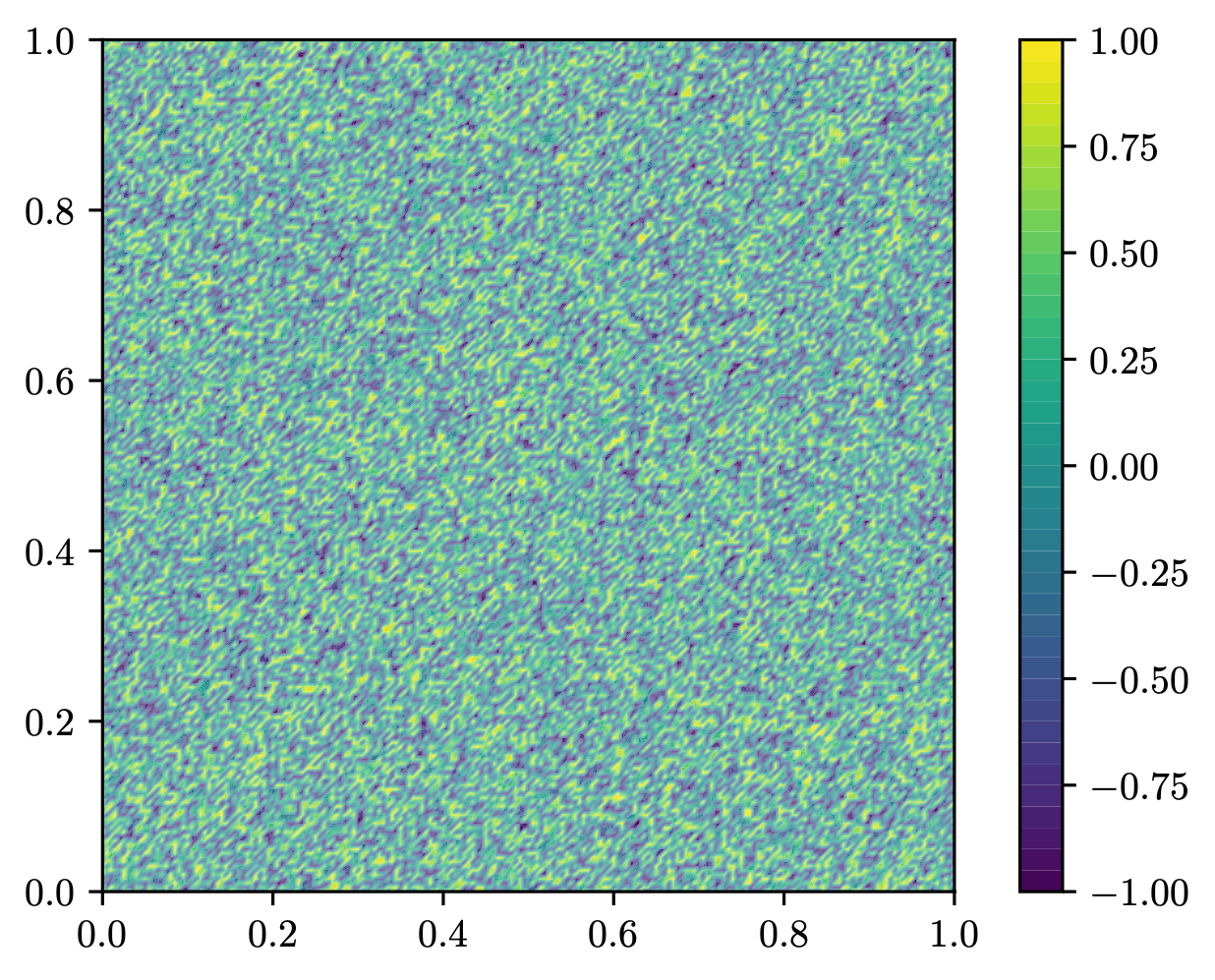} &         \includegraphics[width = 0.3\linewidth]{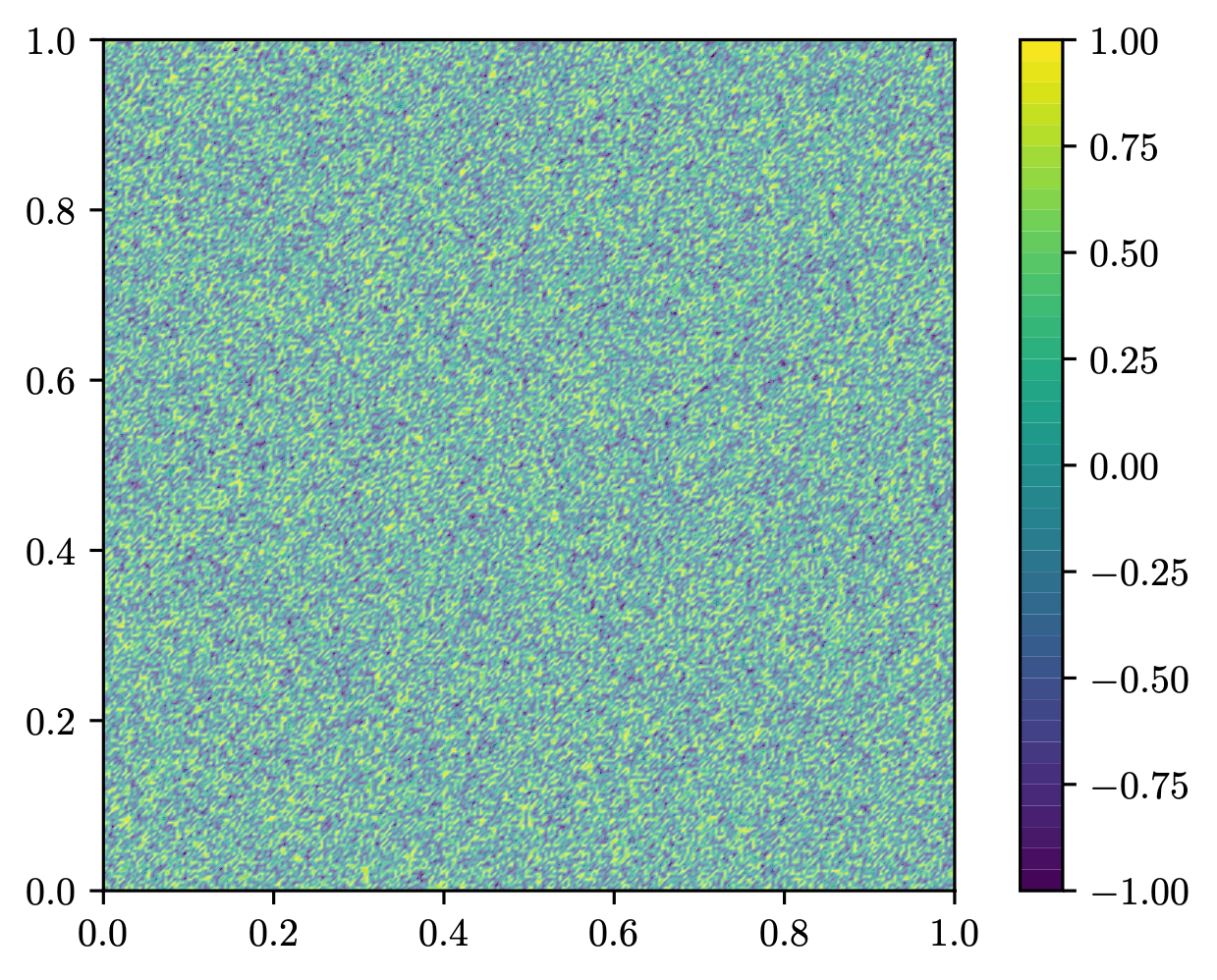}\\\hline
        \includegraphics[width = 0.3\linewidth]{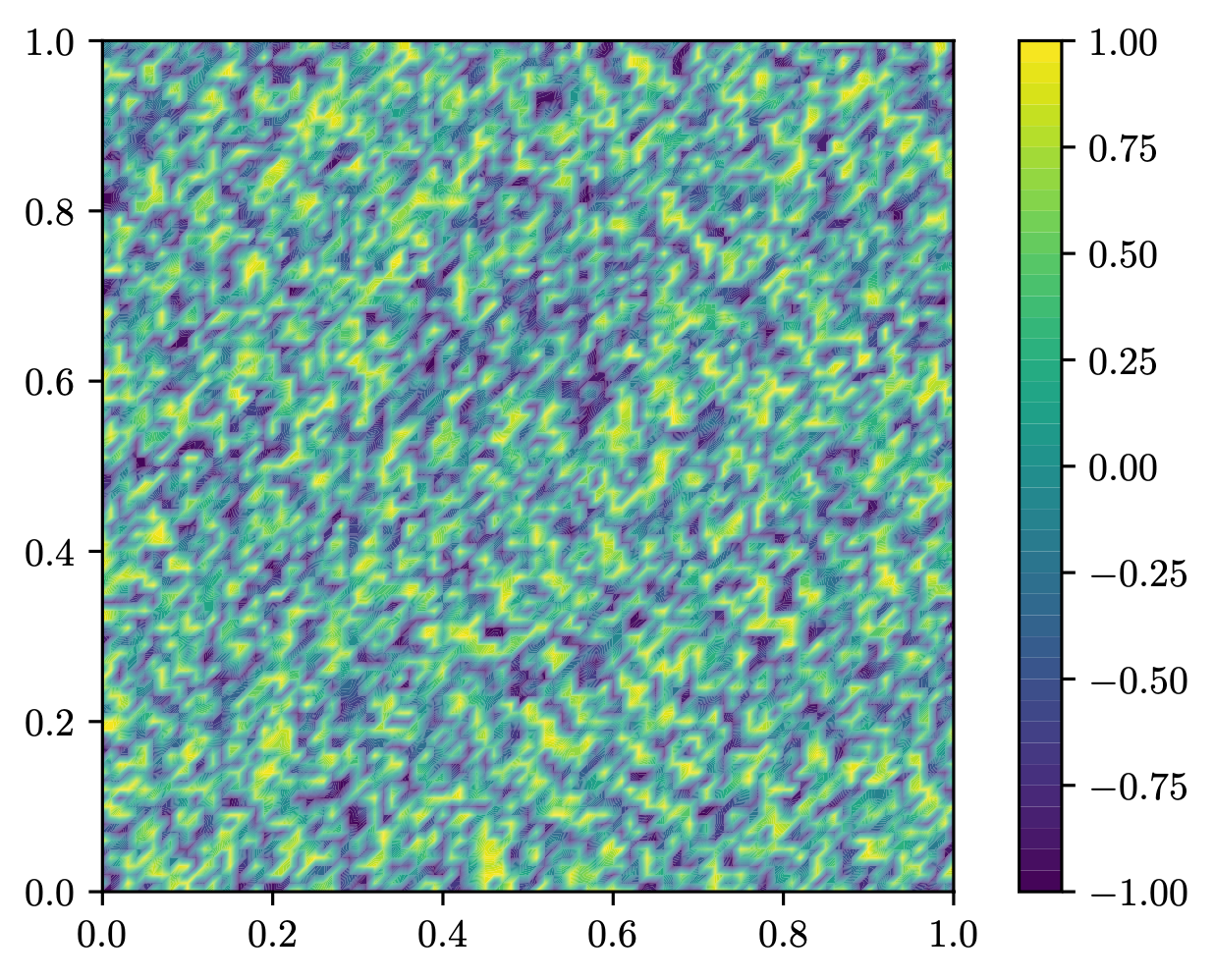} &         \includegraphics[width = 0.3\linewidth]{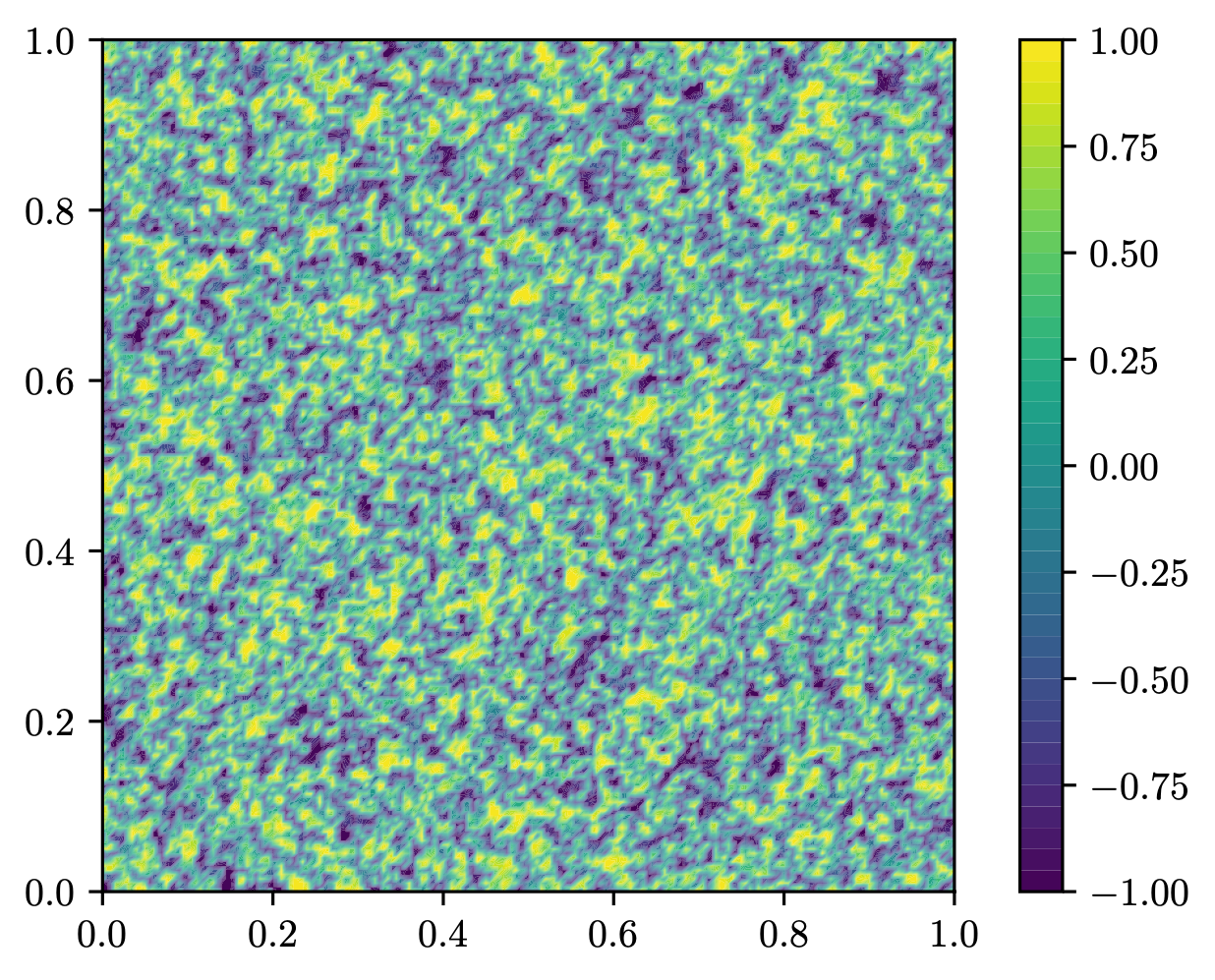} & \includegraphics[width = 0.3\linewidth]{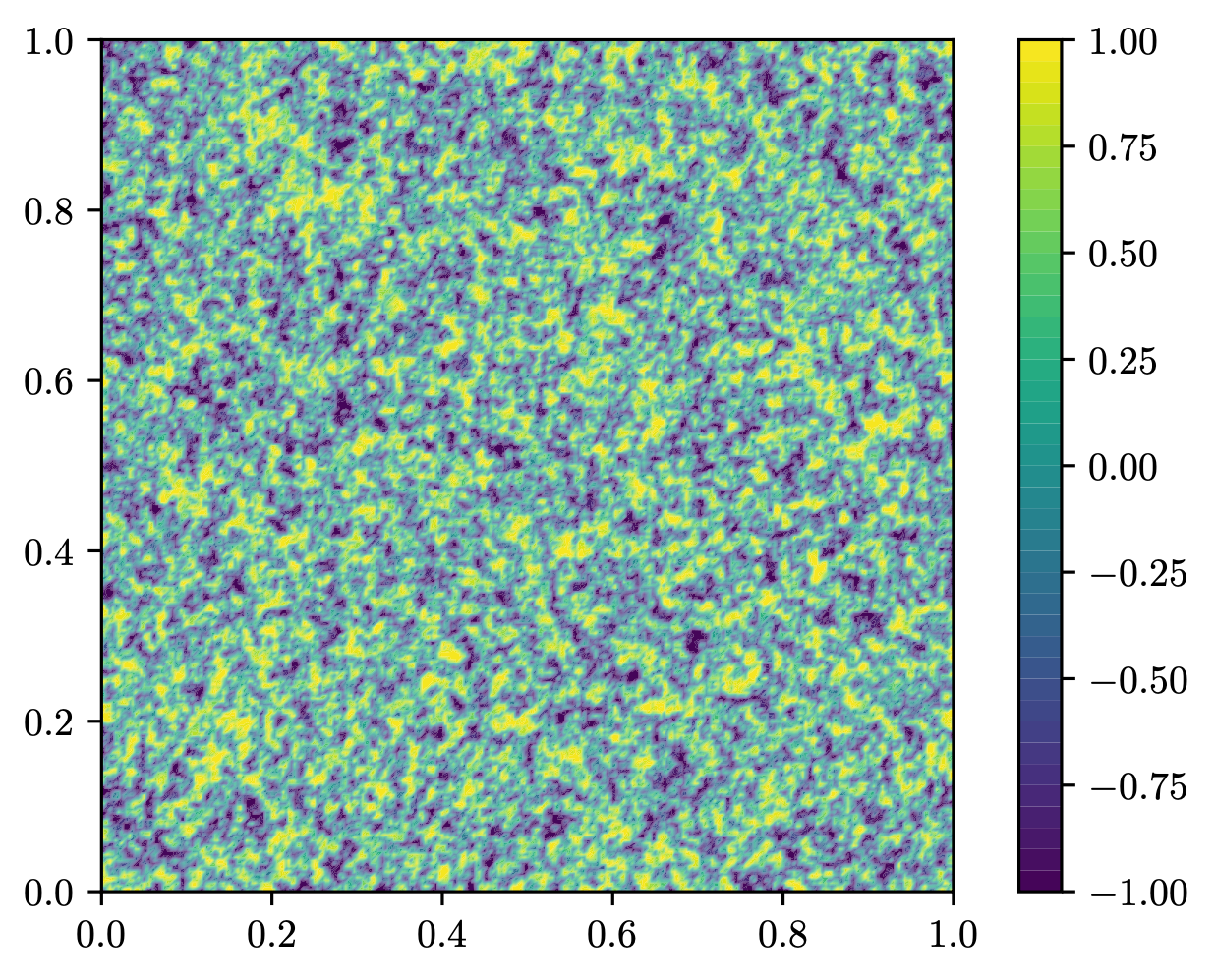}\\
     \hline
    \end{tabular}
    \captionof{figure}{Initial guesses generated on meshes consisting of uniform triangular elements with decreasing cell sizes. The figures on top row are the initial guesses generated by the pointwise-uniform noise \eqref{eq:iid_ig} and the figures on the bottom row are the initial guesses generated by the GRFs \eqref{eq:grf_ig}.}
    \label{tab:ig_plots}
\end{table}

\begin{figure}[!htb]
    \centering
    \includegraphics[width = 0.5\linewidth]{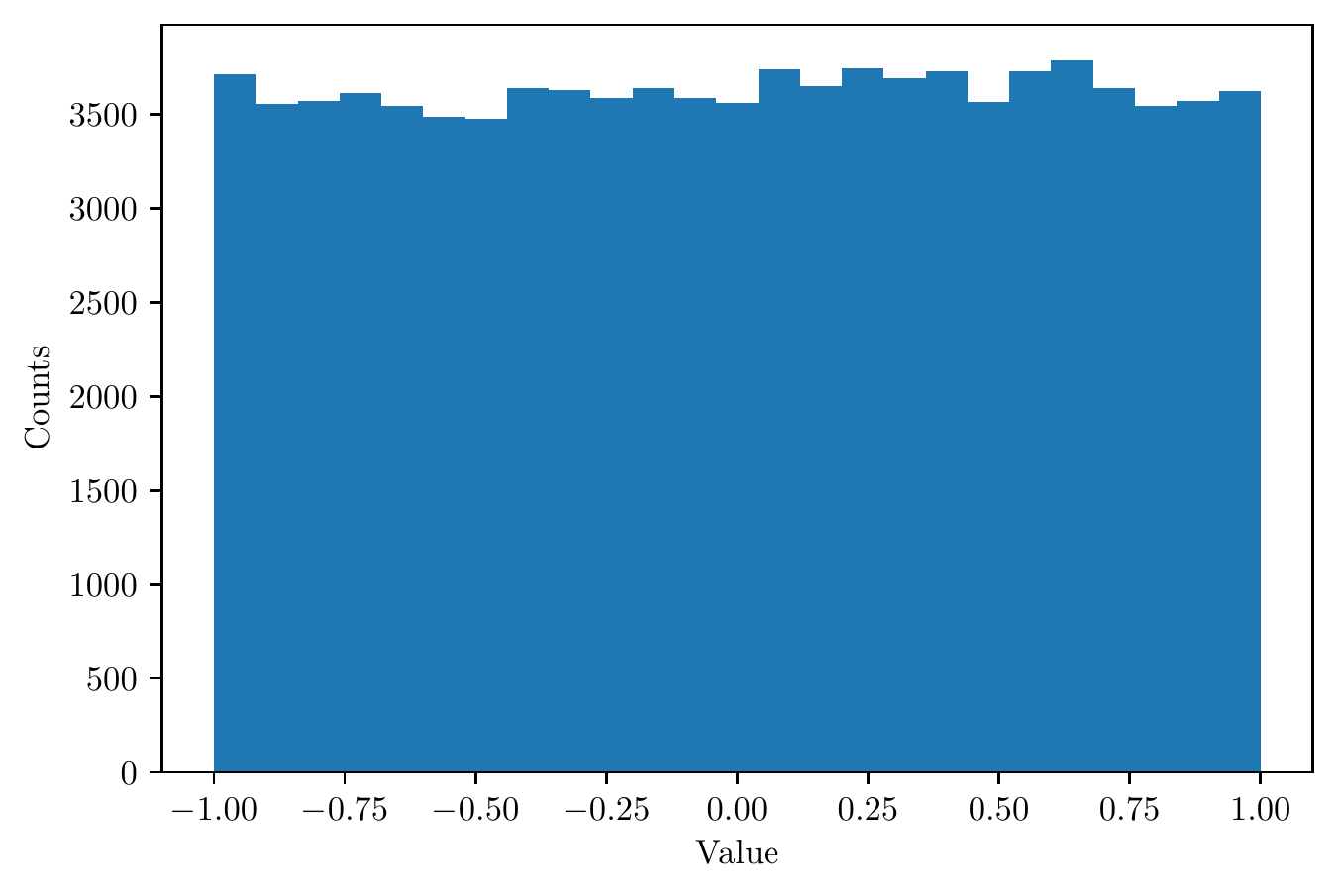}
    \caption{A histogram of the vertex values of the transformed GRFs in Figure \ref{tab:ig_plots} with the mesh of $300\times300$ cells. The pointwise-uniform distribution across samples is implied by ergodicity.}
    \label{fig:pwv_grf}
\end{figure}

%% file: sec_05/subsec_02_mass-conservation_and_energy_stability.tex
\subsection{Mass conservation, monotonic energy descent, and asymptotic quadratic convergence}
    We consider a square 2D domain of $\Omega = [0,40]^2$. The domain is discretized with uniform triangular elements with $100$ cells in each direction. The finite element discretization and the resulting systems of equations of the $\mathring{H}^{-1}$ Newton step problem and the homogeneous Neumann Poisson problem are presented in Appendix \ref{appendix:b}. All the systems of equations are symmetrized and solved by direct solvers through Cholesky decomposition.\\
    \indent Let us consider a set of parameters $(m,\kappa,\epsilon,\sigma) = (0, 1, 0.4, 0.7)$. A minimization sequence, $\{u_n, \mu_n\}_{n=0}^N$, for the parameters is produced via the proposed scheme. The double well backtracking sequence is set to $\vect{\gamma} = \{1, 0.5, 0\}$ and the residual norm tolerance of $\textit{tol} = 10^{-8}$. The initial guess is generated by setting $(s, \delta_G, \gamma_G) = (0.05, 2.5, 0.064)$. In Figure \ref{fig:evolution}, the evolution of (1) the OK free energy, (2) the distance from $\mathring{H}^1_m$, $\{\big|(u_n -m, 1)\big|\}_{n = 0}^N$, (3) the residual norm, $\big\{\norm{B(u_n,\mu_n)}_{((H^1)^*)^2}\big\}_{n=1}^N$, and (4) the weight on the modified Hessian operator, $\{\gamma_n\}_{n=1}^N$, of the minimization sequence are presented. The figures show that the energy is monotonically decreasing. The double well backtracking algorithm for generating energy-descending Newton steps is activated many times, indicated by the change in $\gamma_n$. These results confirm that the proposed scheme leads to monotonic energy descent that enables global convergence. Notice that $u_0\not\in \mathring{H}^1_m$ and the first step is a projection step. The rest of the sequence stays in $\mathring{H}^1_m$, which verifies Lemma \ref{lemma:projection} of the proposed scheme. The figures show that asymptotic quadratic convergence in the residual norm, with no corresponding double well backtracking ($\gamma_n = 1$). The evolution of the order parameter is shown in Figure \ref{tab:evolution_order_parameter}.
    \begin{figure}
        \centering
        \begin{subfigure}{0.49\linewidth}
        \includegraphics[width =  \linewidth]{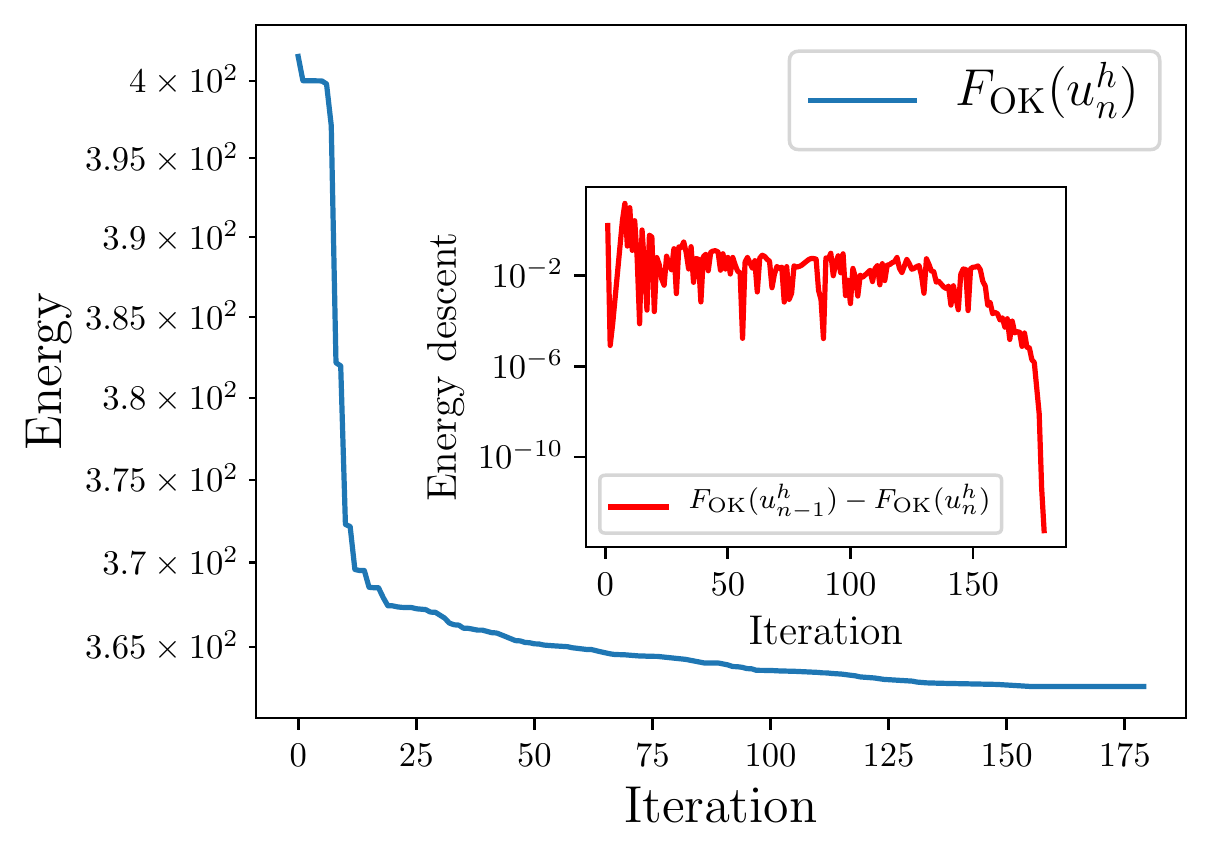}
        \end{subfigure}
        \begin{subfigure}{0.49\linewidth}
        \includegraphics[width =  \linewidth]{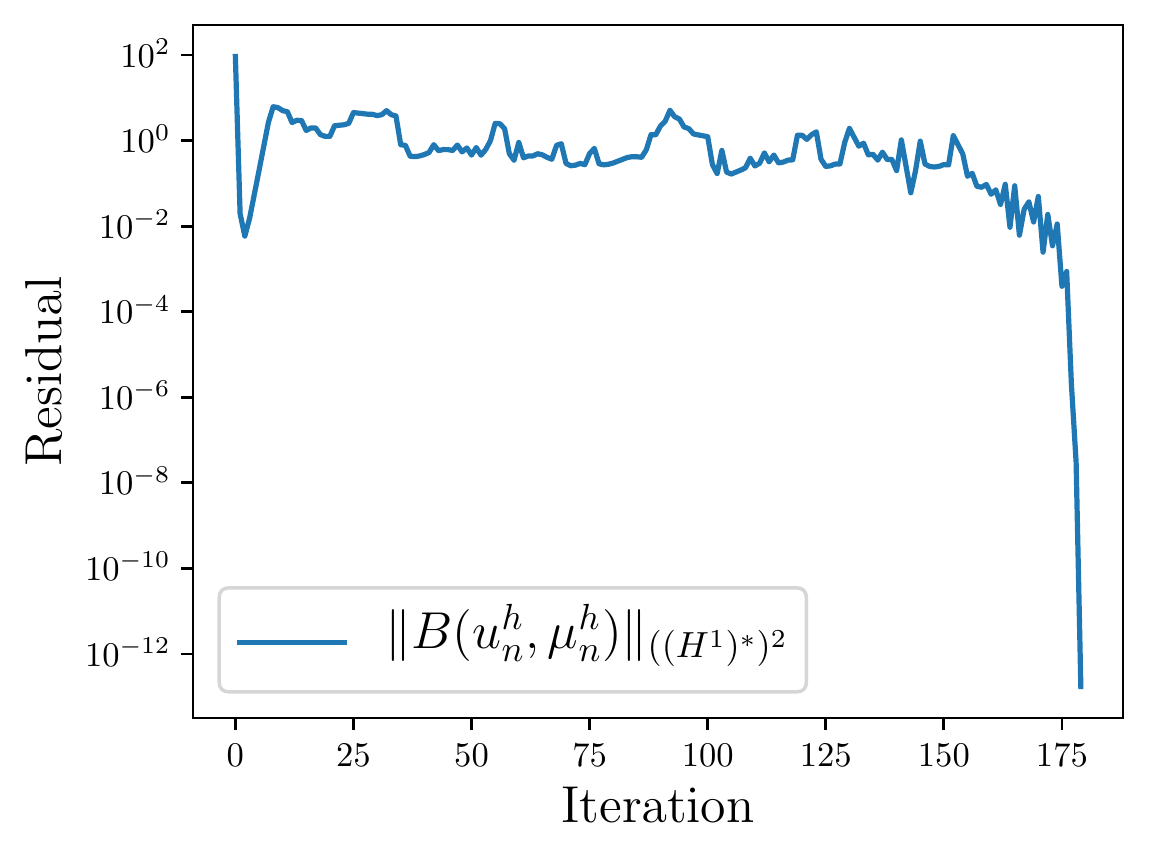}
        \end{subfigure}
        \begin{subfigure}{0.49\linewidth}
        \includegraphics[width =  \linewidth]{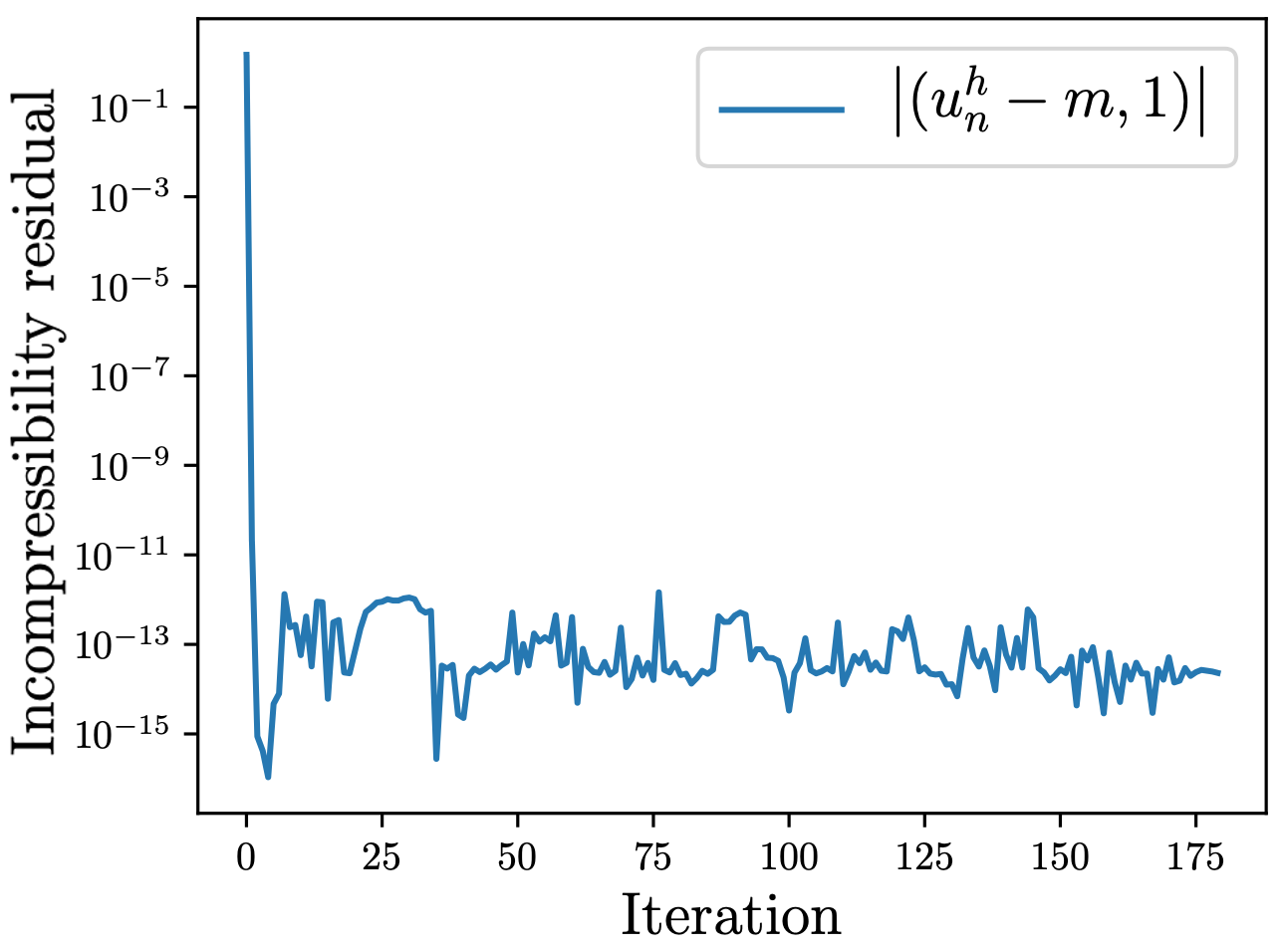}
        \end{subfigure}
        \begin{subfigure}{0.49\linewidth}
        \includegraphics[width =  \linewidth]{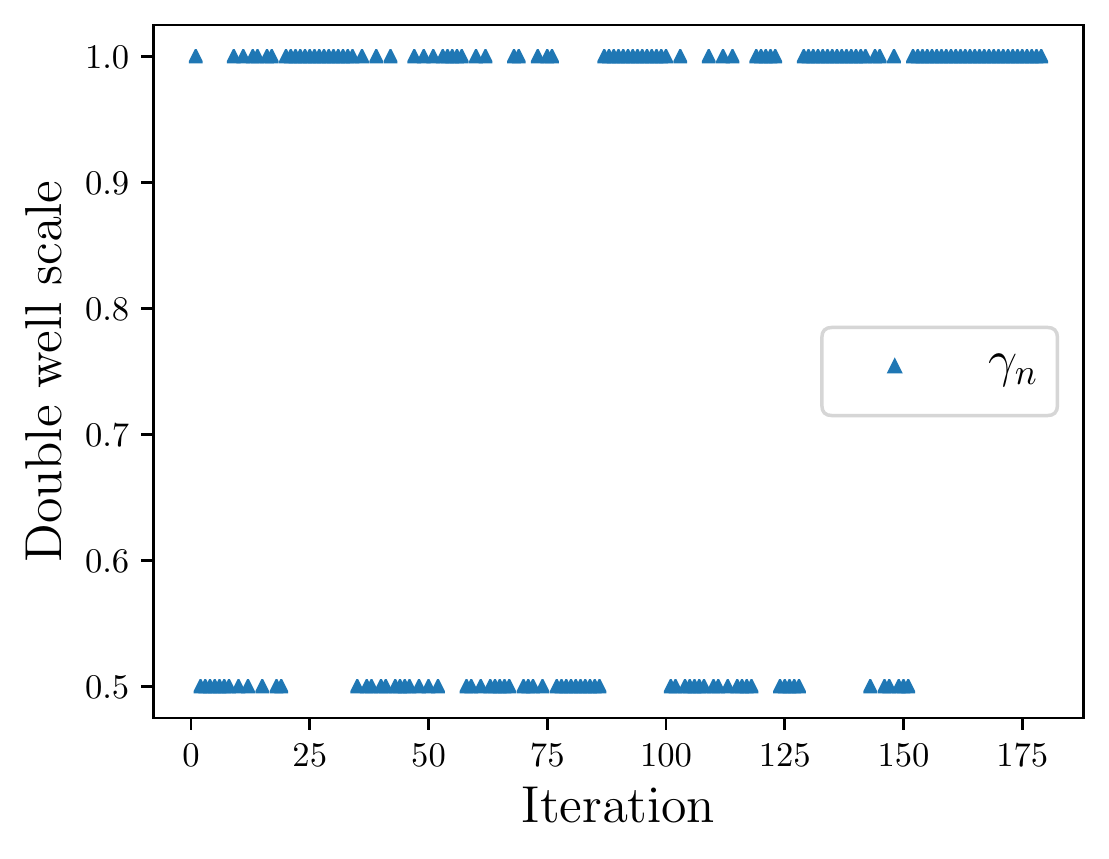}
        \end{subfigure}
        \caption{The evolution of the OK energy, the residual norm, the distance to $\mathring{H}^1_m$ (or the incompressibility residual), and the double well backtracking constant of a minimization sequence generated by the proposed scheme, with the parameters $(m,\kappa,\epsilon,\sigma) = (0, 1, 0.4, 0.7)$ and the double well backtracking sequence $\vect{\gamma} = \{1, 0.5, 0\}$.}
        \label{fig:evolution}
    \end{figure}
    
    \begin{table}
    \centering
    \begin{tabular}{|c|c|c|c|c|}\hline
        $n = 0$ & $n = 20$ & $n = 40$ & $n= 60$ & $n= 80$ \\\hline
        \includegraphics[width = 0.16\linewidth, trim=5 5 70 0, clip]{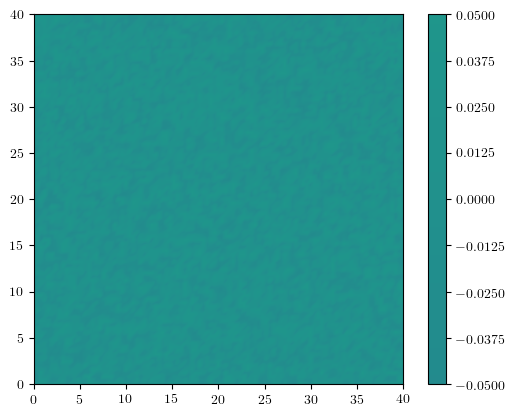} &        \includegraphics[width = 0.16\linewidth, trim=5 5 70 0, clip]{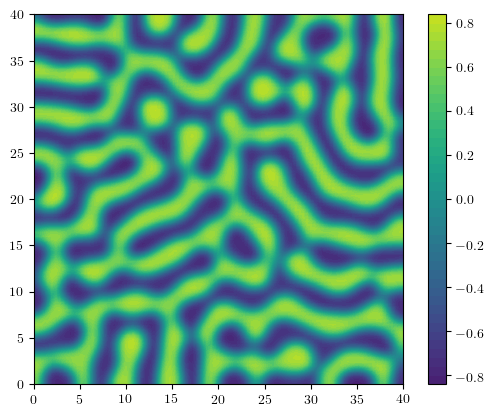} &         \includegraphics[width = 0.16\linewidth, trim=5 5 70 0, clip]{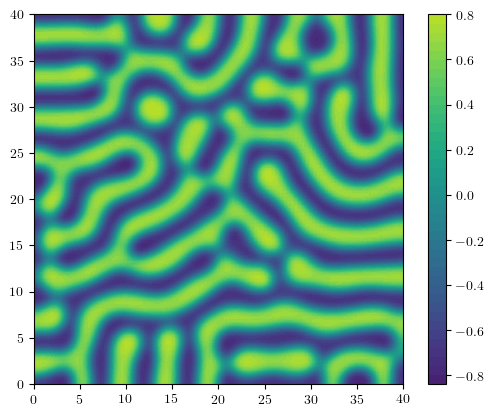} & \includegraphics[width = 0.16\linewidth, trim=5 5 70 0, clip]{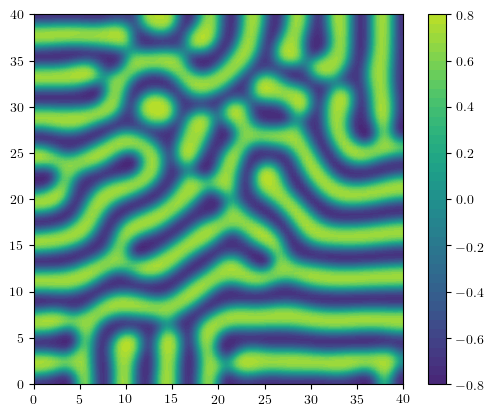} &         \includegraphics[width = 0.16\linewidth, trim=5 5 70 0, clip]{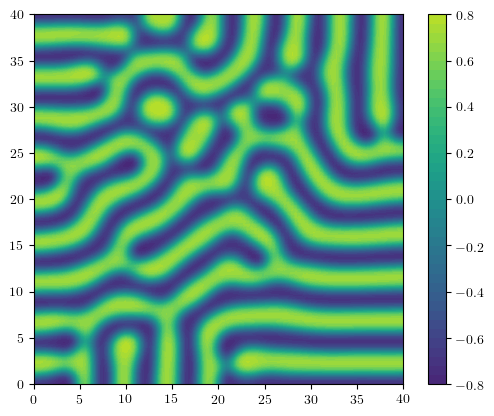} \\\hline
        $n = 100$ & $n = 120$ & $n = 140$ & $n= 160$ & $n= 179$ (final) \\\hline
        \includegraphics[width = 0.16\linewidth, trim=5 5 70 0, clip]{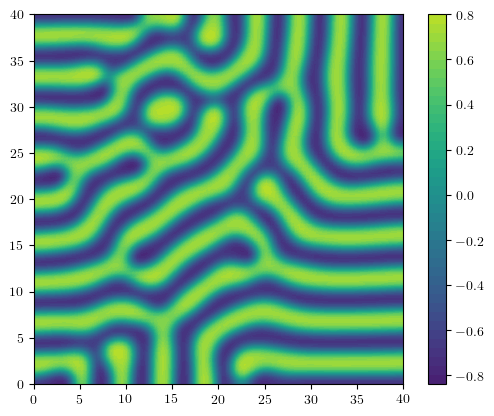} &        \includegraphics[width = 0.16\linewidth, trim=5 5 70 0, clip]{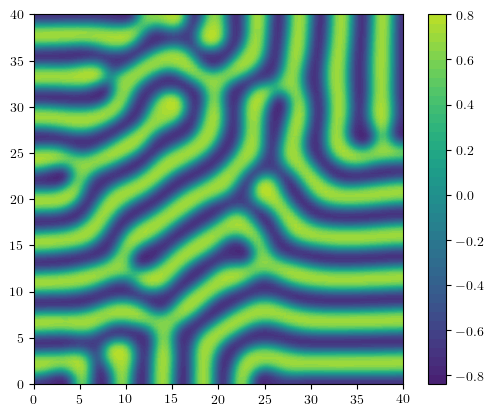} &         \includegraphics[width = 0.16\linewidth, trim=5 5 70 0, clip]{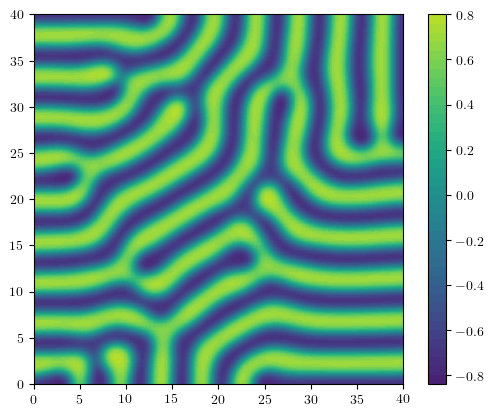} & \includegraphics[width = 0.16\linewidth, trim=5 5 70 0, clip]{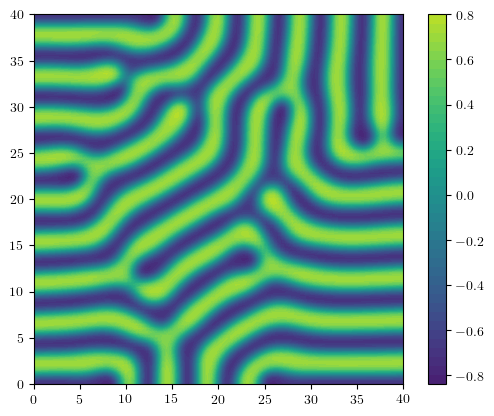} &         \includegraphics[width = 0.16\linewidth, trim=5 5 70 0, clip]{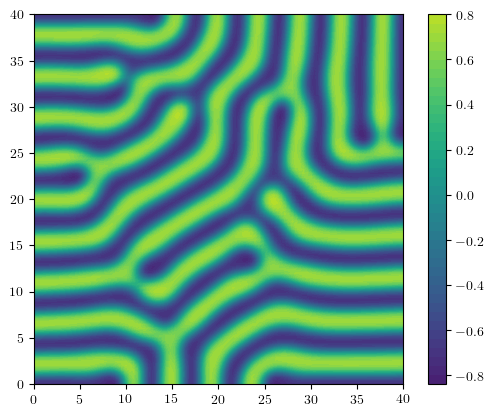} \\
     \hline
    \end{tabular}
    \captionof{figure}{The evolution of the order parameter for a minimization sequence generated by the proposed scheme, with the parameters $(m,\kappa,\epsilon,\sigma) = (0, 1, 0.4, 0.7)$ and the double well backtracking sequence $\vect{\gamma} = \{1, 0.5, 0\}$.}
    \label{tab:evolution_order_parameter}
\end{table}
    

%% file: sec_05/subsec_03.tex
\subsection{Independent of the number of iterations from mesh discretization}
We present a numerical example to demonstrate that the number of iterations to obtain the OK model solutions via the proposed scheme is generally independent of the number of the DoFs of the discretization. Consider the square 2D domain of $\Omega = [0,40]^2$ discretized with uniform triangular elements of $200$, $400$, $600$, and $800$ cells in each direction. The OK model is solved via the proposed scheme with the parameters $(m, \kappa, \epsilon,\sigma) = (0.35, 1,  0.3, 0.7)$, the double well backtracking sequence $\bm{\gamma} = \{1, 0.5, 0\}$, and the four different discretizations. An initial guess is generated with $(s,\delta_G, \gamma_G) = (0.05, 2.5, 0.064)$ and the $800\times 800$ mesh discretization. The initial guess is then projected to the finite element spaces defined by the four discretizations. The number of iterations for each solve as well as the solution states are shown in Table \ref{tab:mesh_refine}. The numerical results show that the proposed Newton scheme takes a similar number of iterations to reach similar solution states for spatial discretization with increasing DoFs.

    \begin{table}
    \centering
    \begin{tabular}{|c|c|c|c|c|}\hline
        \textbf{Discretization} &$200\times 200$ cells& $400\times 400$ cells & $600\times 600$ cells & $800\times 800$ cells \\\hline
        \textbf{Number of iterations} & $118$ & $103$ & $102$ & $111$\\\hline
         & \multirow{7}{*}{\includegraphics[width = 0.16\linewidth]{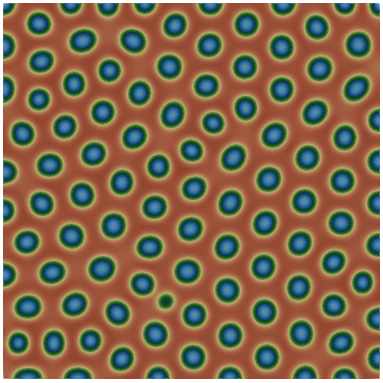}} & \multirow{7}{*}{\includegraphics[width = 0.16\linewidth]{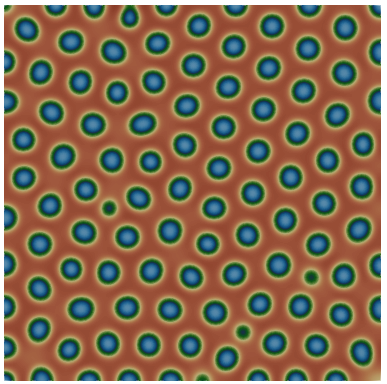}} & \multirow{7}{*}{\includegraphics[width = 0.16\linewidth]{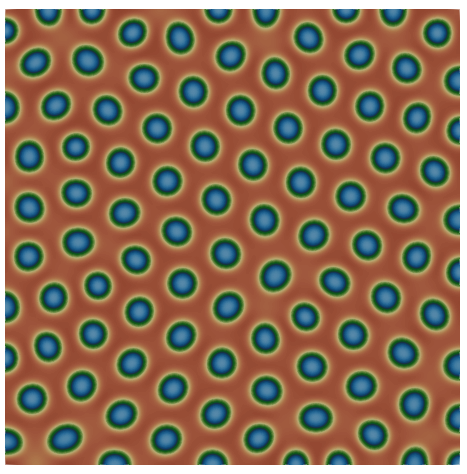}} & \multirow{7}{*}{\includegraphics[width = 0.16\linewidth]{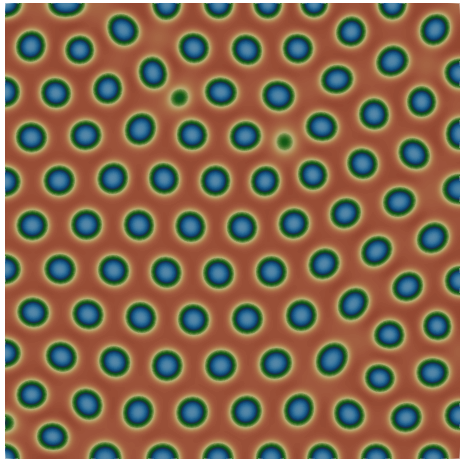}} \\
        & & & &\\
        & & & &\\
        \textbf{Model solutions}& & & &\\
        & & & &\\
        & & & &\\
        & & & &\\\hline
    \end{tabular}
    \caption{The model solutions and the number of iterations to reach below a residual norm value of $10^{-8}$ given by the proposed Newton scheme, with model parameters $(m, \kappa, \epsilon,\sigma) = (0.35, 1,  0.3, 0.7)$, the double well backtracking sequence $\bm{\gamma} = \{1, 0.5, 0\}$, and increasing numbers of DoFs for the discretization.}
    \label{tab:mesh_refine}
\end{table}

%% file: sec_05/subsec_04_comparison_to_the_gradient_flow_approach.tex
\subsection{Comparison to the gradient flow approach}
\indent Now we consider the gradient flow approach to minimize the OK energy, i.e., solving the following Cahn-Hilliard equation
\begin{subequations}
\begin{align}
    u_t(\vect{x},t) - \lap\Big(kW'\big(u(\vect{x},t)\big)-\epsilon^2\lap u(\vect{x},t)\Big) + \sigma\big(u(\vect{x},t)-m\big) &= 0&&(\vect{x},t)\in\Omega\times (0,T]\\ 
    \grad u(\vect{x},t)\cdot\nu = \grad \big(\lap u(\vect{x},t)\big)\cdot\nu &= 0 &&(\vect{x},t)\in\partial\Omega\times(0,T]\\
    u(t,\vect{x})- u_0(\vect{x}) &=0 && (\vect{x},t)\in\Omega\times\{0\}
\end{align}
\end{subequations}
where $u:H^1\big(0, T;H^4(\Omega)\big)$.\\
\indent The solutions obtained by the gradient flow approach and the proposed scheme are compared in Table \ref{tab:comparison}. The computational domain for the solutions is the unit square and it is discretized by uniform triangular finite elements with $200$ cells in each direction. Solutions for the same parameters are obtained through the same initial state generated with $(s,\delta_G,\gamma_G) = (0.1, 100, 0.0016)$. Here we consider two sets of parameters with $(\kappa,\epsilon,\sigma) = (1.0, 0.01, 500)$ and two different mass averages, $m = 0$ and $m=0.3$. For the proposed scheme, a double well backtracking sequence of $\vect{\gamma} = \{1,0.75, 0.5, 0.25, 0\}$ is employed. For the gradient flow approach, the mixed formulation and the first-order convex-splitting scheme for temporal discretization is implemented \cite{Elliott1993,Eyre1998, Wu2014,Guan2014}. The step size for the temporal discretization is set to the conventional value $\delta t = \epsilon^2 = 10^{-4}$.\\
\indent The minimizers obtained by both methods are similar in terms of the feature size and periodicity, while the gradient flow approach is more costly. It requires $\mathcal{O}(10^{5})$ time steps to reach a value below the residual norm tolerance of $10^{-6}$. On the contrary, the proposed Newton scheme takes fewer than $200$ iterations by solving systems of equations with similar asymptotic computational complexity at each Newton iteration as that of the gradient flow approach at each time step. A comparison of the computational complexity at each time step or Newton iteration for the gradient flow approach and the proposed Newton scheme is presented in Table \ref{tab:complexity}. From this example, we can see that the proposed scheme is around three orders of magnitude faster than the traditional gradient flow approach.
\begin{table}
    \centering
    \begin{tabular}{|C|C|C|}\hline
    & \textbf{Gradient flow approach} & \textbf{Proposed Newton scheme} \tabularnewline\hline
     $\mathbf{m=0}$ & \includegraphics[width=0.9\linewidth,trim = 400 160 400 10, clip]{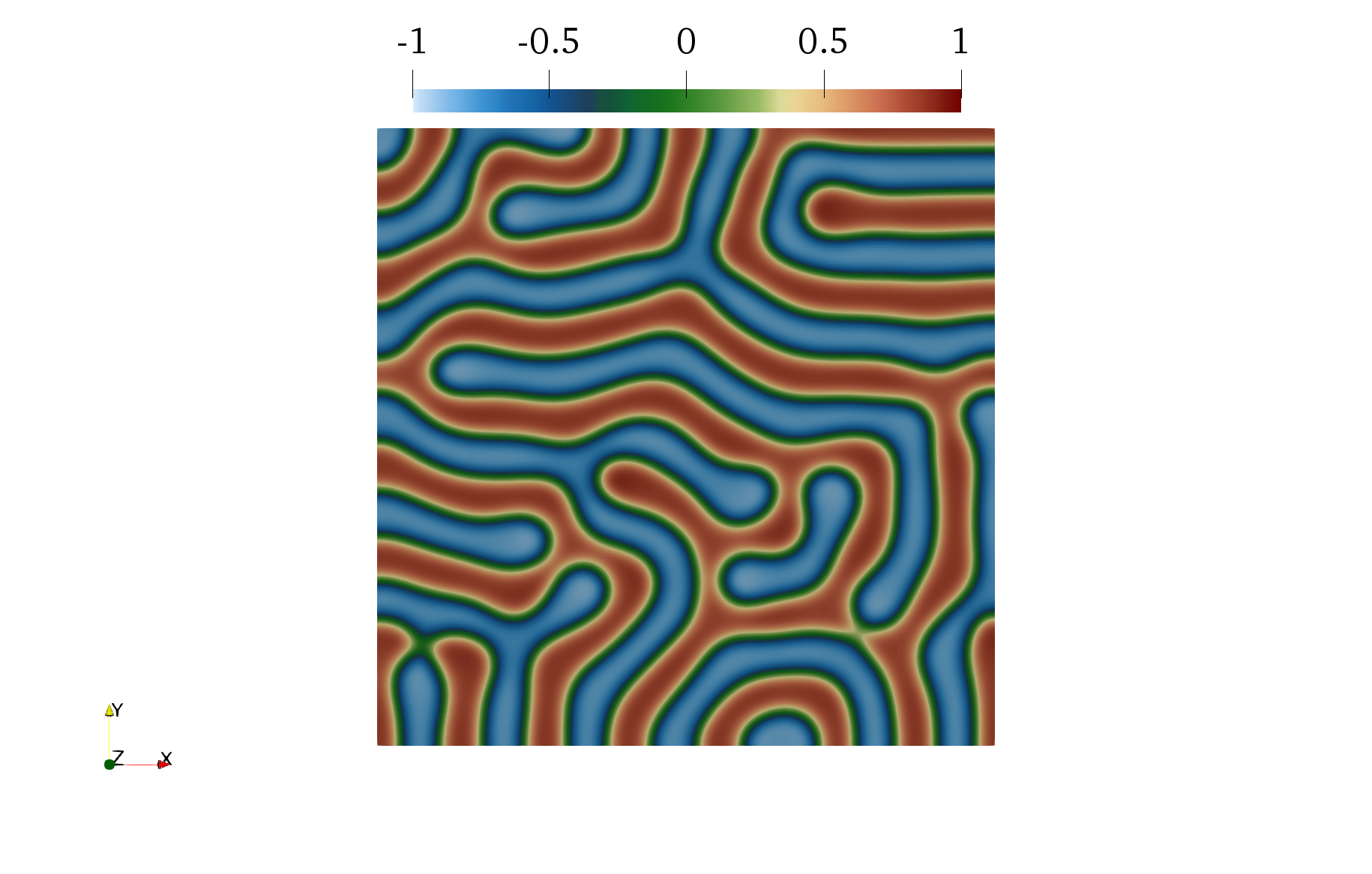} & \includegraphics[width=0.9\linewidth, trim = 400 160 400 10, clip]{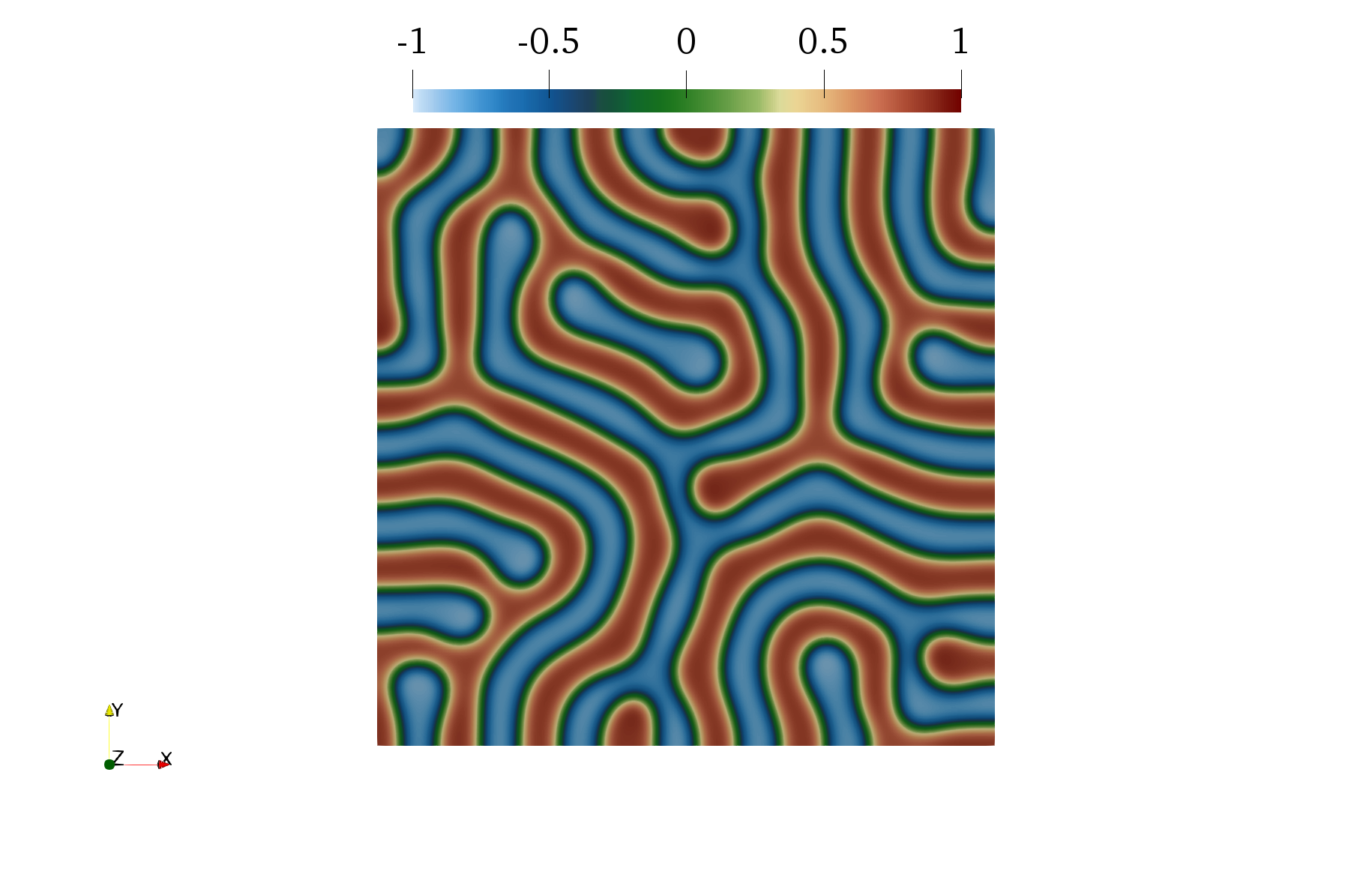}\tabularnewline\hline
     \textbf{Number of time steps or iterations} & $6.54\times10^5$ & $160$\tabularnewline\hline
      $\mathbf{m=0.3}$  & \includegraphics[width=\linewidth, trim = 400 160 400 10, clip]{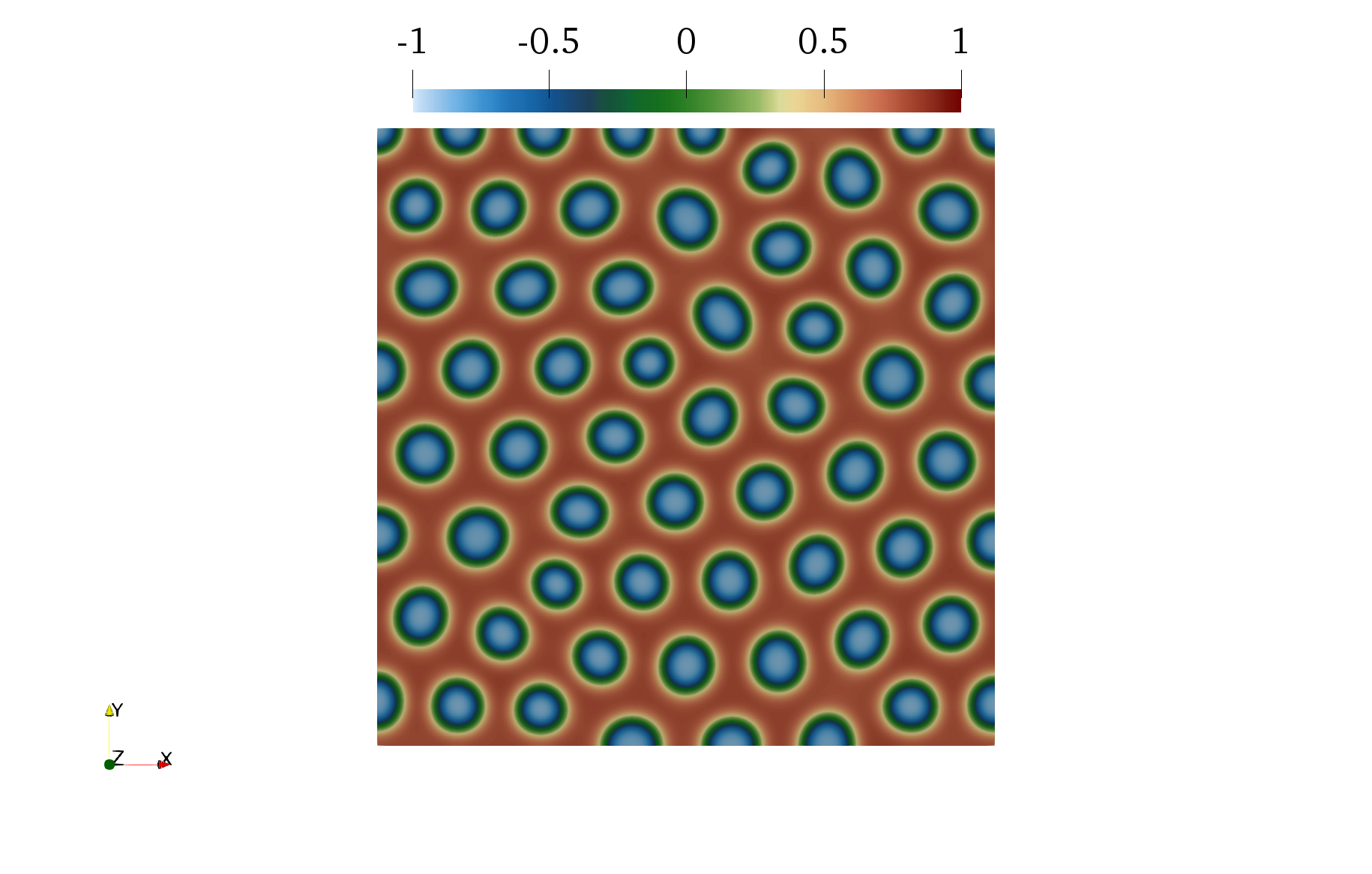} & \includegraphics[width=\linewidth, trim = 400 160 400 10, clip]{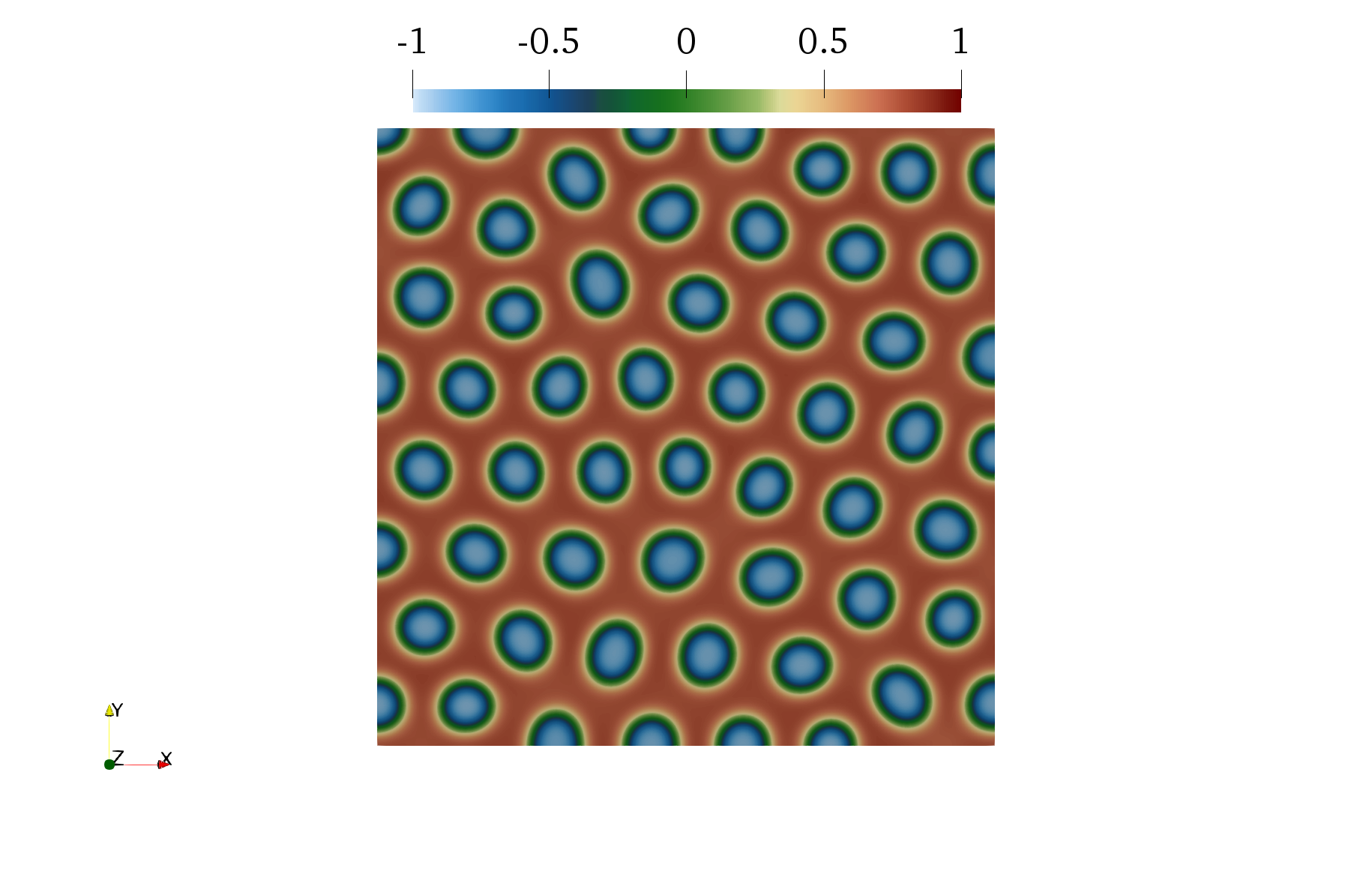}\tabularnewline\hline
     \textbf{Number of time steps or iterations} & $1.36\times 10^5$ & $141$\tabularnewline\hline
    \end{tabular}
    \caption{A comparison between the minimizers obtained through the gradient flow and the proposed scheme, and the number of time steps (gradient flow) or Newton iterations (proposed method) to reach below a residual norm value below $10^{-6}$.}
    \label{tab:comparison}
\end{table}

\begin{table}
    \centering
    \begin{tabular}{|C|C|C|}\hline
    & \textbf{Linear system solves per time step or iteration} & \textbf{Operation counts}\\\hline
    \textbf{Gradient flow approach} & Solving a nonlinear algebraic system with $2N\times 2N$ Jacobian via an iterative method & $\mathcal{O}(N^{3/2})$ $\times$ number of iterative solves\\\hline
    \multirow{3}{*}{\shortstack{\textbf{Proposed Newton}\\\textbf{scheme}}} & Solving $2N\times 2N$ Newton step problems with backtracking on the double well & $\mathcal{O}(N^{3/2})$ $\times$ number of double well backtracking steps \\\cline{2-3}
    & Solving an $(N+1)\times (N+1)$ linear system for $\delta w_n$ & $\mathcal{O}(N\log N)$ with pre-factorization\\\cline{2-3}
    & Solving an $N\times N$ linear system for $\mu_n$ &  $\mathcal{O}(N\log N)$ with pre-factorization \\\hline
    \end{tabular}
    \caption{A comparison of the computational complexity of the gradient flow approach and the proposed Newton method, measured by the dominant cost of linear system solves at each time step (gradient flow) or Newton iteration (proposed method). The 2D domain discretized with $N$ DoFs, and an optimal complexity sparse directed solver is assumed. The proposed Newton scheme is seen to have the same asymptotic complexity as the gradient flow approach per time step or iteration.}
    \label{tab:complexity}
\end{table}

\begin{remark}
Although the gradient flow approach requires many more iterative solves to reach solutions compared to the proposed scheme, several studies have been done to construct preconditioners that make the nonlinear solve at each time step scalable with respect to the DoFs of the spatial discretization \cite{Farrell2017,Li2018}. It is noteworthy that the $\mathring{H}^{-1}$ Newton step problem of the proposed scheme is much more ill-conditioned than the Newton step problem for the nonlinear solves of the gradient flow approach. Therefore, an appropriate preconditioner is imperative for accelerating the proposed scheme, especially when considering 3D problems that require iterative solves for the linear systems. See Appendix \ref{appendix:c} for a discussion on a preconditioner that accelerates iterative solvers for the $\mathring{H}^{-1}$ Newton step problem in a portion of the parameter space.
\end{remark}

%% file: sec_06_chemical_subsrate_guided_dsa_of_bcp.tex
\section{Chemical Substrate Guided DSA of BCPs}\label{sec:chemical_substrate}
In this section, the proposed scheme is used to study the chemical substrate guided DSA of BCPs. A novel polymer-substrate interaction energy based on the OK model is first introduced. A generalized modified Hessian operator is proposed to accommodate the additional polymer-substrate interaction. Then we present a numerical study based on a physically-reasonable set up.

\input{sec_06/subsec_01_a_polymer-substrate_interation_model}
\input{sec_06/subsec_02_adapting_the_modified_hessian_operator}
\input{sec_06/subsec_03_a_numerical_study_with_the_proposed_scheme}

%% file: sec_06/subsec_01_a_polymer-substrate_interation_model.tex
\subsection{The polymer-substrate interaction energy}
Suppose that a patterned chemical substrate is placed on the bottom boundary ($x_3 = 0$) of a 3d box domain $D$, which contains a diblock copolymer melt consisting of monomers of type A and B. Let $\partial D_s\subset \partial D$ be the region occupied by the chemical substrate. The chemical substrate interacts with both monomers in the diblock copolymer. Let $\eta_K$, $K = A, B$, be the dimensionless parameters that represent the strength of interaction between monomer type $K$ and the substrate. In accordance with \cite{Detcheverry2010}, we define the contribution of the polymer-substrate interaction to the free energy:
\begin{equation}
    F_s(u_A, u_B) = \sum_{K = A, B} \eta_K \exp\Big(-\frac{x_3^2}{2d_s^2}\Big)u_K\vect{1}_{(x_1,x_2,0)\in\partial D_s}\;,
\end{equation}
where $\vect{1}$ is an indicator function and $d_s$ is the decay length scale of the substrate effect in the direction perpendicular to the substrate. It has the following form:
\begin{equation}
    d_s = c_sR_e = c_s\sqrt{N}l\;,
\end{equation}
where $c_s$ is a positive scalar parameter and $R_e$ is the mean-square end-to-end distance of the polymer. Following the relations between the material parameters and the OK model parameters \eqref{eq:model_param}, the decay length scale can be expressed as:
\begin{equation}\label{eq:decay_length}
    d_s^2 = \cfrac{12\sqrt{3}c_s^2\epsilon}{\sqrt{\sigma(1-m^2)}}\;.
\end{equation}

The above energy can be reformulated as a function of the order parameter $u\coloneqq u_A-u_B$. It is further shifted and normalized by the Flory-Huggins parameter $\chi$, in accordance with the derivation in \cite{Choksi2003}, leading to the following form of the polymer-substrate interaction energy:
\begin{equation}
    F_s(u) = 
    \begin{cases}
    \cfrac{\eta_s}{2} \exp\Big(-\cfrac{x_3^2}{2d_s^2}\Big)u \vect{1}_{(x_1,x_2,0)\in\partial \Omega_s}\;, & u\in[0,1]\;;\\
    \infty\;, & \textrm{otherwise}\;,
    \end{cases}
\end{equation}
where $\eta_s\coloneqq (\eta_B-\eta_A)/\chi$ and $\partial \Omega_s$ is the region occupied by the chemical substrate on the normalized domain $\Omega$ with a unit volume. Notice that, if $\eta_s > 0$, then the monomer A is more attractive to the substrate, resulting in a lower energy at $u = 1$, and vice versa. If $|\eta_s|>1$, one of $u=\pm 1$ is no longer a local minimum\footnotemark, hence we only consider $\eta_s\in[-1,1]$. The total interaction energy now has the form:
\footnotetext{If one of $u=\pm 1$ is not a local minimum, then the density multiplication effect \cite{Liu2013, Ruiz2008} can no longer be induced by the chemical substrate.}
\begin{equation}\label{eq:total_interaction_energy_original}
    I(u) = W(u) + F_s(u) = 
    \begin{cases}
    \cfrac{1}{4}(1-u^2)+ \cfrac{\eta_s}{2} \exp\Big(-\cfrac{x_3^2}{2d_s^2}\Big)u\vect{1}_{(x_1,x_2,0)\in\partial \Omega_s}\;, & u\in[-1,1]\;;\\
    \infty\;, & \textrm{otherwise}\;,
    \end{cases}
\end{equation}
\indent We now seek to formulate an approximate form of the total interaction energy that has the following properties: (1) it has the form of a double well, i.e., with local minima at $u=\pm 1$; (2) it maintains the same energy difference between $u=\pm 1$ as the original form; (3) it resembles the shape of the original form in $u\in[-1, 1]$. We thus fit a fourth-order polynomial $P_4(u)$ with the following constraints:
\begin{equation}
    P_4(1) = 0\;,\quad P_4(-1) = 2I(-1) \;,\quad P_4'(1) = 0\;,\quad P_4'(-1) = 0\;,\quad P_4(u_\textrm{max}) = I(u_{\textrm{max}})\;,
\end{equation}
where $u_{\textrm{max}}$ is the local maximizer of $I$. The solution to the polynomial fitting problem reveals that the approximate total interaction energy has the following form:
\begin{equation}\label{eq:approx_total_interaction}
    P_4(u) = \frac{1}{4}(1-u^2)^2 +  \frac{\eta_s}{4}\exp\Big(-\cfrac{x_3^2}{2d_s^2}\Big)(u^3-3u+2)\vect{1}_{(x_1,x_2,0)\in\partial \Omega_s}\;.
\end{equation}
A comparison between the approximate form $P_4$ and the original form $I$ is shown in Figure \ref{fig:approximated_interaction_energy}. The approximate form is simply a sum between the approximate double well potential and a third-order polynomial that represents the contribution of the polymer-substrate interaction. We thus redefine the polymer-substrate interaction energy and the total interaction energy as
\begin{align}
    F_s(u) &= \frac{\eta_s}{4}\exp\Big(-\cfrac{x_3^2}{2d_s^2}\Big)(u^3-3u+2)\vect{1}_{(x_1,x_2,0)\in\partial \Omega_s}\;,\\
    I(u) &= \frac{1}{4}(1-u^2)^2 +  \frac{\eta_s}{4}\exp\Big(-\cfrac{x_3^2}{2d_s^2}\Big)(u^3-3u+2)\vect{1}_{(x_1,x_2,0)\in\partial \Omega_s}\;.
\end{align}
\begin{figure}[!htbp]
    \centering
    \begin{subfigure}{0.32\linewidth}
    \includegraphics[width = \linewidth]{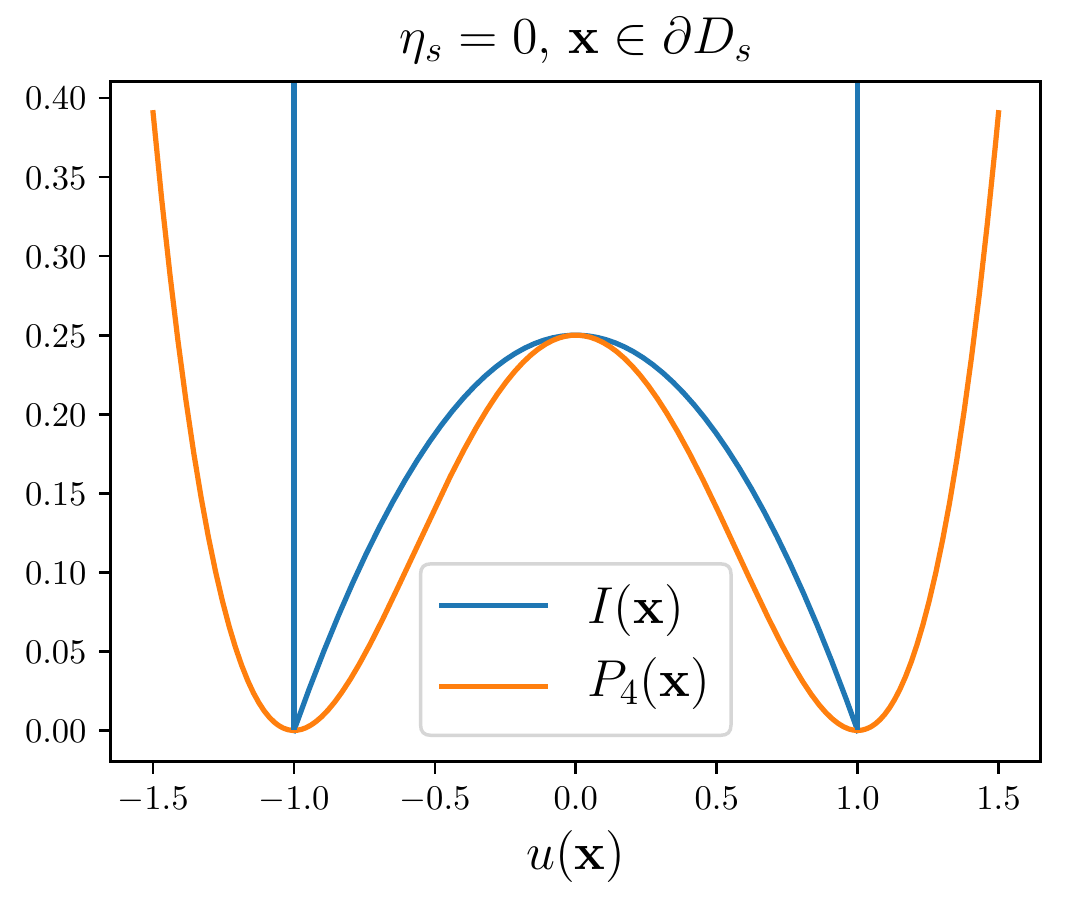}
    \end{subfigure}
    \begin{subfigure}{0.32\linewidth}
    \includegraphics[width = \linewidth]{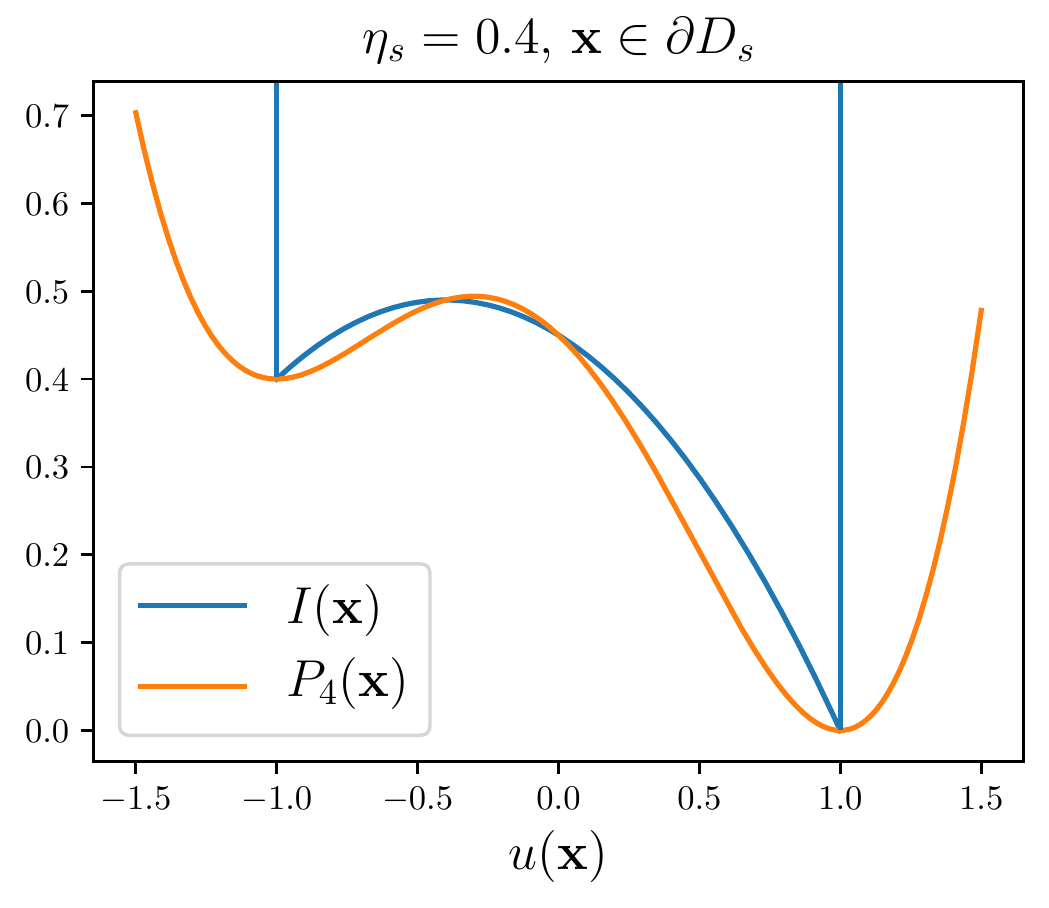}
    \end{subfigure}
    \begin{subfigure}{0.32\linewidth}
    \includegraphics[width = \linewidth]{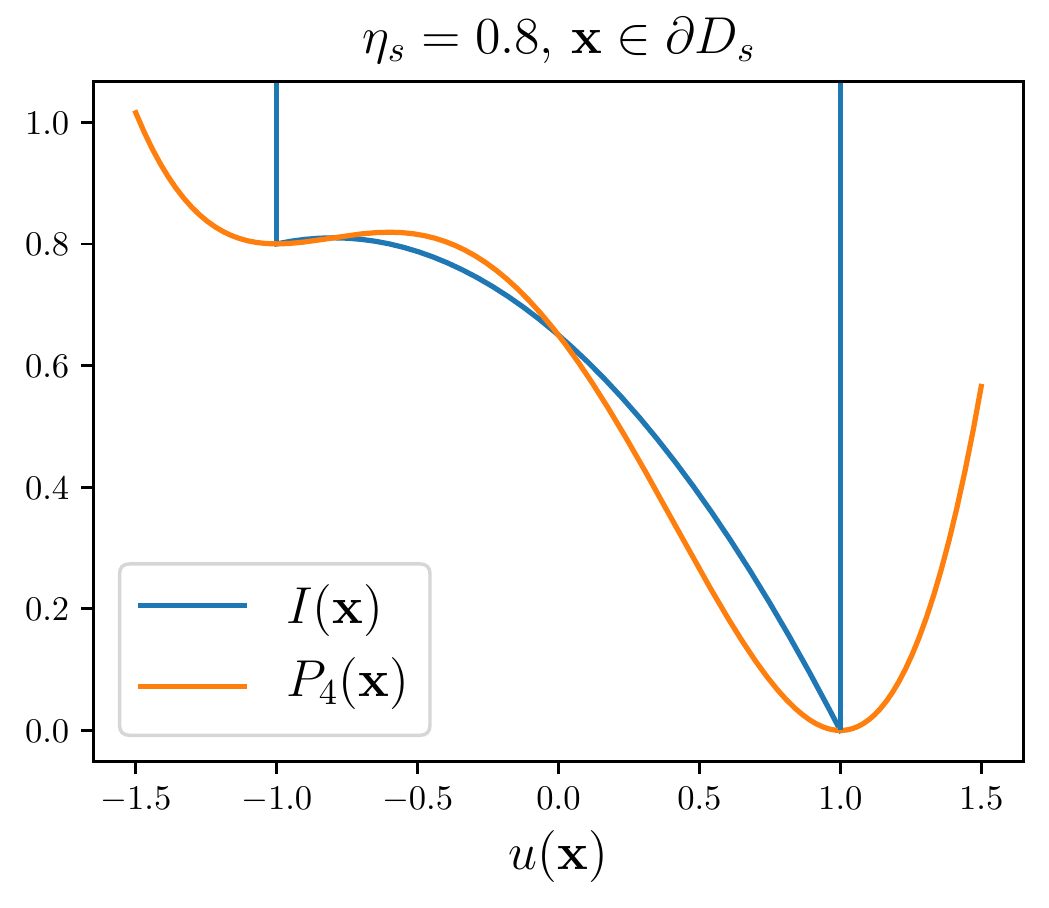}
    \end{subfigure}
    \caption{A comparison between the approximate total interaction energy \eqref{eq:approx_total_interaction} and its original form \eqref{eq:total_interaction_energy_original} at $\eta_s = 0, 0.4, 0.8$ and $\vect{x} \in\partial D_s$.}
    \label{fig:approximated_interaction_energy}
\end{figure}

%% file: sec_06/subsec_02_adapting_the_modified_hessian_operator.tex
\subsection{Generalizing the modified Hessian operator}
To solve the OK model with the additional polymer-substrate interaction energy, we generalize the modified Hessian operator in Section \ref{subsubsec:bt_gamma}. Consider  the following decomposition of $I''(u)$:
\begin{equation}
    \begin{aligned}
    I''(u) &= 3u^2 +\frac{3}{2}\eta_s\exp\Big(-\cfrac{x_3^2}{2d_s^2}\Big)u\vect{1}_{(x_1,x_2,0)\in\partial \Omega_s} - 1\\
    & = \Big(\sqrt{2}u + \frac{3\sqrt{2}}{8}\eta_s\exp\Big(-\cfrac{x_3^2}{2d_s^2}\Big)\vect{1}_{(x_1,x_2,0)\in\partial \Omega_s}\Big)^2\\
    &\qquad + u^2-1-\frac{9}{32}\eta_s^2\exp\Big(-\cfrac{x_3^2}{d_s^2}\Big)\vect{1}_{(x_1,x_2,0)\in\partial \Omega_s}\;.
    \end{aligned}
\end{equation}
A scalar weight $\gamma\in[0,1]$ is introduced, similar to \eqref{eq:weighted_dw}:
\begin{equation}
    \begin{aligned}
    I''_\gamma(u) &= \Big(\sqrt{2}u + \frac{3\sqrt{2}}{8}\eta_s\exp\Big(-\cfrac{x_3^2}{2d_s^2}\Big)\vect{1}_{(x_1,x_2,0)\in\partial \Omega_s}\Big)^2\\
    &\qquad + \gamma\Big(u^2-1 - \frac{9}{32}\eta_s^2\exp\Big(-\cfrac{x_3^2}{d_s^2}\Big)\vect{1}_{(x_1,x_2,0)\in\partial \Omega_s}\Big)\;.
    \end{aligned}
\end{equation}
Notice that, when $\eta_s = 0$ or $(x_1,x_2,0) \not \in\partial\Omega_s$, we recover the results in Section \ref{subsubsec:bt_gamma}. The second derivative above is positive when $\gamma = 0$ and it is straightforward to see that Lemma \ref{prop:coercivity} still applies for the modified Hessian operator defined with $I''_{\gamma}(u)$. Consequently, the double well backtracking algorithm for generating descending Newton steps, Algorithm \ref{alg:newton_step}, still holds for the generalized modified Hessian operator and we retain the global convergence of the proposed algorithm.

%% file: sec_06/subsec_03_a_numerical_study_with_the_proposed_scheme.tex
\subsection{A numerical study with the proposed scheme}
We provide an application of the proposed scheme in the context of the chemical substrate guided DSA of BCPs. Suppose we have a physical domain $D=[0, 400]\textrm{ nm}\times[0, 250]\textrm{ nm}\times[0,30]\textrm{ nm}$, which contains the melt of a symmetric diblock copolymer. We take the material parameters for the polymer to be
\begin{equation*}
    f = 0.5\;,\quad\chi = 0.1\;,\quad N = 900,\quad l = 0.65 \textrm{ nm}\;.
\end{equation*}
\indent Let the computational domain be $\Omega = [0,40]\times[0, 25]\times [0, 3]$. The OK parameters computed by the relations \eqref{eq:model_param} are $(m,\kappa,\epsilon,\sigma) = (0, 1, 0.237, 1.68)$. We first compute two equilibrium morphologies without the chemical substrate through the proposed scheme. The equilibrium morphologies as well as the evolution of the residual norm and the energy are shown in Table \ref{tab:states_w/o_substrate}. The equilibrium morphologies have dominant features of strips perpendicular to the $x_3$ axis, with varying orientations in the $x_1$-$x_2$ plane. Minor non-uniformity is seen in the $x_3$ direction, mainly located near the boundary.

\begin{table}[!htbp]
\centering
    \begin{tabular}{|c|c|c|c|}\hline
    \multicolumn{2}{|c|}{\textbf{Model solutions}} &  \textbf{Energy evolution} & \textbf{Residual evolution}\\ \hline
     \multicolumn{2}{|c|}{\includegraphics[width = 0.25\linewidth]{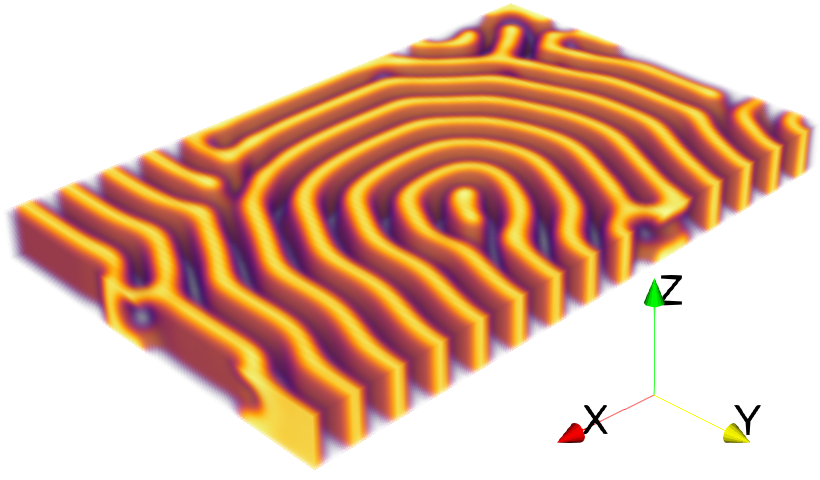}\includegraphics[width = 0.25\linewidth]{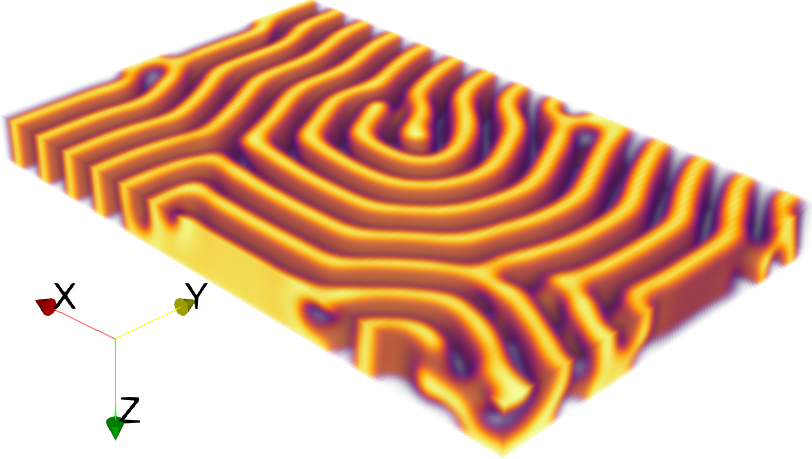}} & \includegraphics[width = 0.21\linewidth]{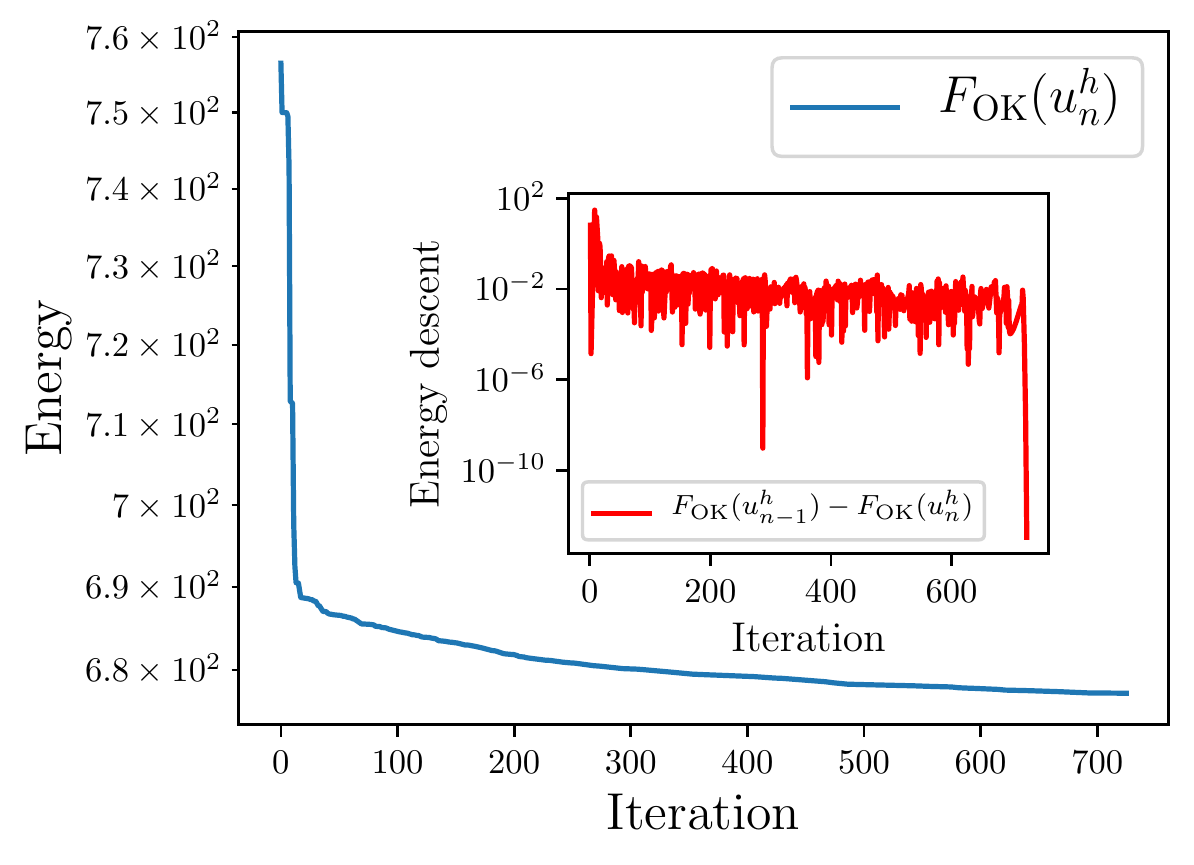}& \includegraphics[width = 0.21\linewidth]{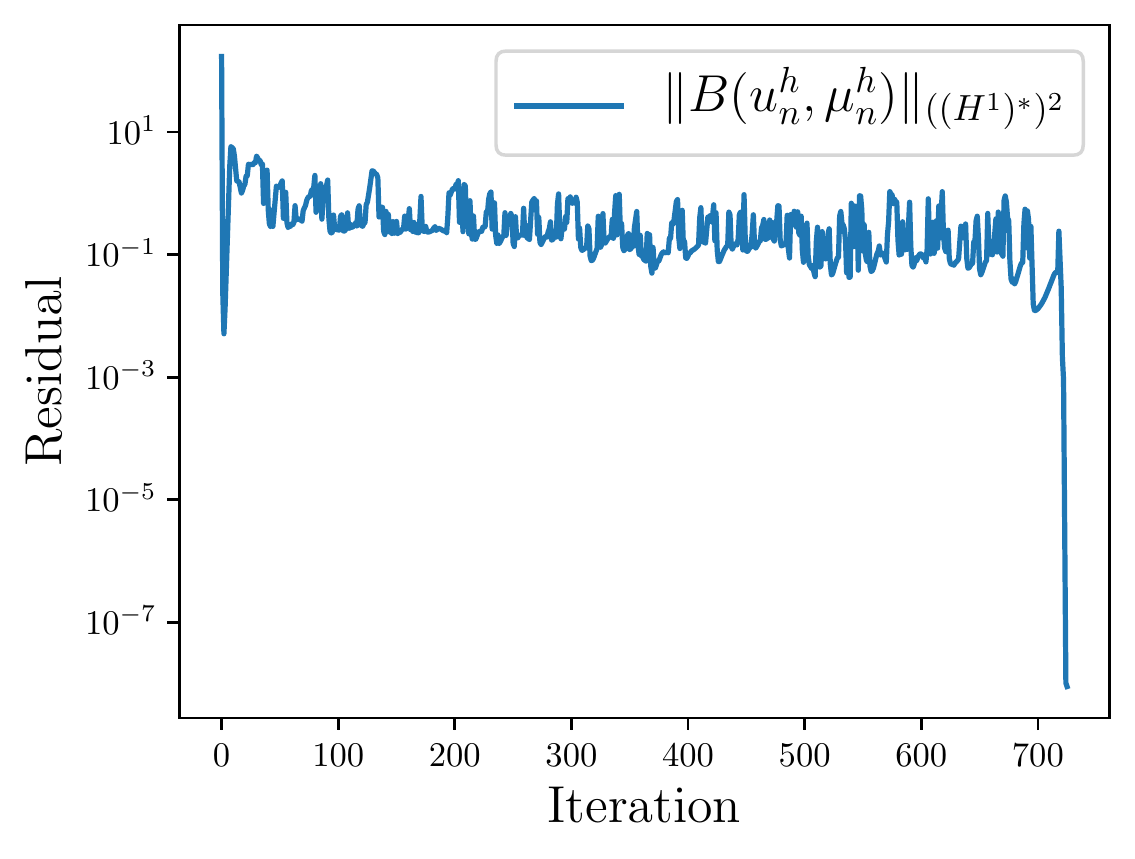}\\\hline
    \multicolumn{2}{|c|}{\includegraphics[width = 0.25\linewidth]{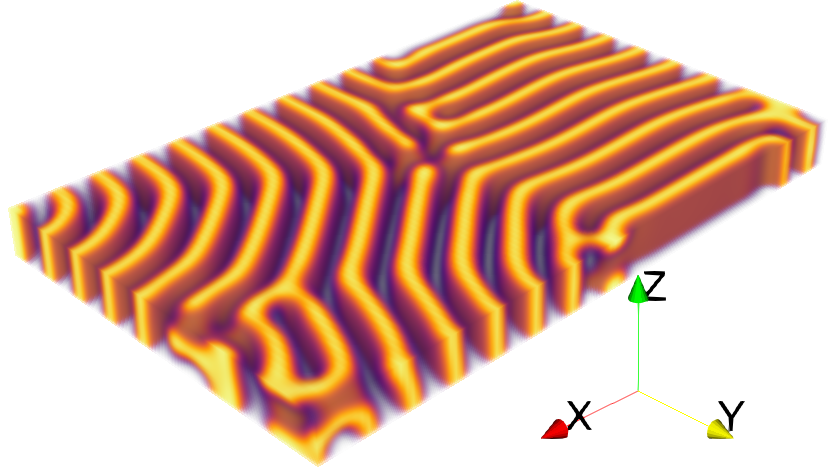}\includegraphics[width = 0.25\linewidth]{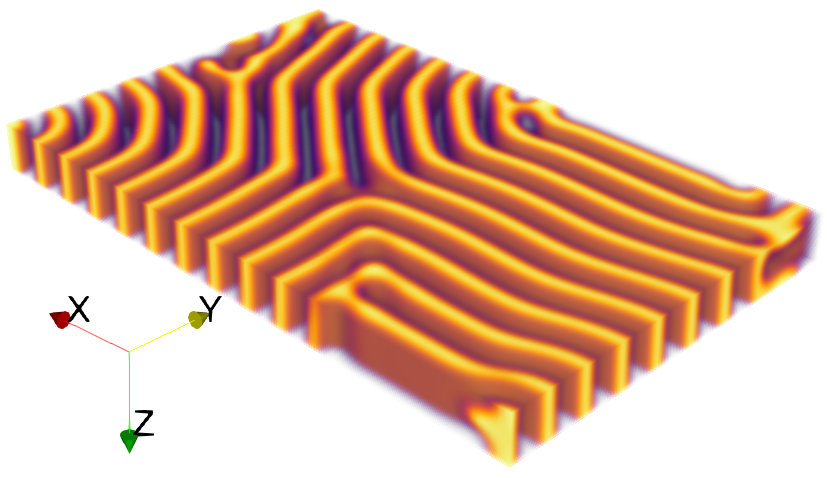}} & \includegraphics[width = 0.21\linewidth]{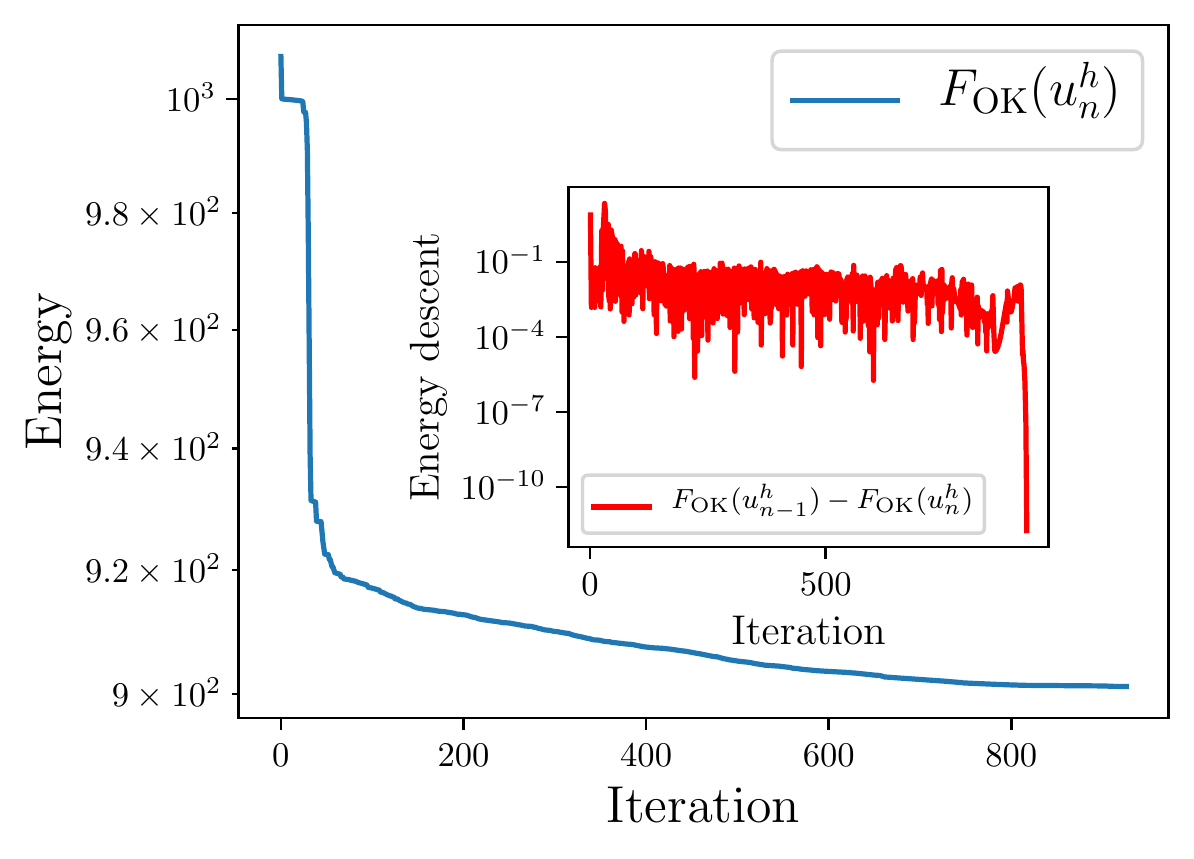}& \includegraphics[width = 0.21\linewidth]{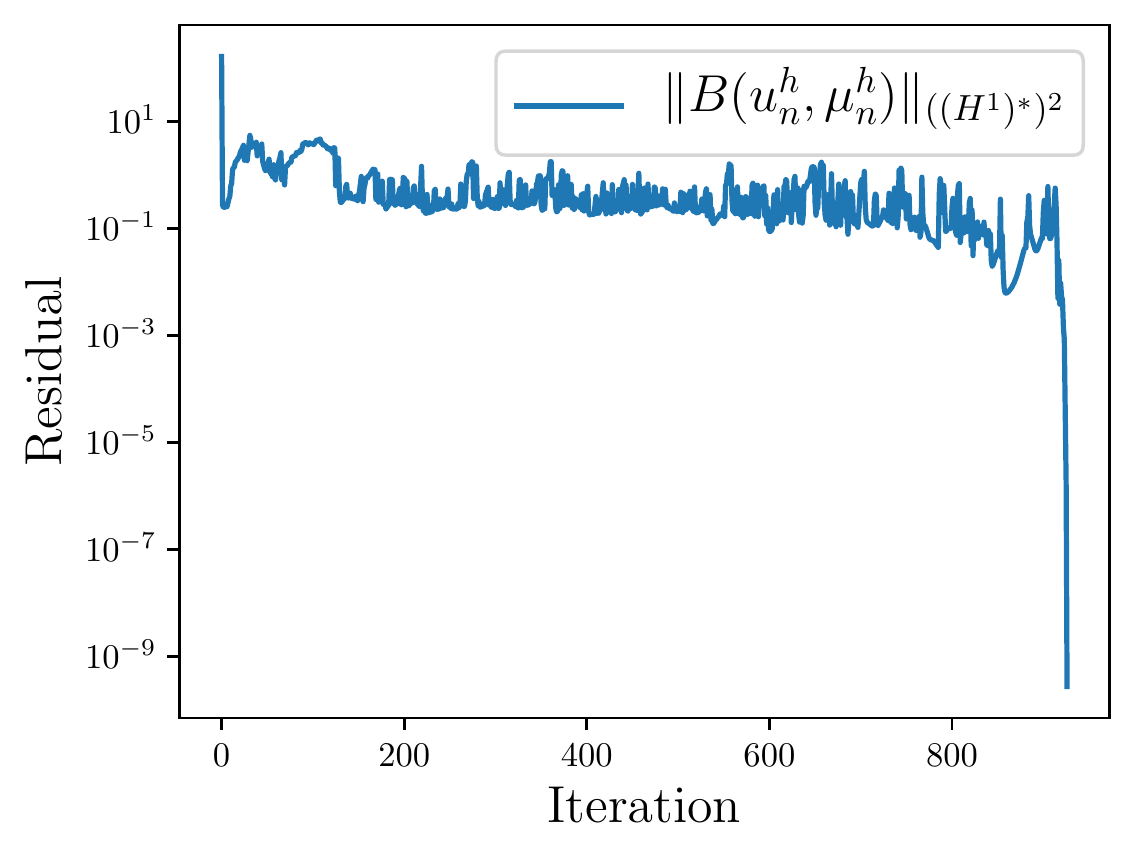} \\\hline
    \end{tabular}
    \caption{The top and bottom view of two OK model solutions obtained through the proposed scheme without the chemical substrate. The domain is discretized with uniform triangular elements with cells of $160\times100\times 12$. Each minimization sequence starts at a random initial guess generated with parameters $(s,\delta,\gamma) = (0.05, 80, 0.02)$. The double well backtracking sequence is set to $\vect{\gamma} = \{1, 0.5, 0\}$. The sequence is terminated when the residual norm is lower than $10^{-7}$. The Hessian-modified $\mathring{H}^{-1}$ Newton step problems at $\gamma = 0.5, 0$ are solved iteratively with the GMRES method and the preconditioner introduced in Appendix \ref{appendix:c}.}
    \label{tab:states_w/o_substrate}
\end{table}
Now consider the scenario where a chemical substrate is placed on a subset of the bottom boundary of the physical domain, $\partial D_s\subset \partial D$. The chemical substrate pattern is periodic strips with width of $62$ nm and periodicity distance of $100$ nm:
\begin{equation}
    \partial D_s = \bigcup_{k=0}^3 [100k+19, 100k+81] \textrm{ nm} \times [0, 250] \textrm{ nm} \times\{0\} \textrm{ nm}\;.
\end{equation}
Its corresponding position on the computational domain is:
\begin{equation}
    \partial \Omega_s = \bigcup_{k=0}^3 [10k+1.9, 10k+8.1] \times [0, 25] \times\{0\}\;.
\end{equation}
\begin{figure}[!htbp]
    \centering
    \begin{subfigure}{0.49\linewidth}
    \centering
    \includegraphics[width=0.9\linewidth]{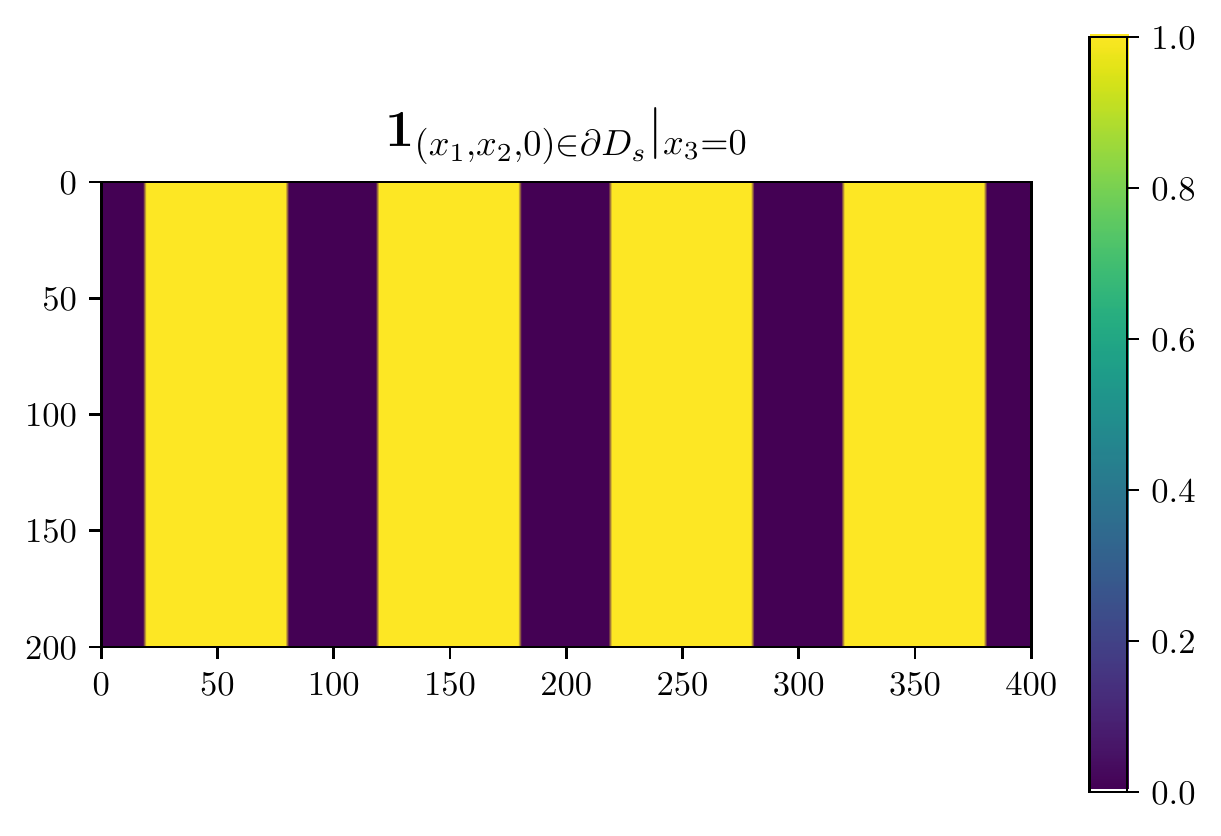}
    \end{subfigure}
    \begin{subfigure}{0.49\linewidth}
    \centering
    \includegraphics[width = 0.8\linewidth]{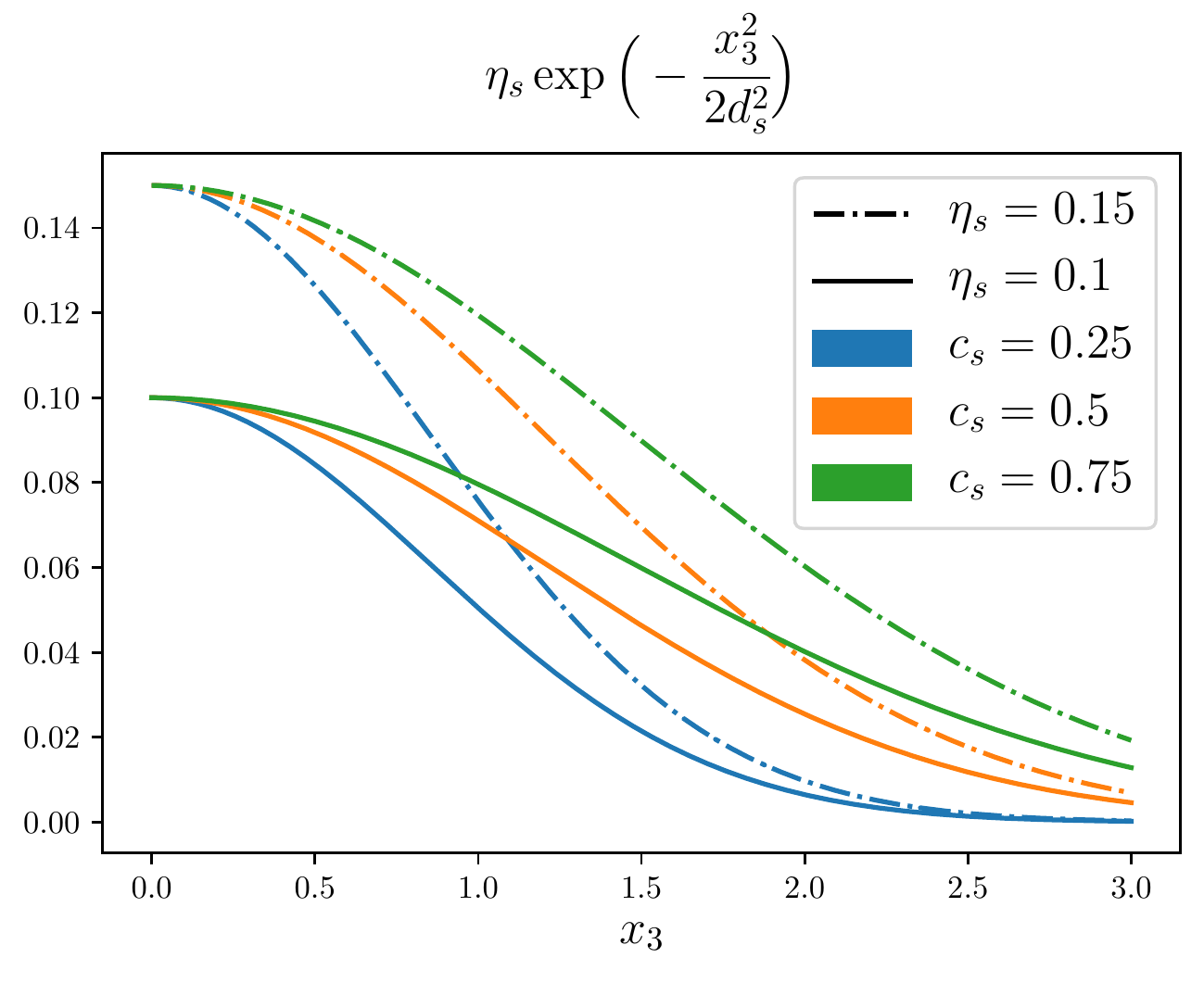}
    \end{subfigure}
    \caption{The chemical substrate pattern placed on the bottom boundary of the physical domain and the decay profile of the polymer-substrate interaction energy in the $x_3$ direction.}
    \label{fig:substrate}
\end{figure}
We consider two different polymer-substrate interaction parameters, $\eta_s=0.1, 0.15$, and three different decay length scales, $c_s = 0.25, 0.5, 0.75$. The decay length scale of the polymer-substrate interaction, $d_s = 2.925c_s$, is computed through \eqref{eq:decay_length}.\\
\indent The OK energy plus the additional polymer-substrate interaction energy is minimized through the proposed scheme to obtain the equilibrium morphologies for the chemical substrate guided DSA. The equilibrium states and the evolution of the residual norm are shown in Table \ref{tab:states_w_substrate}. When the decay length scale is small, the equilibrium morphologies exhibit phase separation parallel to the substrate. The monomer $A$ phase ($u=1$) spans the space above the chemical substrate, indicating that the polymer-substrate interaction dominants the phase separation. As the decay length scale increases, disconnected strips, perpendicular and aligned with the chemical substrate, start to emerge. The size of the feature is similar to that of the equilibrium morphologies without the chemical substrate. At $c_s = 0.5$, there are still large regions of the monomer A phase on the bottom boundary where the chemical substrate is placed, suggesting that the polymer-substrate interaction is the main cause of the non-uniformity in the $x_3$ direction. For $c_s = 0.75$, aligned parallel strips are formed at $\eta_s = 0.15$, while defects are formed among the parallel strips at $\eta_s = 0.1$. The slow convergence of the minimization iterations at $\eta_s = 0.1$ suggests that the energy landscape is flat and the swift convergence of the minimization iterations at $\eta_s = 0.15$ suggests that the lamelle morphology is a dominant minimizer. The numerical study shows that the decay profile of the polymer-substrate interaction is crucial to the formation of the lamella morphology beyond the regime of the thin-film approximation and the polymer-substrate interaction parameter $\eta_s$ is important to the elimination of defects in the chemical substrate guided DSA of BCPs.

\begin{table}[!htbp]
    \centering
    \begin{tabular}{|c|c|c|}\hline
    &$\eta_s = 0.1$ & $\eta_s = 0.15$\\\hline
    \multirow{2}{*}{$c_s = 0.25$} &\includegraphics[width = 0.2\linewidth]{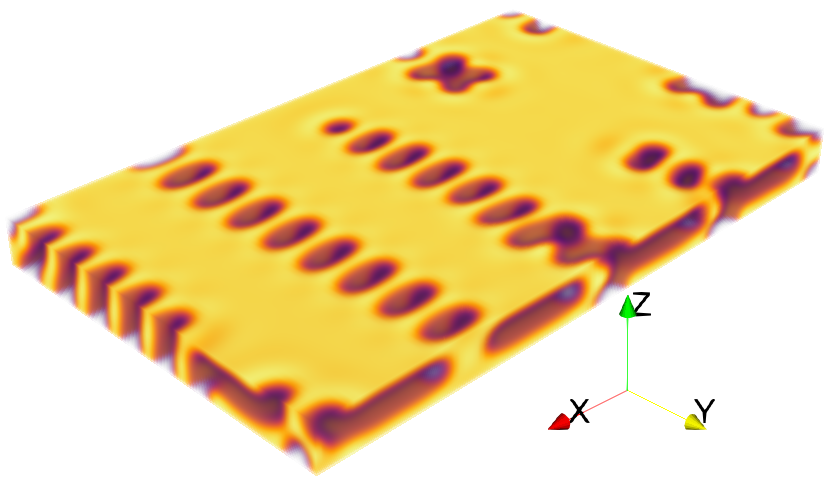}
    \includegraphics[width = 0.2\linewidth]{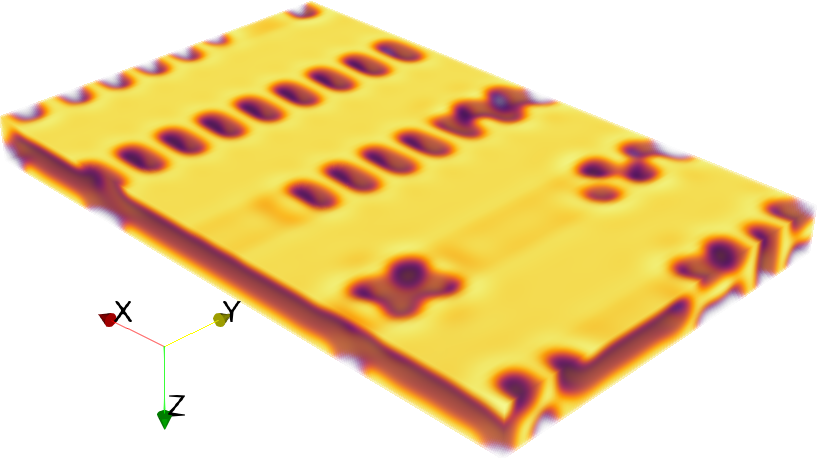} &     \includegraphics[width = 0.2\linewidth]{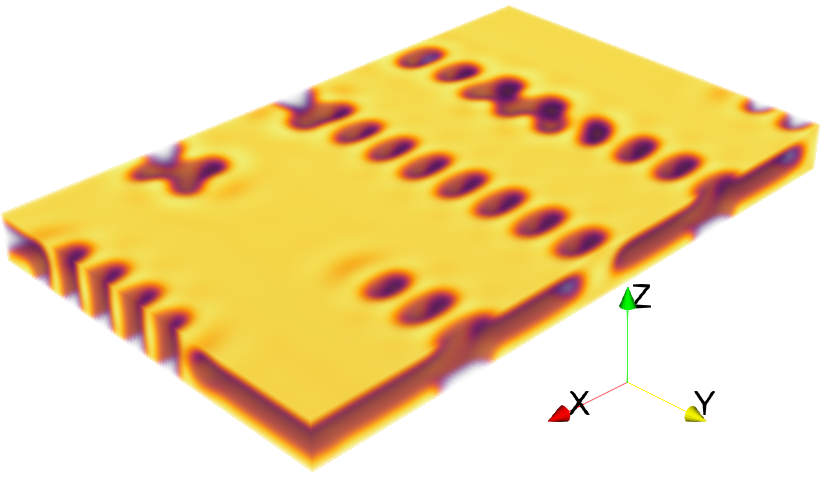}
    \includegraphics[width = 0.2\linewidth]{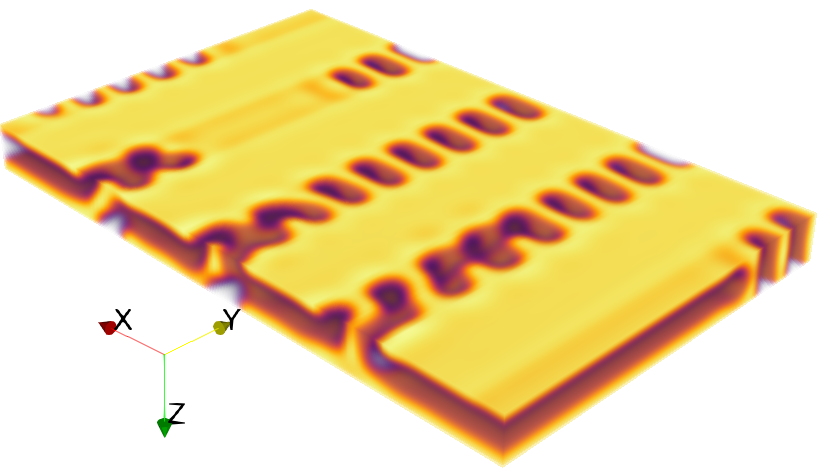}\\ 
    &\includegraphics[width = 0.2\linewidth]{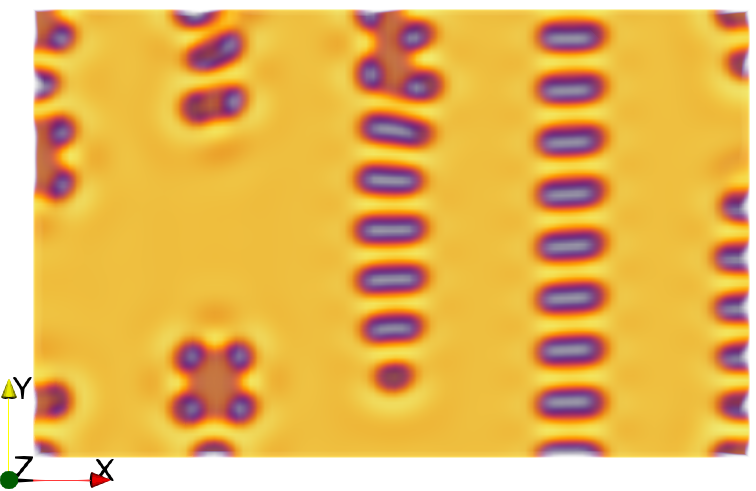} \includegraphics[width = 0.2\linewidth]{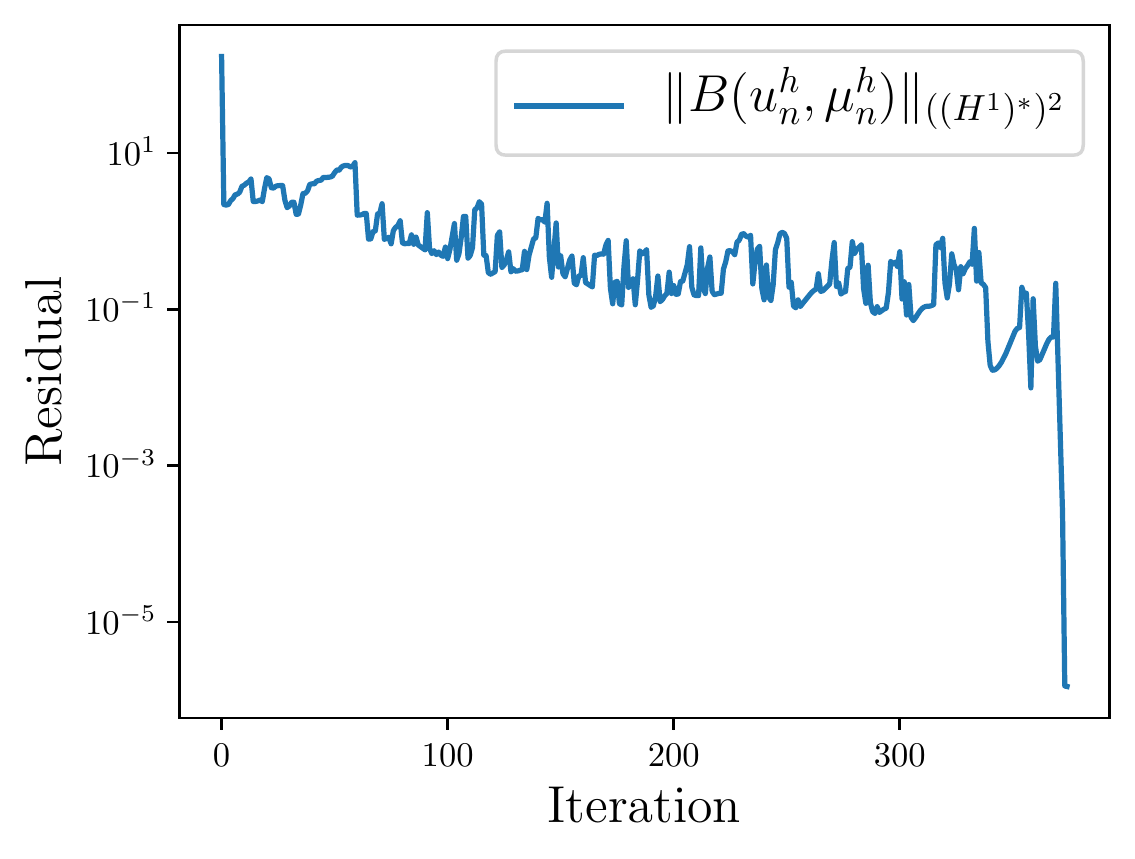} &
    \includegraphics[width = 0.2\linewidth]{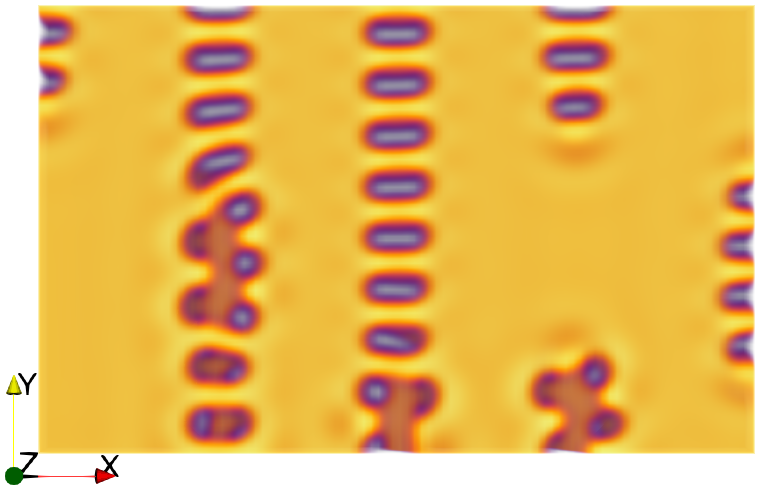}\includegraphics[width = 0.2\linewidth]{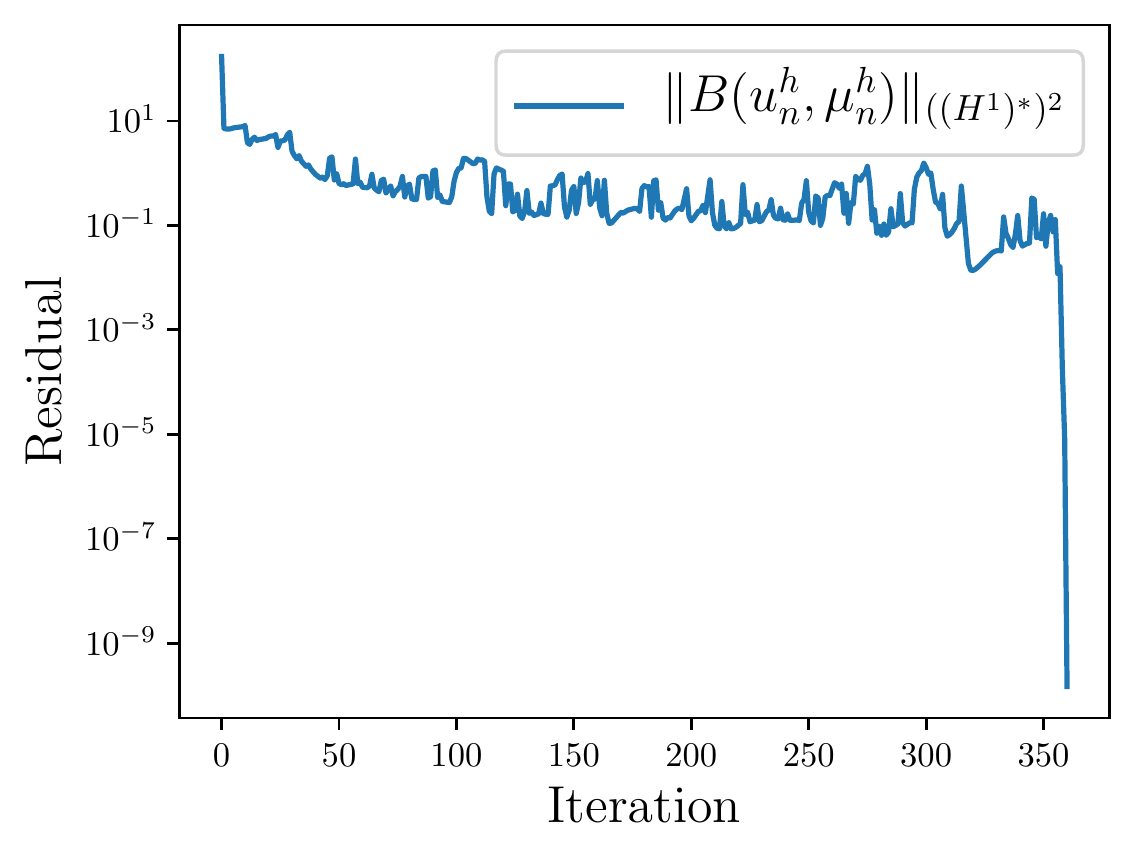}\\\hline
    \multirow{2}{*}{$c_s = 0.5$} &\includegraphics[width = 0.2\linewidth]{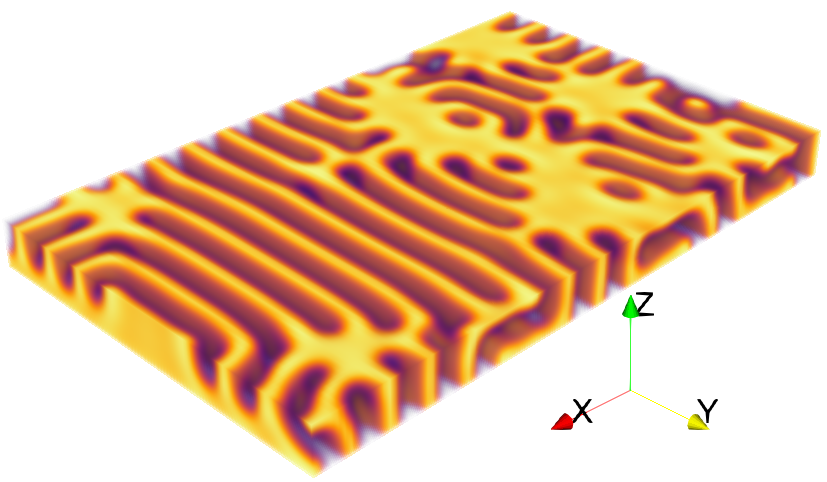}
    \includegraphics[width = 0.2\linewidth]{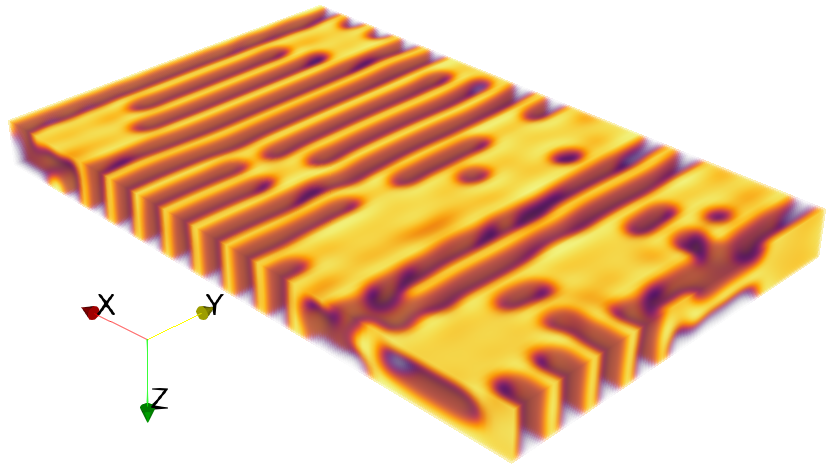} &     \includegraphics[width = 0.2\linewidth]{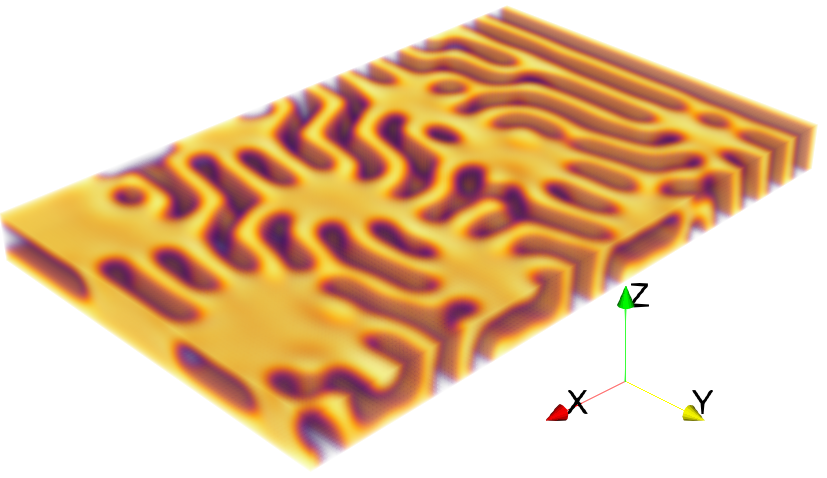}
    \includegraphics[width = 0.2\linewidth]{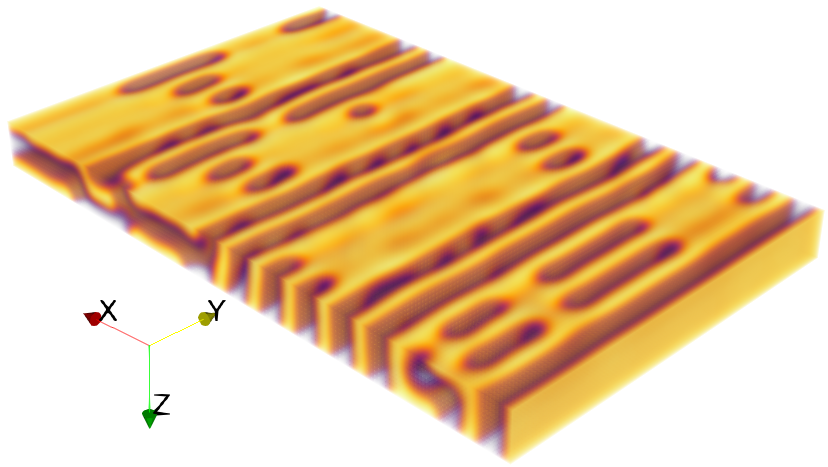}\\ 
    &\includegraphics[width = 0.2\linewidth]{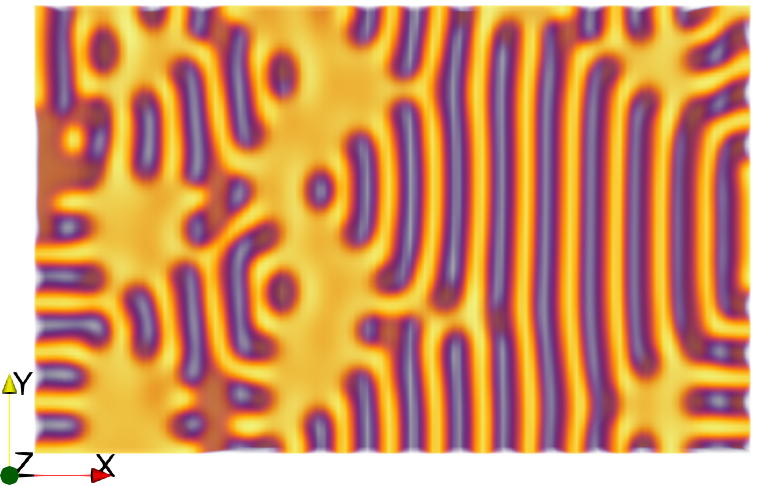}\includegraphics[width = 0.2\linewidth]{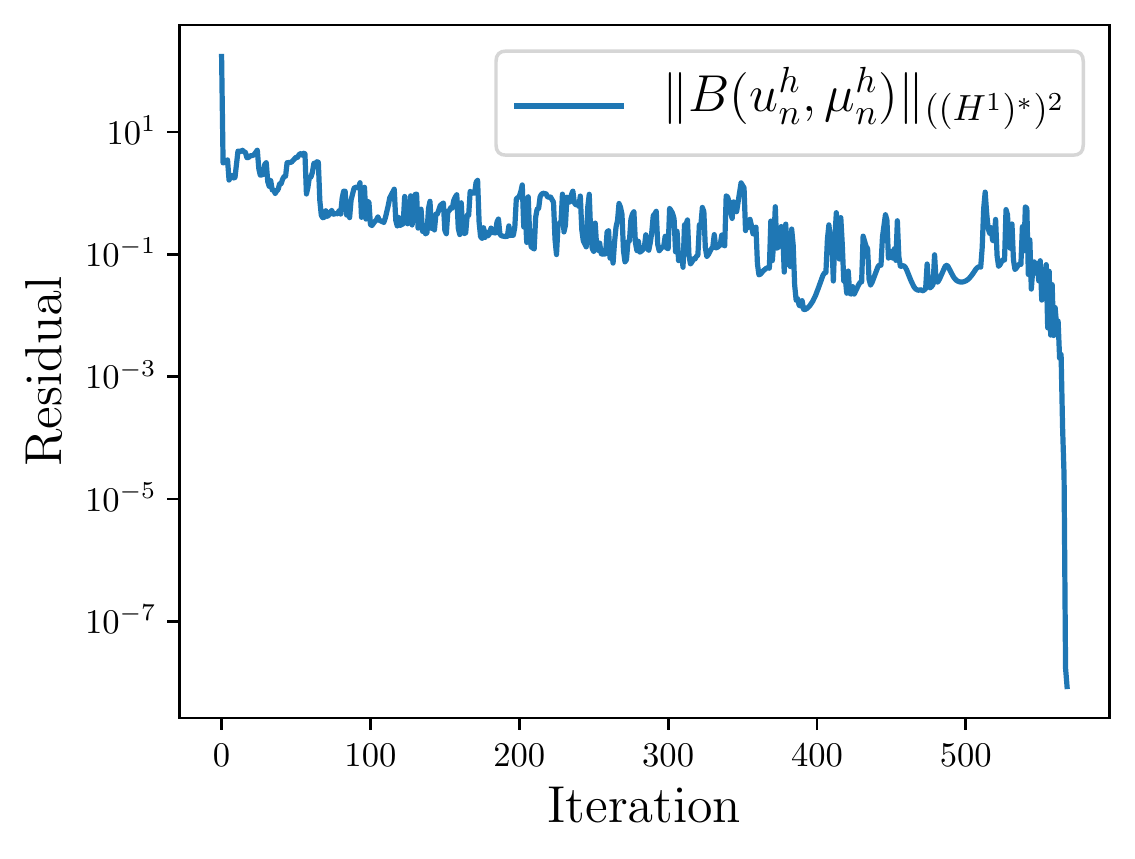} &
    \includegraphics[width = 0.2\linewidth]{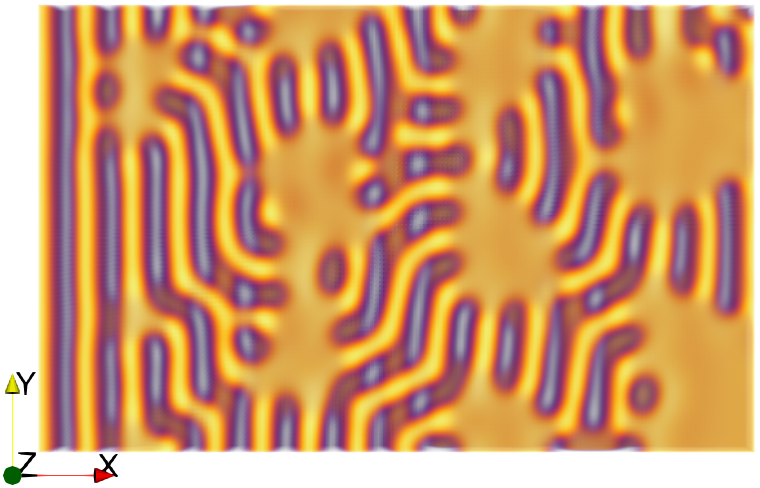}\includegraphics[width = 0.2\linewidth]{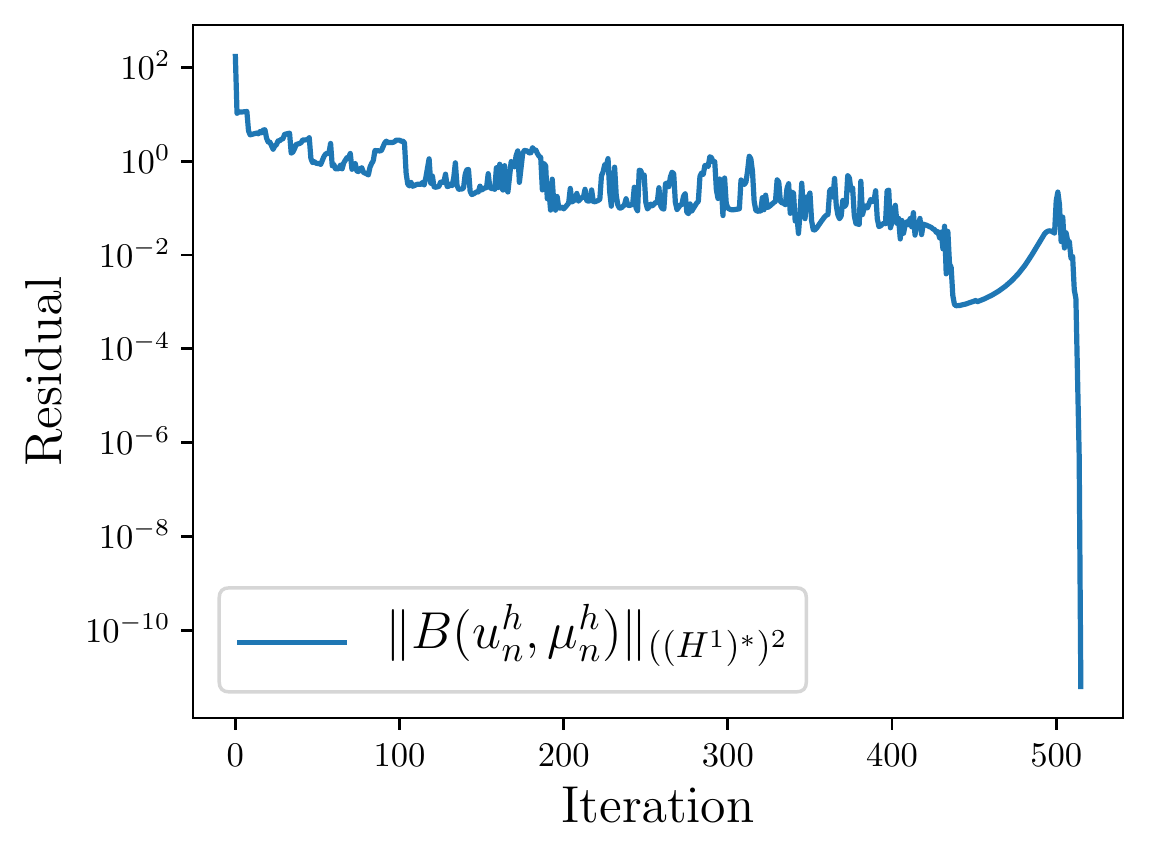}\\\hline
    \multirow{2}{*}{$c_s = 0.75$} &\includegraphics[width = 0.2\linewidth]{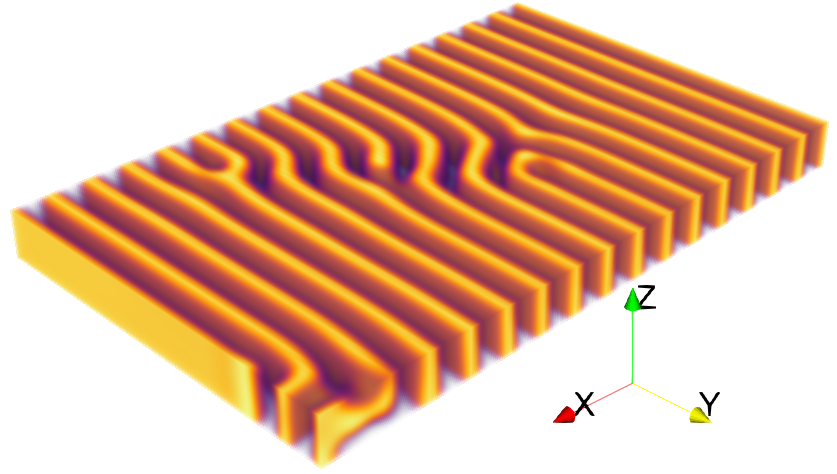}
    \includegraphics[width = 0.2\linewidth]{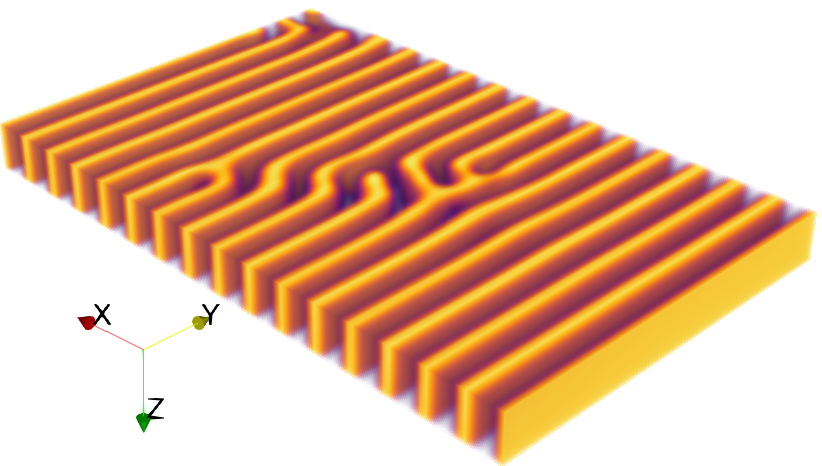} &     \includegraphics[width = 0.2\linewidth]{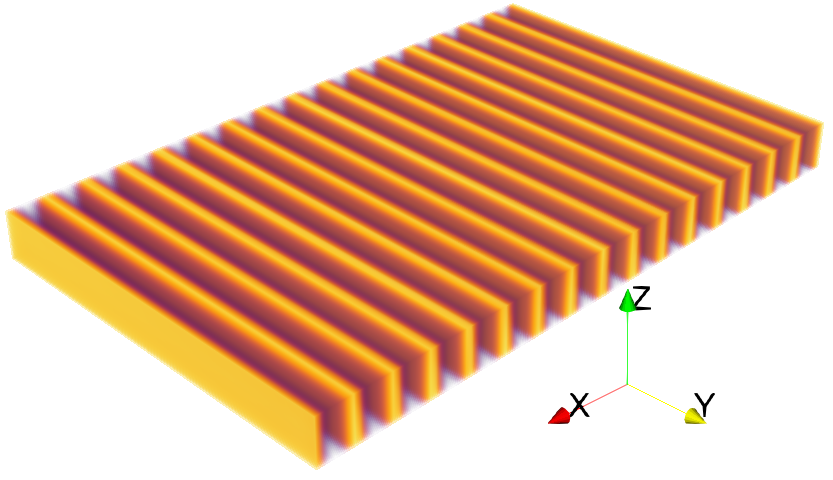}
    \includegraphics[width = 0.2\linewidth]{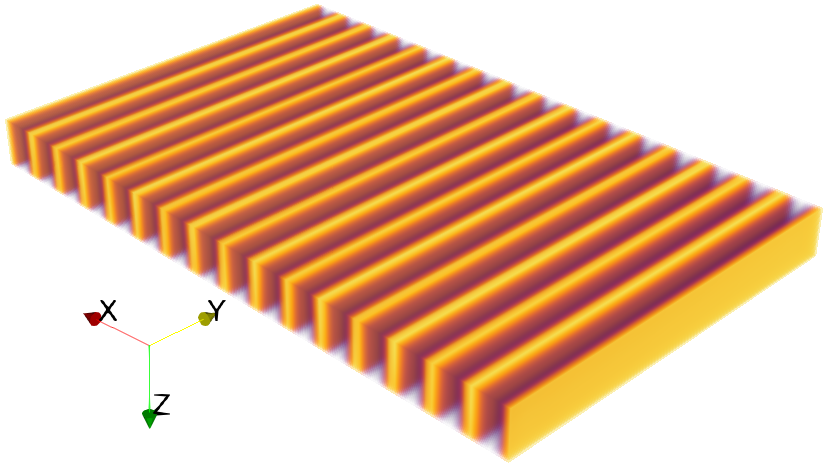}\\ 
    &\includegraphics[width = 0.2\linewidth]{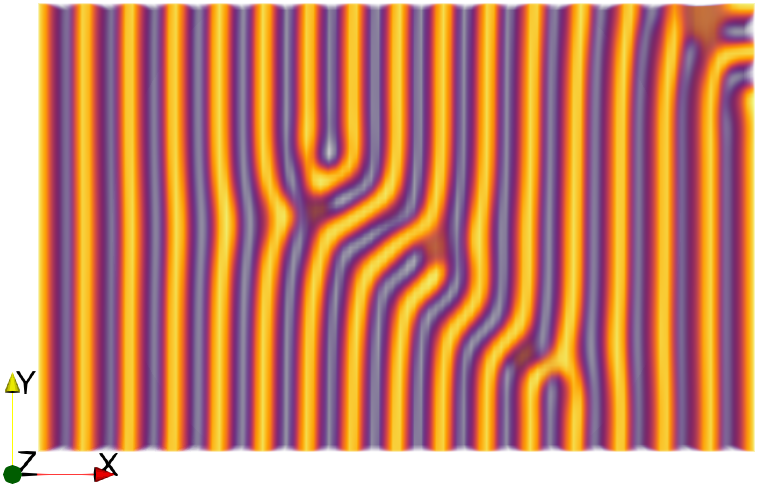}\includegraphics[width = 0.2\linewidth]{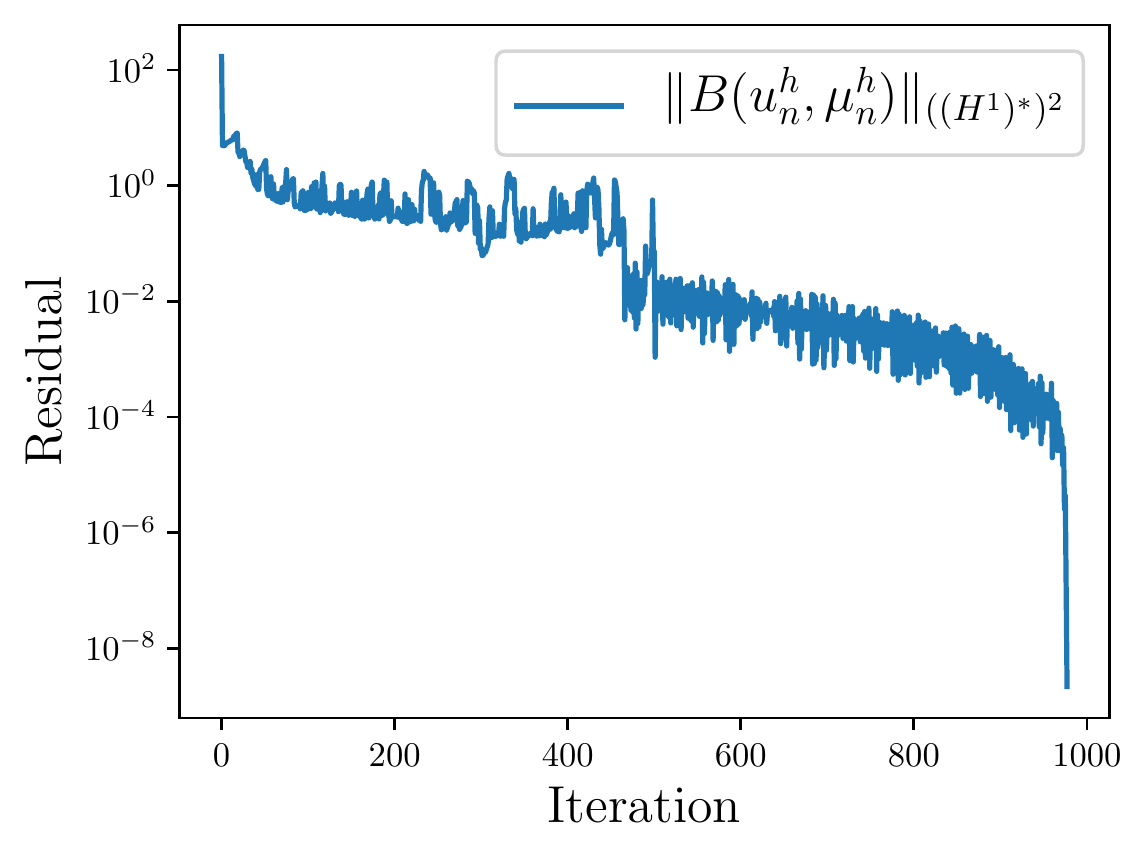} &
    \includegraphics[width = 0.2\linewidth]{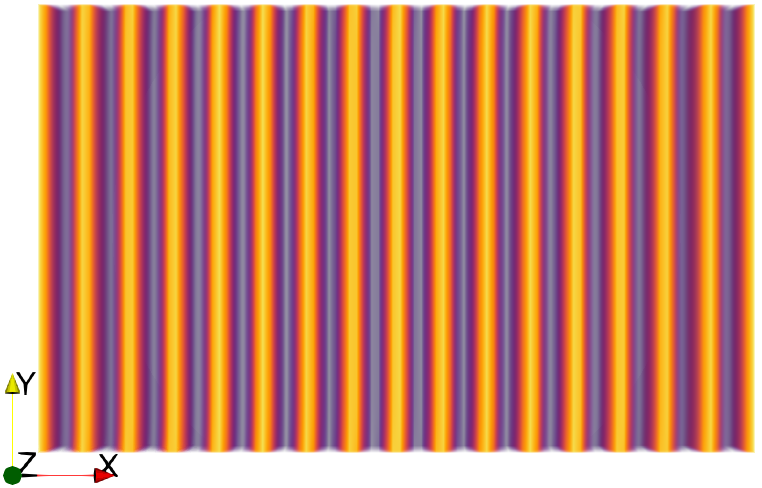}\includegraphics[width = 0.2\linewidth]{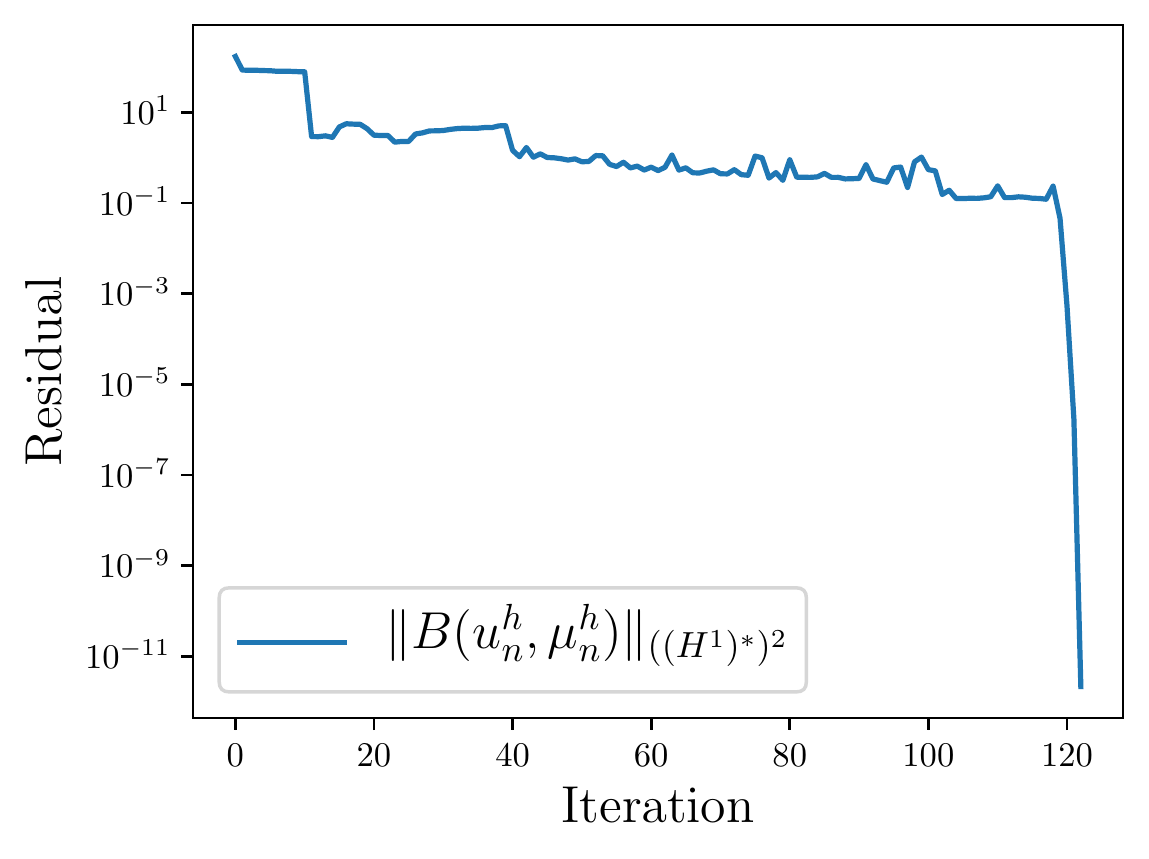}\\\hline
    \end{tabular}
    \caption{The top, bottom, and top-down view of the equilibrium morphologies as well as the evolution of the residual norm obtained through the proposed scheme with the chemical substrate pattern of periodic strips, varying polymer-substrate interaction parameters and decay length scales. The numerical implementation is performed as described in Table \ref{tab:states_w/o_substrate}.}
    \label{tab:states_w_substrate}
\end{table}

%% file: sec_07_conclusion.tex
\section{Conclusion}
We proposed a fast and robust scheme for the direct minimization of the OK energy. The minimization problem is posed in the $\mathring{H}^1$ space. We studied the regularity of the critical points of the OK energy and transformed the first order condition via the $\mathring{H}^{-1}$ inner product such that the solution space and the test space are both in $H^1$. We introduced a globally convergent modified Newton method for the energy-descending and mass-conservative minimization of the OK energy. The monotonic energy descent is guaranteed by backtracking on a Gauss-Newton relaxation of the nonconvexity of the double well potential. Gaussian random fields are used as initial guesses for the minimization iterations. The properties of the proposed scheme and its three order of magnitude superior efficiency compared to the gradient flow approach are demonstrated through numerical examples. We then apply the proposed scheme in the context of chemical substrate guided DSA of BCPs. A novel polymer-substrate interaction model, in the form of a third-order polynomial, is derived. We conducted a numerical study, via the proposed scheme, of the chemical substrate guided DSA of BCPs based on a physically-reasonable set up. We showed, via a numerical example, that the decay profile of the polymer-substrate interaction and the polymer-substrate interaction strength is crucial to the formation of defect-free morphologies beyond the regime of the thin-film approximation.

%% file: acknowledgement.tex
\section*{Acknowledgement}
The work of the authors is supported by \href{https://aeolus.oden.utexas.edu}{the AEOLUS center}, funded by the U.S. Department of Energy, Office of Science, Office of Advanced Scientific Computing Research under Award DE‐SC0019303. The authors would like to acknowledge Joshua Chen, Prashant Jha, and Peng Chen for their suggestions on this work.

%% file: appendix.tex
\begin{appendices}

\input{appendix/sec_a_the_regularity_of_the_critical_points_of_f_ok.tex}
\input{appendix/sec_b_finite_element_discretization}

\input{appendix/sec_c_preconditioner}
\end{appendices}

%% file: appendix/sec_a_the_regularity_of_the_critical_points_of_f_ok.tex
\section{The regularity of the critical points of \texorpdfstring{$F_{OK}$}{}}\label{appendix:a}
Please refer to \cite{Adams2003, Folland1995, Gilbarg1983} for known theorems with regard to elliptic PDEs and Sobolev spaces used in the derivation below. Pieces of the proof below overlap with the work in \cite{Cowan2005} in the context of the Cahn-Hilliard equation.
\begin{proof}[Proof of Proposition \ref{prop:regularity}]
    Let $s_2 = (\kappa W'(u_{c}), 1)$ and define $w\coloneqq(-\lap_N)^{-1}R(u_{c}-m)$. The first order condition \eqref{eq:1st_var} implies
    \begin{equation*}
        \big(\epsilon^2\grad (u_c-m), \grad \tilde{u}\big) = -(\kappa W'(u_c)-s_2, \tilde{u}) - (\sigma w,\tilde{u})\quad\forall \tilde{u}\in H^1\;,
    \end{equation*}
    where $\tilde{u}_0\in \mathring{H}^1$ is extended to $\tilde{u}\in H^1$. Consequently, $\epsilon^2 (u_c-m)= (-\lap_N)^{-1}R(f)$ with $f = -\kappa W'(u_c)+s_1-\sigma w$ and $W'(u_c) = u_c^3-u_c $.\\
    \indent Now we apply the elliptic regularity theorem on $u_c$, since it is a solution to the above weak form of the pure Neumann Poisson equation. First, notice that Sobolev embedding theorem implies  $u_c\in H^1 \Rightarrow u_c^3\in L^2$. $H^k$ is an algebra for $k>d/2$ in a domain that satisfies the cone condition. Thus,
    \begin{equation*}
        u_c\in H^k \Rightarrow u_c^3\in H^k\;\forall k\geq 2\;.
    \end{equation*}
    Define $H^0 \coloneqq L^2$. For $r\geq 1$, if $u_c\in\mathring{H}_m^{k-1}$ with $1\leq k\leq r$, the elliptic regularity implies that $w\in\mathring{H}^{k+1}$. Combining the results above, we have $f\in\mathring{H}^{k-1}_m$ and $u_c\in\mathring{H}^{k+1}$ by regularity. Take $k=r$ and we have $u_c\in\mathring{H}^{r+1}_m$.
\end{proof}

\begin{proof}[Proof of Proposition \ref{prop:bc}]
    Based on the results in the proof of Proposition \ref{prop:regularity}, if $u_c\in \mathring{H}^2_m$, then $\grad u_c \cdot \nu = 0$ on $\partial\Omega$ is implied by Green's identity. Moreover, if $u_c\in\mathring{H}^4_m$, we have $A(u_c)\equiv 0$ by Theorem \ref{thm:h-1_1st}. Notice that we have $(\epsilon^2\grad(\lap u_c)\cdot\nu, \tilde{u})_{L^2(\partial\Omega)} + (\kappa W''(u_c)\grad u_c\cdot\nu, \tilde{u})_{L^2(\partial\Omega)} = 0\;\forall\tilde{u}\in H^1$, which leads to $\grad(\lap u_c)\cdot\nu = 0$ on $\partial\Omega$.
\end{proof}

%% file: appendix/sec_b_finite_element_discretization.tex
\section{Finite element discretization}\label{appendix:b}
We shall consider a finite-dimensional approximation to the minimization problem and the $\mathring{H}^{-1}$ Newton iteration via the finite element method.

\indent First, we introduce a conforming Galerkin discretization employing Ciarlet-Raviart mixed finite elements. Let $\mathcal{P}^h$ denote a conformal partition $\mathcal{P}^h$ of $\Omega$ into non-overlapping elements $K$ such that $\Omega = \cup_{K\in\mathcal{P}^h} \bar{K}$. Each element is the image of an invertible, generally affine map $F_K$ of a reference element $\hat{K}$. Let $P^r(\hat{K})$ denote the space of complete polynomials of degree less than or equal to $r$ on $\hat{K}$, and let $Q^r(\hat{K})$ denote the space of tensor products of polynomials of degree less than or equal to $r$ on $\hat{K}$. We construct the finite-dimensional spaces $\mathcal{S}^h$ by:
\begin{equation}
    \mathcal{S}^h\coloneqq S^{h,r}(\mathcal{P}^h)\coloneqq \Big\{v^h\in H^1(\Omega):v^h|_K=\hat{v}\circ F_K^{-1}, \hat{v}\in P^r(\hat{K}) \text{ or } \hat{v}\in Q^r(\hat{K}) \;\forall K\in\mathcal{P}^h\Big\}\;.
\end{equation}

\subsection{The Newton step problem}
\indent We employ linear finite elements, i.e., the lowest-order approximation based on $S^{h,1}\coloneqq S^{h,1}(\mathcal{P}^h)$, for the Newton step problem at given $u_n^h\in S^{h,1}$:
\begin{subequations}
\begin{align}
    &\text{Find } (\delta u_n^h, \hat{\mu}_n^h)\in(S^{h,1}, S^{h,1}) \text{ such that}:\nonumber\\
    &\qquad\qquad(\grad \hat{\mu}^h, \grad\tilde{u}) + (\sigma\delta u_n^h, \tilde{u})= (\grad \mu^h_n, \grad\tilde{u}) + \big(\sigma(u_n^h-m), \tilde{u}\big)\qquad &&\forall \tilde{u}\in S^{h,1}\;,\label{eq:h-1_newton_approx_del_u}\\
    &\qquad\qquad(\hat{\mu}^h_n, \tilde{\mu}) - \big(\kappa W_{\gamma}''(u_n^h)\delta u_n^h, \tilde{\mu}\big) - (\epsilon^2\grad \delta u_n^h, \grad\tilde{\mu}) = 0&&\forall \tilde{\mu}\in S^{h,1}\;,
\end{align}
\end{subequations}
where $\gamma\in[0,1]$ and the right hand side function $\mu^h$ is given by the following problem:
\begin{equation}\label{eq:h-1_mu_approx}
    \text{Find } \mu^h_n\in S^{h,1} \text{ such that}:\quad (\mu^h_n,\tilde{\mu}) = \big(\kappa W'(u_n^h), \tilde{\mu}\big) + (\epsilon^2\grad u_n^h, \grad \tilde{\mu})\quad\forall \tilde{\mu}\in S^{h,1}\;.
\end{equation}
\indent Now we consider the above problem in the discrete form. Let $\{e_i\}_{i=1}^N$ be the basis functions for $S^{h,1}$; then:
\begin{equation}
    u^h_n = \sum_{i=1}^N (\vect{u}_n)_i e_i\;, \quad\delta u_n^h = \sum_{i=1}^N (\delta\vect{u}_n)_i e_i\;,\quad
    \hat{\mu}^h_n = \sum_{i=1}^N (\hat{\vect{\mu}}_n)_i e_i\;,\quad\mu_n^h = \sum_{i=1}^N (\vect{\mu}_n)_i e_i\;.
\end{equation}
where $\vect{u}_n, \delta \vect{u}_n, \hat{\vect{\mu}}, \vect{\mu}\in\R^N$. We define the matrices $\vect{M},\vect{K}\in\R^{N\times N}$ that correspond to the discretized $L^2$ and $\mathring{H}^1$ inner products in $S^{h,1}$ and the matrix $\vect{D}_n\in\R^{N\times N}$ that corresponds to the nonlinear term in the modified Hessian operator:
\begin{equation}
    \vect{M}_{ij} = (e_i,e_j)\;,\quad \vect{K}_{ij} = (\grad e_i, \grad e_j)\;,\quad(\vect{D}_n)_{ij} = \big((u^h_n)^2 e_i, e_j\big)\;\;.
\end{equation}
Consequently, the discrete Newton step problem is posed as
\begin{equation}\label{eq:h-1_newton_disc}
    \vect{H}_n(\gamma)\begin{bmatrix}
    \delta\vect{u}_n\\
    \vect{\hat{\mu}}_n
    \end{bmatrix} = \begin{bmatrix}
    \sigma \vect{M} & \vect{K}\\
    -\kappa(2+\gamma)\vect{D}_n + \kappa\gamma 
    \vect{M} - \epsilon^2\vect{K} & \vect{M}
    \end{bmatrix}
    \begin{bmatrix}
    \delta\vect{u}_n\\
    \hat{\vect{\mu}}_n
    \end{bmatrix} = 
    \begin{bmatrix}
        \vect{f}_n\\
        \vect{0}
    \end{bmatrix} = \vect{g}_n
\end{equation}
where $\vect{f}_n\coloneqq \vect{f}(\vect{\mu}_n, \vect{u}_n)\in\R^N$ corresponds to the right-hand side of \eqref{eq:h-1_newton_approx_del_u} and depends on $u^h_n$ and $\mu^h_n$. It has the following form:
\begin{equation}
    \vect{f}_n = \vect{K}\vect{\mu}_n + \sigma\vect{M}\vect{u}_n -\sigma m\vect{M}\vect{c}\;.
\end{equation}
where $\vect{c}\in\R^N$ is a vector with entries of all ones and $\mu^h_n$ is given by the following discrete problem:
\begin{equation}\label{eq:mu_disc}
    \vect{M}\vect{\mu}_n = \kappa\vect{d}_n - \kappa \vect{M}\vect{u}_n + \epsilon^2 \vect{K}\vect{u}_n\;,
\end{equation}
where $\vect{d}_n\in\R^N$ corresponds to the term in the double well potential: $(\vect{d}_n)_i = \big((u^h_n)^3, e_i\big)$.

\subsection{The pure Neumann Poisson equation}
As demonstrated in Section \ref{subsec:btls}, solving the Poisson equation with homogeneous Neumann boundary condition is necessary for conducting the Armijo backtracking line search on the energy. There are many different formulations for this problem, as elaborated by Bochev and Lehoucq in \cite{Bochev2005}. Here we employ a saddle point formulation. Consider the following variational problem corresponds to solving $\delta w_n = (-\lap_N)^{-1}R(\delta u_n)$:
\begin{subequations}
\begin{align}
    &\text{Find } (\delta w_n, \lambda)\in H^1\times\R \text{ such that}:\nonumber\\
    &\qquad\qquad\qquad(\grad \delta w_n, \grad \tilde{w}) + (\lambda, \tilde{w}) = (\delta u_n, \tilde{w}) &&\forall \tilde{w}\in H^1\;,\\
    &\qquad\qquad\qquad ( w_n,\tilde{\lambda}) = 0 &&\forall \tilde{\lambda}\in\R\;.
\end{align}
\end{subequations}
This formulation is equivalent to a projected problem, in which the right hand side functional is projected onto $\mathring{H}^{-1}$.\\
\indent The finite element approximation results in the discrete problem:
\begin{equation}\label{eq:poisson_disc}
    \vect{L}
    \begin{bmatrix}
        \delta \vect{w}_n\\
        \lambda
    \end{bmatrix}
    =
    \begin{bmatrix}
        \vect{K} & \vect{M}\vect{c}\\
        c^T\vect{M} & \vect{O}
    \end{bmatrix}
    \begin{bmatrix}
        \delta\vect{w_n}\\
        \lambda
    \end{bmatrix}
    =
    \begin{bmatrix}
        \vect{M}\delta\vect{u}_n\\
        \vect{0}
    \end{bmatrix}
\end{equation}
where $\delta w^h_n = \sum_{i=1}^N (\delta \vect{w}_n)_ie_i\in S^{h,1}$. To solve for $w^h_n = \sum_{i=1}^N (\vect{w}_n)_ie_i\in S^{h,1}$, one
can solve the above system by replacing $\delta\vect{u}_n$ with $\vect{u}_n$.
\subsection{Armijo backtracking line search on the energy}\label{subsubsec:btlc_num}
The backtracking line search on the energy only requires evaluating the energy difference at each step size. The majority of the terms that appear in the energy difference are linear, which allows us to reduce the computational cost of the line search.\\
\indent Assume $\delta u_n^h$ is obtained by the Newton problem at $u_n^h$. Assume $w^h_n$ and $\delta w^h_n$ are obtained by solving the pure Neumann Poisson equation. The energy difference at any backtracking constant $\alpha$ is
\begin{equation}
    F_{OK}(u_n^h+\alpha\delta u_n^h) - F_{OK}(u_n^h) = \delta U_n(\alpha) + \delta U_n(0) + \alpha \delta L_n\;.
\end{equation}
where $\delta U(\alpha)$ and $\delta L_n$ are given by
\begin{equation}
    \delta U(\alpha) \coloneqq \frac{1}{4}\kappa\big(\vect{d}_n(\alpha)\big)^T(\vect{u}_n+\alpha \delta \vect{u}_n)\;,
\end{equation}
\begin{equation}
    \delta L_n \coloneqq \delta \vect{u}_n^{T}\vect{M}(-\kappa\vect{u}_n-\frac{1}{2}\kappa \delta \vect{u}_n
    + \sigma \vect{w}_n + \sigma \delta \vect{w}_n) + \epsilon^2 \delta\vect{u}_n^{T}\vect{K}(2\vect{u}_n + \delta \vect{u}_n) + \sigma \vect{u}_n^T\vect{M}\delta\vect{w}_n\;,
\end{equation}
where $(\vect{d}_n(\alpha))_i = \big((u^h_n+\alpha_k\delta u^h_n)^3, e_i\big)$. The term $\delta U_n$ needs to be evaluated at each step size, while $\delta L_n$ can be determined prior to the backtracking iterations. Once the backtracking converges at $\alpha_K$, $\vect{d}_n(\alpha_K)$ can be reused as $\vect{d}_{n+1}$ at the next minimization iteration.

%% file: appendix/sec_c_preconditioner.tex
\section{A preconditioner for the discrete Newton step problem}\label{appendix:c}
Solving the discrete Newton step problem \eqref{eq:h-1_newton_disc} is the most computationally expensive portion of the algorithm. In particular, the double well backtracking algorithm introduced in Section \ref{subsubsec:bt_gamma} to attain monotonic energy descent further increases the cost. In this section, we consider a preconditioner that can considerably decrease the cost of solving the discrete Newton step problem in a certain range of the weight $\gamma$ and the parameter space $(\kappa,\epsilon,\sigma)$.
\subsection{The approximate Schur complement preconditioner}
We formulate the preconditioner that closely follows a matching technique presented in \cite{Parsons2012,Pearson2018,Farrell2017,Li2018}. The Schur complement of the $(1,1)$-block in $\vect{H}_n(\gamma)$ \eqref{eq:h-1_newton_disc} is
\begin{equation}
    \vect{S}_n = \vect{M} + \frac{\epsilon^2}{\sigma}\vect{K}\vect{M}^{-1}\vect{K} + \frac{\kappa(2+\gamma)}{\sigma}\vect{D}_n\vect{M}^{-1}\vect{K} - \frac{\kappa\gamma}{\sigma}\vect{K}\;.
\end{equation}
Consider an approximation to the Schur complement, $\tilde{\vect{S}} = \hat{\vect{S}}\vect{M}^{-1}\hat{\vect{S}}$, with $\hat{\vect{S}} = \vect{M} + \frac{\epsilon}{\sqrt{\sigma}}\vect{K}$:
\begin{equation}
    \tilde{\vect{S}} = \vect{M} + \frac{\epsilon^2}{\sigma} \vect{K}\vect{M}^{-1} \vect{K} + \frac{2\epsilon}{\sqrt{\sigma}}\vect{K}\;.
\end{equation}
The following proposition holds:
\begin{proposition} \label{prop:precond} \cite{Farrell2017}
    The eigenvalues of $\tilde{\vect{S}}^{-1}\vect{S}_n$ are real and satisfy
\begin{equation}
    \lambda(\tilde{\vect{S}}^{-1}\vect{S}_n)\in \Big[\frac{1}{2}-\frac{\kappa\gamma}{4\epsilon\sqrt{\sigma}}, 1+\frac{\kappa}{4\epsilon\sqrt{\sigma}}\big(\gamma + \kappa(2+\gamma)\lambda_+\big)\Big]\;,
\end{equation}
where $\lambda_+=\max\{\lambda_{\text{max}}(\vect{M}^{-1}\vect{D}_n), 0\}$. If $u^h_n(\vect{x})\in[-1,1]\;\forall \vect{x}\in\Omega$, then  $\lambda_+\leq 1$.
\end{proposition}
Based on the above proposition, we propose a preconditioner $\vect{P}\in\R^{2N\times 2N}$ that only depends on the parameters and discretization, and not on $u_n$:
\begin{equation}
    \vect{P} = 
    \begin{bmatrix}
        \vect{M} & \sigma^{-1}\vect{K}\\
        -\epsilon^2\vect{K} & \vect{M}+ \frac{2\epsilon}{\sqrt{\sigma}}\vect{K}
    \end{bmatrix}\;.
\end{equation}
The preconditioner corresponds to the re-scaled discrete Newton step problem:
\begin{equation}\label{eq:h-1_newton_disc_rescaled}
    \tilde{\vect{H}}_n(\gamma)\begin{bmatrix}
    \delta\vect{u}_n\\
    \vect{\hat{\mu}}_n
    \end{bmatrix} = \begin{bmatrix}
    \vect{M} & \sigma^{-1}\vect{K}\\
    -\kappa(2+\gamma)\vect{D}_n + \kappa\gamma 
    \vect{M} - \epsilon^2\vect{K} & \vect{M}
    \end{bmatrix}
    \begin{bmatrix}
    \delta\vect{u}_n\\
    \hat{\vect{\mu}}_n
    \end{bmatrix} = 
    \begin{bmatrix}
        \sigma^{-1}\vect{f}_n\\
        \vect{0}
    \end{bmatrix}\;.
\end{equation}
This preconditioner leads to the following lemma:
\begin{lemma} \label{lemma:precond}\cite{Li2018}
    The eigenvalues of $\vect{P}^{-1}\tilde{\vect{H}}_n(\gamma)$ are real and their values fall within the range of $\lambda(\tilde{\vect{S}}^{-1}\vect{S}_n)$.
\end{lemma}

\subsection{On the choice of the double well backtracking sequence}
According to Proposition \ref{prop:precond} and Lemma \ref{lemma:precond}, the eigenvalues of the preconditioned discrete Newton step system have a lower bound of $\frac{1}{2}-\frac{\kappa\gamma}{4\epsilon\sqrt{\sigma}}$. Ideally, we would like the eigenvalues to be clustered away from the origin for fast convergence of iterative solvers, such as GMRES, which would require
\begin{equation}
    \gamma < 2\epsilon\sqrt{\sigma}/\kappa
\end{equation}
To reduce the cost of obtaining a descent direction, it is then sensible for the backtracking sequence $\vect{\gamma} = \{\gamma^{(k)}\}_{k=1}^K$ to be
\begin{equation}
    \gamma^{(1)} = 1\;,\quad \gamma^{(2)}<2\epsilon\sqrt{\sigma}/\kappa\;,\quad\dots\;,\quad\gamma^{(K)} = 0\;. 
\end{equation}
\indent Unfortunately, in the case where $\epsilon\sqrt{\sigma}/\kappa > \frac{1}{2}$, the preconditioner is ineffective when $\gamma = 1$. Further work needs to be done on developing a preconditioner suitable for the full range of $\gamma\in[0,1]$.